\documentclass{lmcs} 
\pdfoutput=1

\usepackage{lastpage}
\lmcsdoi{19}{2}{2}
\lmcsheading{}{\pageref{LastPage}}{}{}%
{Sep.~23,~2021}{Apr.~14,~2023}{}

\usepackage{color} 
\usepackage{url}
\usepackage{xcolor}
\usepackage{latexsym}
\usepackage{bm}
\usepackage{graphicx}
\usepackage{relsize}
\usepackage{amsmath,amssymb,amsthm}
\usepackage[utf8]{inputenc}
\usepackage{ragged2e}
\usepackage{float}
\usepackage{pbox}
\usepackage{tabularx}
\usepackage{tcolorbox}
\usepackage{listings}
\usepackage{moreverb}
\usepackage{extarrows}
\usepackage{tikz}
\usepackage{qtree}
\usepackage{forest}
\usepackage{cite}
\usepackage{hyperref}
\usepackage{xspace}
\usepackage{cleveref}

\newcommand{\N}{\mathbb{N}} 

\newcommand{\ante}{\text{ante}}
\newcommand{\red}{\text{red}}
\newcommand{\lv}{\text{lv}}

\newcommand{\resop}[1]{\stackrel{#1}{\bowtie}}

\newcommand{\var}{\text{var}}

\usetikzlibrary{arrows}
\usetikzlibrary{shapes}

\newcommand*\circled[1]{\tikz[baseline=(char.base)]{
		\node[shape=circle,draw,inner sep=2pt] (char) {#1};}}

\def\hlinewd#1{%
	\noalign{\ifnum0=`}\fi\hrule \@height #1 %
	\futurelet\reserved@a\@xhline} 
\newcommand{\entailsSigma}{\vDash_{\Sigma_1^\exists}}

\newcommand{\hsc}[1]{{\footnotesize\MakeUppercase{#1}}}
\newcommand{\lo}{\text{\textsf{L\hsc{ev}-O\hsc{rd}}}}
\newcommand{\ao}{\text{\textsf{A\hsc{ny}}}}
\newcommand{\ar}{\text{\textsf{R\hsc{ed}}}}
\newcommand{\nr}{\text{\textsf{N\hsc{o}-R\hsc{ed}}}}
\newcommand{\aso}{\text{\textsf{E\hsc{xi}-A\hsc{ny}-N\hsc{o}-R\hsc{ed}}}}
\newcommand{\asro}{\text{\textsf{U\hsc{ni}-A\hsc{ny}}}}
\newcommand{\loar}{\mathsf{QCDCL}^\textsf{{\scriptsize L\tiny{\MakeUppercase{ev}}\scriptsize-O\tiny{\MakeUppercase{rd}}}}_\textsf{{\scriptsize R\tiny{\MakeUppercase{ed}}}}}

\newcommand{\aoar}{\mathsf{QCDCL}^\textsf{{\scriptsize A\tiny{\MakeUppercase{ny}}}}_\textsf{{\scriptsize R\tiny{\MakeUppercase{ed}}}}}

\newcommand{\lonr}{\mathsf{QCDCL}^\textsf{{\scriptsize L\tiny{\MakeUppercase{ev}}\scriptsize-O\tiny{\MakeUppercase{rd}}}}_\textsf{{\scriptsize N\tiny{\MakeUppercase{o}}}\scriptsize-{\scriptsize R\tiny{\MakeUppercase{ed}}}}}

\newcommand{\aonr}{\mathsf{QCDCL}^\textsf{{\scriptsize A\tiny{\MakeUppercase{ny}}}}_\textsf{{\scriptsize N\tiny{\MakeUppercase{o}}}\scriptsize-{\scriptsize R\tiny{\MakeUppercase{ed}}}}}

\newcommand{\asoar}{\mathsf{QCDCL}^\textsf{{\scriptsize E\tiny{\MakeUppercase{xi}}\scriptsize-A\tiny{\MakeUppercase{ny}}\scriptsize-N\tiny{\MakeUppercase{o}}\scriptsize-R\tiny{\MakeUppercase{ed}}}}_\textsf{{\scriptsize R\tiny{\MakeUppercase{ed}}}}}

\newcommand{\asonr}{\mathsf{QCDCL}^\textsf{{\scriptsize E\tiny{\MakeUppercase{xi}}\scriptsize-A\tiny{\MakeUppercase{ny}}\scriptsize-N\tiny{\MakeUppercase{o}}\scriptsize-R\tiny{\MakeUppercase{ed}}}}_\textsf{{\scriptsize N\tiny{\MakeUppercase{o}}}\scriptsize-{\scriptsize R\tiny{\MakeUppercase{ed}}}}}

\newcommand{\asroar}{\mathsf{QCDCL}^\textsf{{\scriptsize U\tiny{\MakeUppercase{ni}}\scriptsize-A\tiny{\MakeUppercase{ny}}}}_\textsf{{\scriptsize R\tiny{\MakeUppercase{ed}}}}}

\newcommand{\asronr}{\mathsf{QCDCL}^\textsf{{\scriptsize U\tiny{\MakeUppercase{ni}}\scriptsize-A\tiny{\MakeUppercase{ny}}}}_\textsf{{\scriptsize N\tiny{\MakeUppercase{o}}}\scriptsize-{\scriptsize R\tiny{\MakeUppercase{ed}}}}}

\newcommand{\proofsystem}[1]{{\sf{#1}}\xspace}

\newcommand{\qcdcl}{\proofsystem{QCDCL}}
\newcommand{\qres}{\proofsystem{Q-resolution}}
\newcommand{\cdcl}{\proofsystem{CDCL}}
\newcommand{\ldqres}{\proofsystem{long-distance Q-resolution}}
\newcommand{\Ldqres}{\proofsystem{Long-distance Q-resolution}}
\newcommand{\ldqcdcl}{\proofsystem{long-distance QCDCL resolution}}
\newcommand{\resolution}{\proofsystem{resolution}}

\newcommand{\Lon}{\ensuremath{\mathtt{Lon}}}

\newcommand{\qnpres}{\ensuremath{\mathsf{Q}^{\mathsf{NP}}\text{\sf{-resolution}}}\xspace}	   

\newcommand{\citespace}{\hspace*{-0.26em}}

\definecolor{hellgrau}{RGB}{255,255,255}

\lstdefinestyle{maple}
{morecomment=[l]{##},morekeywords={if,while,for,from,to,do,then,end,proc},sensitive=false}

\keywords{CDCL,QBF,QCDCL,proof complexity,resolution,Q-resolution}

\usepackage{hyperref}
\theoremstyle{plain} 

\begin{document}
	
	\title[THE STRENGTH OF QBF CDCL SOLVERS AND QBF RESOLUTION]{Understanding the Relative Strength of\texorpdfstring{\\}{ }QBF CDCL Solvers and QBF Resolution}
	
	\author[O.~Beyersdorff]{Olaf Beyersdorff}	
	\author[B.~Böhm]{Benjamin Böhm}	
	\address{Friedrich Schiller University Jena}	
	\email{olaf.beyersdorff@uni-jena.de, benjamin.boehm@uni-jena.de}  
	
	
	
	
	
	
	
	\begin{abstract}
		\noindent QBF solvers implementing the QCDCL paradigm are powerful algorithms that successfully tackle many computationally complex applications. However, our theoretical understanding of the strength and limitations of these QCDCL solvers is very limited.
		
		In this paper we suggest to formally model QCDCL solvers as proof systems. We define different policies that can be used for decision heuristics and unit propagation and give rise to a number of sound and complete QBF proof systems (and hence new QCDCL algorithms). With respect to the standard policies used in practical QCDCL solving, we show that the corresponding QCDCL proof system is incomparable (via exponential separations) to Q-resolution, the classical QBF resolution system used in the literature. This is in stark contrast to the propositional setting where CDCL and resolution are known to be p-equivalent.
		
		This raises the question what formulas are hard for standard QCDCL, since Q-resolution lower bounds do not necessarily apply to QCDCL as we show here. In answer to this question we prove several lower bounds for QCDCL, including exponential lower bounds for a large class of random QBFs. 
		
		We also introduce a strengthening of the decision heuristic used in classical QCDCL, which does not necessarily decide variables in order of the prefix, but still allows to learn asserting clauses. We show that with this decision policy, QCDCL can be exponentially faster on some formulas.
		
		We further exhibit a QCDCL proof system that is p-equivalent to Q-resolution. In comparison to classical QCDCL, this new QCDCL version adapts both decision and unit propagation policies. 
	\end{abstract}
	
	\maketitle
	
	\section{Introduction}\label{S:one}

	SAT solving has revolutionised the way we perceive and approach computationally complex problems. While traditionally, NP-hard problems were considered computationally intractable, today SAT solvers routinely and successfully solve instances of NP-hard problems from virtually all application domains, and in particular problem instances of industrial relevance \cite{Var14}. Starting with the classic DPLL algorithm from the 1960s \cite{DP60,DLL62}, there have been a number of milestones in the evolution of SAT solving, but clearly one of the breakthrough achievements was the introduction of clause learning in the late 1990s, leading to the paradigm of \emph{conflict-driven clause learning} (CDCL) \cite{DBLP:conf/iccad/SilvaS96,ZhangMMM01}, the predominant technique of modern SAT solving.  CDCL ingeniously combines a number of crucial ingredients, among them variable decision heuristics, unit propagation, clause learning from conflicts, and restarts (cf.\ \cite{DBLP:series/faia/SilvaLM09} for an overview).
	
	Inspired by the success of SAT solving, many researchers have concentrated on the task to extend the reach of these technologies to computationally even more challenging settings with \emph{quantified Boolean formulas} (QBF) receiving key attention. As a PSPACE-complete problem, the satisfiability problem for QBFs encompasses all problems from the polynomial hierarchy and allows to encode many problems far more succinctly than in propositional logic (cf.\ \cite{ShuklaBPS19} for applications). 
	
	One of the main techniques in QBF solving is the propositional CDCL technique, lifted to QBF in the form of QCDCL \cite{ZM02}. However, solving QBFs presents additional challenges as the quantifier type of variables (existential and universal) needs to be taken into account as well as the variable dependencies stemming from the quantifier prefix.\footnote{In this paper we focus on prenex QBFs with a CNF matrix.} This particularly impacts the variable selection heuristics and details of the unit propagation within QCDCL. 
	In addition to QCDCL there are further QBF solving techniques, exploiting QBF features absent in SAT, such as expanding universal variables in expansion solving \cite{JM15} and dependency schemes in dependency-aware solving \cite{LonsingE17,SlivovskyS16,PeitlSS19}. Compared to SAT solving, QBF solving is still at an earlier stage. However,  QBF solving has seen huge improvements during the past 15 years \cite{PulinaS19}, and there are problems of practical relevance where QBF solvers  outperform SAT  solvers \cite{FaymonvilleFRT17}.
	
	The enormous success of SAT and QBF solving of course raises theoretical questions of utmost importance: why are these solvers so successful and what are their limitations? The main approach through understanding these questions comes from proof complexity \cite{Bus12,Nordstrom15}. 
	The central problem in proof complexity is to determine the size of
	the smallest proof for a given formula in a specified proof system, typically defined through a set of
	axioms and inference rules. Traces of runs of SAT/QBF solvers on unsatisfiable instances yield proofs
	of unsatisfiability, whereby each solver implicitly defines a proof system. In particular, SAT solvers implementing the DPLL and CDCL paradigms are based on resolution \cite{Nordstrom15}, which is arguably the most studied
	proof system in proof complexity. 
	
	\emph{Propositional resolution} operates on clauses and uses the resolution rule
	\begin{equation} \label{eq:res-rule}
	\frac{C\vee x \qquad D \vee \bar x}{C\vee D}
	\end{equation}
	as its only inference rule to derive a new clause $C\vee D$ from the two parent clauses $C\vee x$ and $D \vee \bar x$.\footnote{We denote such a resolution inference with pivot $x$ by $(C\vee x) \resop{x} (D\vee \bar x)$ throughout the paper.}
	There is a host of lower bounds and lower bound techniques available for propositional resolution (cf.\ \cite{BP01,Seg07,Kra19} for surveys).
	
	While it is relatively easy to see that the classic DPLL branching algorithm \cite{DP60,DLL62} exactly corresponds to tree-like resolution (where resolution derivations are in form of a tree), the \emph{relation between CDCL and resolution} is far more complex. On the one hand, resolution proofs can be generated efficiently from traces of CDCL runs on unsatisfiable formulas \cite{DBLP:journals/jair/BeameKS04}, a crucial observation being that learned clauses are derivable by resolution \cite{DBLP:journals/jair/BeameKS04,DBLP:conf/iccad/SilvaS96}. The opposite simulation is considerably more difficult, with a series of works \cite{DBLP:journals/jair/BeameKS04,HertelBPG08,DBLP:journals/ai/PipatsrisawatD11,AtseriasFT11} culminating in the result that CDCL can efficiently simulate arbitrary resolution proofs, i.e., resolution and CDCL are equivalent. This directly implies that all known lower bounds for proof size in resolution translate into lower bounds for CDCL running time. In addition, other measures such as proof space model memory requirements of SAT solvers, thereby implying lower bounds on memory consumption, in particular when considering time-space tradeoffs \cite{Nordstrom08}. 
	
	Exciting as this equivalence between CDCL and resolution is from a theoretical point of view, it has to be interpreted with care. Proof systems are inherently non-deterministic procedures, while CDCL algorithms are largely deterministic (some randomisation might occasionally be used). To overcome this discrepancy, the simulations of resolution by CDCL \cite{DBLP:journals/jair/BeameKS04,DBLP:journals/ai/PipatsrisawatD11} use arbitrary decision heuristics and perform excessive restarts, both of which diverge from practical CDCL policies. Indeed, in very recent work \cite{Vin20} it was shown that CDCL with practical decision heuristics such as VSIDS \cite{ZhangMMM01} is exponentially weaker than resolution, and similar results have been obtained for further decision heuristics \cite{MullPR19}. Regarding restarts there is intense research aiming to determine the power of CDCL without restarts from a proof complexity perspective (cf.~\cite{BussHJ08,BonetBJ14}).
	
	On the QBF level, this naturally raises the question \emph{what proof system corresponds to QCDCL}. As in propositional proof complexity, QBF resolution systems take a prominent place in the QBF proof system landscape, with the basic and historically first \qres system \cite{DBLP:journals/iandc/BuningKF95} receiving key attention. \qres is a refutational system that proves the falsity of fully quantified prenex QBFs with a CNF matrix (QCNFs). The system allows to use the propositional resolution rule~\eqref{eq:res-rule} under the conditions that the pivot $x$ is an existential variable and the resolvent $C\vee D$ is non-tautological. In addition, \qres uses a \emph{universal reduction rule}
	\begin{equation} \label{eq:red-rule}
	\frac{C\vee u}{C}\enspace,
	\end{equation}
	where $u$ is a universal literal that in the quantifier prefix is quantified right of all variables in~$C$, i.e., none of the literals in~$C$ depends on $u$. For \qres we have a number of lower bounds \cite{BWJ14,BeyersdorffCJ19,BBH19} as well as lower bound techniques, some of them lifted from propositional proof complexity \cite{BCMS17,BCS19}, but more interestingly some of them genuine to the QBF domain \cite{BBH19,BBM20} that unveil deep connections between proof size and circuit complexity \cite{BBCP20}, unparalleled in the propositional domain.
	
	Unlike in the relation between SAT and CDCL, it is has been open whether QCDCL runs can be efficiently translated into \qres. Instead, QCDCL runs can be simulated by the stronger QBF resolution system of \emph{\ldqres} \cite{ZM02,Balabanov12}. In fact, this system originates from solving, where it was noted that clauses learned from QCDCL conflicts can be derived in \ldqres \cite{ZM02}. \Ldqres implements a more liberal use of the resolution rule~\eqref{eq:res-rule}, which allows to derive certain tautologies (cf.\ Section~\ref{subsec:qres-ldqres} for details). In general, allowing to derive tautologies with \eqref{eq:res-rule} is unsound, an example is given in Section~\ref{subsec:qres-ldqres}. However, the tautologies allowed in \ldqres do not present problems for soundness and are exactly those clauses needed when learning  clauses in QCDCL. Hence \ldqres simulates QCDCL \cite{ZM02,Balabanov12}. However, it is known that \ldqres allows exponentially shorter proofs than \qres for some QBFs \cite{ELW13,BBH19,BeyersdorffBM19}.
	
	We also remark that there are further QBF resolution systems (cf.\ \cite{BeyersdorffCJ19} for an overview), some of them corresponding to other solving approaches in QBF, such as the system $\forall$\textsf{Exp+Res} that captures expansion QBF solving \cite{JM15}.
	
	In summary, it is fair to say that the relations between QCDCL solving and QBF resolution (either \qres or \ldqres) are \emph{currently not well understood.} In particular, an analogue of the equivalence of CDCL SAT solving and propositional resolution \cite{DBLP:journals/jair/BeameKS04,DBLP:journals/ai/PipatsrisawatD11,AtseriasFT11} is currently absent in the QBF domain. This brings us to the topic of this paper.

	\subsection{Our contributions}
	
	We state and explain our main contributions and provide pointers to where these are proven in the main part.
	
	\subsubsection{QCDCL and Q-resolution are incomparable}
	\label{sec:qcdcl-qres-intro}
	
	Our first contribution establishes that QCDCL and \qres are incomparable by exponential separations (Thm~\ref{TheoremIncomparable}), i.e., there exist QBFs that are easy for QCDCL, but require exponential-size \qres refutations, and vice versa. As explained above, this is in stark contrast to the propositional setting, where CDCL and resolution are equivalent.
	
	Proving the incomparability requires two families of QBFs. For the first we take the parity formulas which are known to require exponential-size \qres refutations \cite{BeyersdorffCJ19}. Here we show that $\mathtt{QParity}_n$ is easy for QCDCL.
	
	This requires to formally state QCDCL in terms of a proof system (we will denote this by \qcdcl and explain it in Sections~\ref{subsec:cdcl-ps-intro} and \ref{sec:qcdcl-ps}) and to construct specific trails and clauses learned from these trails that together comprise a short \qcdcl proof of the formulas.
	
	For the opposite separation we construct formulas $\mathtt{Trapdoor}_n$ and show that they require exponential-size \qcdcl refutations (Proposition~\ref{PropTrapdoorHardForQCDCL}). Hardness of these formulas for \qcdcl rests on the fact that 
in \qcdcl, variables have to be decided in order of the quantifier prefix. 
	On the other hand, it is easy to obtain short \qres refutations of $\mathtt{Trapdoor}_n$ (Proposition~\ref{prop:trapdoor-easy}).
	
	This establishes the separation of \qcdcl and \qres. We remark that in earlier work, Janota \cite{DBLP:conf/sat/Janota16} showed that QCDCL with a specific asserting learning scheme requires large running time on some class of QBFs, whereas the same formulas are easy for \qres. Of course, this raises the question whether another learning scheme might produce short QCDCL runs. In contrast, our  result rules out any simulation of \qres by \qcdcl (or vice versa), regardless of the learning scheme used.

	\subsubsection{Lower bounds for QCDCL}
	\label{subsec:lb-qcdcl-intro}

	The incomparability of \qres and \qcdcl raises the immediate question of what formulas are hard for QCDCL. Previous research has largely concentrated on showing lower bounds for \qres (e.g.\ \cite{DBLP:journals/iandc/BuningKF95,BeyersdorffCJ19,BBH19}). However, by our results from the last subsection, these lower bounds do not necessarily apply to QCDCL, and prior to this paper no dedicated lower bounds for QCDCL (with arbitrary learning schemes) were known.
	
	Here we show that several formulas from the QBF literature, including the equality formulas and a large class of random QBFs \cite{BBH19} are indeed hard for \qcdcl. 
			Both the equality and the random formulas are of the type $\Sigma_3^b$, i.e., they have two quantifier alternations starting with $\exists$.  Also, both require exponential-size proofs in \qres (the random formulas whp) \cite{BBH19}. This is shown in \cite{BBH19} via the size-cost-capacity technique, a semantically grounded QBF lower-bound technique that infers \qres hardness for formulas $\Phi_n$ (and in fact hardness for even stronger systems) from lower bounds for the size of countermodels for $\Phi_n$. 
	
	It is not clear how to directly apply this technique to \qcdcl. Instead, we identify a property, which we term the \emph{$XT$-property} (Definition~\ref{def:XT-property}), that we can use to lift hardness from \qres to \qcdcl. Intuitively, it says that in a $\Sigma_3^b$ formula $\Phi$ with quantifier prefix of the form $\exists X \forall U \exists T$ with blocks of variables $X$, $U$, $T$, there is no direct connection between the $X$ and $T$ variables, i.e., $\Phi$ does not contain clauses with $X$ and $T$ variables, but no $U$ variables (there are some further condition on clauses containing only $T$ variables (Definition~\ref{def:XT-property})).
	
	We can then prove that QCDCL runs on formulas with this $XT$-property can be efficiently transformed into \qres refutations, not only into \ldqres refutations. Thus for formulas with the $XT$-property we can lift the \qres lower bounds to  \qcdcl (Thm~\ref{TheoremXTpropertyRequiresLongDistanceProofsSizeS})
		It is quite easy to check that both the equality formulas as well as the random formulas above have the $XT$-property.  Thus both are exponentially hard for \qcdcl.
	
	Our findings so far reveal an interesting picture on \qcdcl hardness. Firstly, \emph{not all \qres hardness results lift to \qcdcl}: the lower bounds for equality and random formulas shown via size-cost-capacity \cite{BBH19} \emph{do}, but the lower bounds for parity shown via circuit complexity \cite{BeyersdorffCJ19} \emph{do not}. 
	
	Secondly, it is worth to compare the \qcdcl hardness results for $\mathtt{Trapdoor}$ from the previous subsection to the \qcdcl hardness results shown here for equality and random formulas. The hardness of $\mathtt{Trapdoor}$ lifts from propositional hardness for PHP, while the hardness of equality and random formulas lifts from \qres hardness. In fact, this can be made formal by using a model of QBF proof systems with access to an NP oracle \cite{BHP20}, which allows to collapse propositional subderivations of arbitrary size into just one oracle inference step. Hardness under the NP-oracle version of \qres guarantees that the hardness is `genuine' to QBF and not lifted from propositional resolution. We show here that this notion of `genuine' QBF hardness, tailored towards \qcdcl, also holds for the \qcdcl lower bounds for equality and the random QBFs (Proposition~\ref{prop:equality-random-oracle-hardness}). 
	
	On the other hand, the parity formulas also exhibit `genuine' QBF hardness, as they are hard in the NP-oracle version of \qres \cite{BBM20}. Since they are easy for \qcdcl, this means that not all genuine \qres lower bounds lift to \qcdcl.
	
	Thirdly, hardness for \qcdcl can of course also stem from hardness for \ldqres, since the latter system simulates the former.\footnote{A proof system $P$ p-simulates a proof system $S$ if each $S$ proof can be efficiently transformed into a $P$ proof of the same formula \cite{CR79}. If the systems p-simulate each other, they are p-equivalent.} However, there are only very few hardness results for \ldqres known in the literature \cite{BWJ14,BeyersdorffCJ19}, hence our hardness results shown here should be also valuable for practitioners, in particular the hardness results for the large class of random QCNFs. It is also worth noting that the equality formulas are easy for \ldqres \cite{BeyersdorffBM19}, hence our results imply an exponential separation between \qcdcl and \ldqres (Corollary~\ref{cor:qcdcl-ldqres}).
	
	\subsubsection{Our framework: QCDCL as formal proof systems}
	\label{subsec:cdcl-ps-intro}
	
	Technically, this paper hinges on the formalisation of QCDCL solving as precisely defined proof systems, which can subsequently be analysed from a proof-complexity perspective. This involves formalising a number of QCDCL ingredients (cf.\ Section~\ref{subsec:qcdcl} for an informal account on how QCDCL works). 
	
	A \emph{QCDCL trail} $\mathcal{T}$ for a QCNF $\Phi$ is a sequence of literals, which we typically denote as
	\begin{align*}
	\mathcal{T}=(p_{(0,1)},\ldots, p_{(0,g_0)};\mathbf{d_{1}},p_{(1,1)},\ldots,p_{(1,g_{1})};\ldots; \mathbf{d_r},p_{(r,1)},\ldots ,p_{(r,g_r)}  )\text{.}
	\end{align*}
	Here $d_1,\dots,d_r$ are the decision literals and the $p$ literals are propagated by unit propagation. While decisions can be either existential or universal, propagated literals are always existential. In classical QCDCL, the following \emph{decision policy} is adopted: 
	\begin{itemize}
		\item $\lo$\textbf{ - }For each $d_i\in \mathcal{T}$, all variables from quantifier blocks left of $d_i$ in the prefix of $\Phi$ appear left of $d_i$ in $\mathcal{T}$ as positive or negative literals, i.e., all variables on which $d_i$ depends have been decided or propagated (as literals) before $d_i$ is decided.		
	\end{itemize}
	This decision policy therefore follows the order of quantification in the prefix, for which reason we call it level ordered ($\lo$).
	
	In addition to $\lo$, we consider three more decision policies. The first one stems from propositional CDCL where the order of decisions is completely arbitrary:
	\begin{itemize}
		\item $\ao$\textbf{ - }Given a trail $\mathcal{T}$, we can choose any remaining literal as the next decision.
	\end{itemize}
	Before defining the remaining two decision policies, we explain our policies for unit propagation. The first comes again just from propositional CDCL:
	\begin{itemize}
		\item $\nr$\textbf{ - }For each propagated literal $p_{(i,j)}\in \mathcal{T}$ there has to be a clause $C$ in $\Phi$ such that $C$ becomes a single-literal clause under the sub-trail $\mathcal{T}[i,j-1]$ of $\mathcal{T}$ that contains all decisions and propagations in $\mathcal{T}$ before $p_{(i,j)}$.
	\end{itemize}
	To illustrate this with a small example, assume that $\Phi$ contains a clause $C=x\vee\bar y \vee z$ and $\mathcal{T}$ contains the decisions $\bar x, y$. Then $C$ is simplified to the single literal $z$ under the assignment $\mathcal{T}$, and hence $z$ is propagated and included as the next variable in $\mathcal{T}$. 
	
	This is just CDCL propagation. It is, however, not what is done in QCDCL. Assume again we have a clause $C=x\vee\bar y \vee z\vee u$ in $\Phi$ and $\exists x, y, z \forall u$ appears in the prefix of $\Phi$. If the trail contains $\bar x,y$, we cannot propagate $z$ with the policy $\nr$. However, we can use universal reduction on $u$ as in rule~\eqref{eq:red-rule} of \qres to reduce the clause $z\vee u$ (the clause $C$ under the assignment corresponding to $\mathcal{T}$) to the single-literal clause $z$. Hence we can immediately propagate $z$ with the following unit propagation policy:
	\begin{itemize}
		\item $\ar$\textbf{ - }For each propagated literal $p_{(i,j)}\in \mathcal{T}$ there is a clause $C$ in $\Phi$ such that $C$ becomes a single-literal clause under the trail $\mathcal{T}[i,j-1]$ using universal reduction.
	\end{itemize}
	
	In (Q)CDCL, whenever a trail $\mathcal{T}$ runs into a conflict, i.e., a clause $C$ from $\Phi$ is falsified, we perform conflict analysis in the form of \emph{clause learning}. This results in a clause $D$ that follows from $\Phi$ and describes a reason for the conflict. Such conflict clauses are obtained by performing resolution (for CDCL) and \ldqres (for QCDCL), starting from the conflict clause $C$ and resolving along the propagated variables in $\mathcal{T}$ in reverse order (skipping resolution steps when the pivot is missing). 
	
	We prove that this learning process works independently from our policies, e.g., even when $\ao$ or $\nr$ is used, we can correctly perform \ldqres for clause learning as in QCDCL (Proposition~\ref{PropSoundnessLongDistance}). For practical (Q)CDCL, it is important that we do not just learn any clause, but an \emph{asserting clause} $D$, which means that $D$ becomes unit after backtracking.
	
	We notice that though the policy $\ao$ is sound, it does not always allow to learn asserting clauses (Remark~\ref{rem:no-asserting}). Therefore, we introduce further policies, which are intermediate between $\lo$ and $\ao$ and still guarantee that asserting clauses can be learned (Lemmas~\ref{LemmaAssertingClauseASONR} and \ref{LemmaAssertingClauseASROAR}). 
	
	We define two policies $\aso$ and $\asro$, to be used with the unit propagation policies $\nr$ and $\ar$, respectively.
	\begin{itemize}
		\item 
		$\aso$\textbf{ - }We can decide a literal $d_k$ if it is existential, or if it is universal and it holds $\lv(d_1)\leq \ldots \leq \lv(d_k)$.
		\footnote{Under a prefix $Q_1 X_1 Q_2 X_2 \dots Q_s X_s$ 	with disjoint blocks of variables $X_i$ and alternating blocks of quantifiers~$Q_i\in\{\exists,\forall\}$, the \emph{quantifier level} of a variable $x$ is $\lv(x)=i$, if $x\in X_i$.}
		\item $\asro$\textbf{ - }We can only decide an existential variable $x$ next, if and only if we already decided all universal variables $u$ with $\lv(u)<\lv(x)$ before in $\mathcal{T}$. 
	\end{itemize}
	
	\begin{figure}
		\centering
		
		
		\tikzset {_j8eqf49a1/.code = {\pgfsetadditionalshadetransform{ \pgftransformshift{\pgfpoint{0 bp } { 0 bp }  }  \pgftransformrotate{-270 }  \pgftransformscale{2 }  }}}
		\pgfdeclarehorizontalshading{_pogy62zca}{150bp}{rgb(0bp)=(1,1,1);
			rgb(37.5bp)=(1,1,1);
			rgb(50bp)=(0.95,0.95,0.95);
			rgb(50.25bp)=(0.88,0.88,0.88);
			rgb(62.5bp)=(0.96,0.96,0.96);
			rgb(100bp)=(0.96,0.96,0.96)}
		
		
		\tikzset {_8tapmc2lb/.code = {\pgfsetadditionalshadetransform{ \pgftransformshift{\pgfpoint{0 bp } { 0 bp }  }  \pgftransformrotate{-270 }  \pgftransformscale{2 }  }}}
		\pgfdeclarehorizontalshading{_x2z759u63}{150bp}{rgb(0bp)=(1,1,1);
			rgb(37.5bp)=(1,1,1);
			rgb(50bp)=(0.95,0.95,0.95);
			rgb(50.25bp)=(0.88,0.88,0.88);
			rgb(62.5bp)=(0.96,0.96,0.96);
			rgb(100bp)=(0.96,0.96,0.96)}
		
		
		\tikzset {_glmalahg6/.code = {\pgfsetadditionalshadetransform{ \pgftransformshift{\pgfpoint{0 bp } { 0 bp }  }  \pgftransformrotate{-270 }  \pgftransformscale{2 }  }}}
		\pgfdeclarehorizontalshading{_ove9yi8c1}{150bp}{rgb(0bp)=(1,1,1);
			rgb(37.5bp)=(1,1,1);
			rgb(50bp)=(0.95,0.95,0.95);
			rgb(50.25bp)=(0.88,0.88,0.88);
			rgb(62.5bp)=(0.96,0.96,0.96);
			rgb(100bp)=(0.96,0.96,0.96)}
		
		
		\tikzset {_7lht7677c/.code = {\pgfsetadditionalshadetransform{ \pgftransformshift{\pgfpoint{0 bp } { 0 bp }  }  \pgftransformrotate{-270 }  \pgftransformscale{2 }  }}}
		\pgfdeclarehorizontalshading{_26jr9qcf1}{150bp}{rgb(0bp)=(1,1,1);
			rgb(37.5bp)=(1,1,1);
			rgb(50bp)=(0.95,0.95,0.95);
			rgb(50.25bp)=(0.88,0.88,0.88);
			rgb(62.5bp)=(0.96,0.96,0.96);
			rgb(100bp)=(0.96,0.96,0.96)}
		
		
		\tikzset {_bbj93ym0r/.code = {\pgfsetadditionalshadetransform{ \pgftransformshift{\pgfpoint{0 bp } { 0 bp }  }  \pgftransformrotate{-270 }  \pgftransformscale{2 }  }}}
		\pgfdeclarehorizontalshading{_k2der5oz4}{150bp}{rgb(0bp)=(1,1,1);
			rgb(37.5bp)=(1,1,1);
			rgb(50bp)=(0.95,0.95,0.95);
			rgb(50.25bp)=(0.88,0.88,0.88);
			rgb(62.5bp)=(0.96,0.96,0.96);
			rgb(100bp)=(0.96,0.96,0.96)}
		
		
		\tikzset {_zb9z0e029/.code = {\pgfsetadditionalshadetransform{ \pgftransformshift{\pgfpoint{0 bp } { 0 bp }  }  \pgftransformrotate{-270 }  \pgftransformscale{2 }  }}}
		\pgfdeclarehorizontalshading{_0pwh4sk7t}{150bp}{rgb(0bp)=(1,1,1);
			rgb(37.5bp)=(1,1,1);
			rgb(50bp)=(0.95,0.95,0.95);
			rgb(50.25bp)=(0.88,0.88,0.88);
			rgb(62.5bp)=(0.96,0.96,0.96);
			rgb(100bp)=(0.96,0.96,0.96)}
		
		
		\tikzset {_b74qmq0s2/.code = {\pgfsetadditionalshadetransform{ \pgftransformshift{\pgfpoint{0 bp } { 0 bp }  }  \pgftransformrotate{-270 }  \pgftransformscale{2 }  }}}
		\pgfdeclarehorizontalshading{_xsbqfayhq}{150bp}{rgb(0bp)=(1,1,1);
			rgb(37.5bp)=(1,1,1);
			rgb(50bp)=(0.95,0.95,0.95);
			rgb(50.25bp)=(0.88,0.88,0.88);
			rgb(62.5bp)=(0.96,0.96,0.96);
			rgb(100bp)=(0.96,0.96,0.96)}
		
		
		\tikzset {_hfq0n7d22/.code = {\pgfsetadditionalshadetransform{ \pgftransformshift{\pgfpoint{0 bp } { 0 bp }  }  \pgftransformrotate{-270 }  \pgftransformscale{2 }  }}}
		\pgfdeclarehorizontalshading{_s76zcr58y}{150bp}{rgb(0bp)=(1,1,1);
			rgb(37.5bp)=(1,1,1);
			rgb(50bp)=(0.95,0.95,0.95);
			rgb(50.25bp)=(0.88,0.88,0.88);
			rgb(62.5bp)=(0.96,0.96,0.96);
			rgb(100bp)=(0.96,0.96,0.96)}
		\tikzset{every picture/.style={line width=0.75pt}} 
		
		
		\tikzset {_mckjdonky/.code = {\pgfsetadditionalshadetransform{ \pgftransformshift{\pgfpoint{0 bp } { 0 bp }  }  \pgftransformrotate{-270 }  \pgftransformscale{2 }  }}}
		\pgfdeclarehorizontalshading{_pvcfv9qai}{150bp}{rgb(0bp)=(1,1,1);
			rgb(37.5bp)=(1,1,1);
			rgb(50bp)=(0.95,0.95,0.95);
			rgb(50.25bp)=(0.88,0.88,0.88);
			rgb(62.5bp)=(0.96,0.96,0.96);
			rgb(100bp)=(0.96,0.96,0.96)}
		\tikzset{every picture/.style={line width=0.75pt}} 
		
		\begin{tikzpicture}[x=0.75pt,y=0.75pt,yscale=-0.75,xscale=1]
		

		\path  [shading=_pogy62zca,_j8eqf49a1] (130,348) .. controls (130,343.58) and (133.58,340) .. (138,340) -- (293,340) .. controls (297.42,340) and (301,343.58) .. (301,348) -- (301,372) .. controls (301,376.42) and (297.42,380) .. (293,380) -- (138,380) .. controls (133.58,380) and (130,376.42) .. (130,372) -- cycle ; 
		\draw   (130,348) .. controls (130,343.58) and (133.58,340) .. (138,340) -- (293,340) .. controls (297.42,340) and (301,343.58) .. (301,348) -- (301,372) .. controls (301,376.42) and (297.42,380) .. (293,380) -- (138,380) .. controls (133.58,380) and (130,376.42) .. (130,372) -- cycle ; 
		
		\path  [shading=_x2z759u63,_8tapmc2lb] (359,348) .. controls (359,343.58) and (362.58,340) .. (367,340) -- (522,340) .. controls (526.42,340) and (530,343.58) .. (530,348) -- (530,372) .. controls (530,376.42) and (526.42,380) .. (522,380) -- (367,380) .. controls (362.58,380) and (359,376.42) .. (359,372) -- cycle ; 
		\draw   (359,348) .. controls (359,343.58) and (362.58,340) .. (367,340) -- (522,340) .. controls (526.42,340) and (530,343.58) .. (530,348) -- (530,372) .. controls (530,376.42) and (526.42,380) .. (522,380) -- (367,380) .. controls (362.58,380) and (359,376.42) .. (359,372) -- cycle ; 
		
		\path  [shading=_ove9yi8c1,_glmalahg6] (130,278) .. controls (130,273.58) and (133.58,270) .. (138,270) -- (293,270) .. controls (297.42,270) and (301,273.58) .. (301,278) -- (301,302) .. controls (301,306.42) and (297.42,310) .. (293,310) -- (138,310) .. controls (133.58,310) and (130,306.42) .. (130,302) -- cycle ; 
		\draw   (130,278) .. controls (130,273.58) and (133.58,270) .. (138,270) -- (293,270) .. controls (297.42,270) and (301,273.58) .. (301,278) -- (301,302) .. controls (301,306.42) and (297.42,310) .. (293,310) -- (138,310) .. controls (133.58,310) and (130,306.42) .. (130,302) -- cycle ; 
		
		\path  [shading=_26jr9qcf1,_7lht7677c] (359,278) .. controls (359,273.58) and (362.58,270) .. (367,270) -- (522,270) .. controls (526.42,270) and (530,273.58) .. (530,278) -- (530,302) .. controls (530,306.42) and (526.42,310) .. (522,310) -- (367,310) .. controls (362.58,310) and (359,306.42) .. (359,302) -- cycle ; 
		\draw   (359,278) .. controls (359,273.58) and (362.58,270) .. (367,270) -- (522,270) .. controls (526.42,270) and (530,273.58) .. (530,278) -- (530,302) .. controls (530,306.42) and (526.42,310) .. (522,310) -- (367,310) .. controls (362.58,310) and (359,306.42) .. (359,302) -- cycle ; 
		
		\path  [shading=_k2der5oz4,_bbj93ym0r] (359,208) .. controls (359,203.58) and (362.58,200) .. (367,200) -- (522,200) .. controls (526.42,200) and (530,203.58) .. (530,208) -- (530,232) .. controls (530,236.42) and (526.42,240) .. (522,240) -- (367,240) .. controls (362.58,240) and (359,236.42) .. (359,232) -- cycle ; 
		\draw   (359,208) .. controls (359,203.58) and (362.58,200) .. (367,200) -- (522,200) .. controls (526.42,200) and (530,203.58) .. (530,208) -- (530,232) .. controls (530,236.42) and (526.42,240) .. (522,240) -- (367,240) .. controls (362.58,240) and (359,236.42) .. (359,232) -- cycle ; 
		
		\path  [shading=_0pwh4sk7t,_zb9z0e029] (130,208) .. controls (130,203.58) and (133.58,200) .. (138,200) -- (293,200) .. controls (297.42,200) and (301,203.58) .. (301,208) -- (301,232) .. controls (301,236.42) and (297.42,240) .. (293,240) -- (138,240) .. controls (133.58,240) and (130,236.42) .. (130,232) -- cycle ; 
		\draw   (130,208) .. controls (130,203.58) and (133.58,200) .. (138,200) -- (293,200) .. controls (297.42,200) and (301,203.58) .. (301,208) -- (301,232) .. controls (301,236.42) and (297.42,240) .. (293,240) -- (138,240) .. controls (133.58,240) and (130,236.42) .. (130,232) -- cycle ; 
		
		
		\path  [shading=_xsbqfayhq,_b74qmq0s2] (130,138) .. controls (130,133.58) and (133.58,130) .. (138,130) -- (293,130) .. controls (297.42,130) and (301,133.58) .. (301,138) -- (301,162) .. controls (301,166.42) and (297.42,170) .. (293,170) -- (138,170) .. controls (133.58,170) and (130,166.42) .. (130,162) -- cycle ; 
		\draw   (130,138) .. controls (130,133.58) and (133.58,130) .. (138,130) -- (293,130) .. controls (297.42,130) and (301,133.58) .. (301,138) -- (301,162) .. controls (301,166.42) and (297.42,170) .. (293,170) -- (138,170) .. controls (133.58,170) and (130,166.42) .. (130,162) -- cycle ; 
		
		\path  [shading=_s76zcr58y,_hfq0n7d22] (245,68) .. controls (245,63.58) and (248.58,60) .. (253,60) -- (408,60) .. controls (412.42,60) and (416,63.58) .. (416,68) -- (416,92) .. controls (416,96.42) and (412.42,100) .. (408,100) -- (253,100) .. controls (248.58,100) and (245,96.42) .. (245,92) -- cycle ; 
		\draw   (245,68) .. controls (245,63.58) and (248.58,60) .. (253,60) -- (408,60) .. controls (412.42,60) and (416,63.58) .. (416,68) -- (416,92) .. controls (416,96.42) and (412.42,100) .. (408,100) -- (253,100) .. controls (248.58,100) and (245,96.42) .. (245,92) -- cycle ; 
			\path  [shading=_pvcfv9qai,_mckjdonky] (240,68) .. controls (240,63.58) and (243.58,60) .. (248,60) -- (412,60) .. controls (416.42,60) and (420,63.58) .. (420,68) -- (420,92) .. controls (420,96.42) and (416.42,100) .. (412,100) -- (248,100) .. controls (243.58,100) and (240,96.42) .. (240,92) -- cycle ; 
		\draw   (240,68) .. controls (240,63.58) and (243.58,60) .. (248,60) -- (412,60) .. controls (416.42,60) and (420,63.58) .. (420,68) -- (420,92) .. controls (420,96.42) and (416.42,100) .. (412,100) -- (248,100) .. controls (243.58,100) and (240,96.42) .. (240,92) -- cycle ; 
		\draw    (220,310) -- (220,340) ;
		\draw    (220,240) -- (220,270) ;
		\draw    (450,310) -- (450,340) ;
		\draw    (450,240) -- (450,270) ;
		\draw    (220,170) -- (220,200) ;
		\draw    (330,100) -- (220,130) ;
		\draw    (330,100) -- (450,200) ;
		
		\draw (450.27,359.5) node   [align=left] {$\loar=$\:\qcdcl\:};
		\draw (450.27,290.5) node   [align=left] {$\asroar$};
		\draw (450.27,220.5) node   [align=left] {$\aoar$};
		\draw (221.27,360.5) node   [align=left] {$\lonr$};
		\draw (221.27,290.5) node   [align=left] {$\asonr$};
		\draw (221.27,219.5) node   [align=left] {$\aonr$};
		\draw (221.27,150.5) node   [align=left] {\qres};
		\draw (330.27,80.5) node   [align=left] {\ldqres};
		
		\end{tikzpicture}
		
		\caption{Overview of the defined QCDCL proof systems. Lines denote p-simulations. \label{fig:simulation-qcdcl}}
	\end{figure}

	Combining the two policies $\ar$ and $\nr$ for unit propagation and the four policies $\ao$, $\lo$, $\aso$, and $\asro$, we obtain six QCDCL systems. These are depicted in Figure~\ref{fig:simulation-qcdcl} (we are not interested in the systems $\asoar$ and $\asronr$ since $\aso$ and $\asro$ would not be beneficial in these combinations). As mentioned, combining $\lo$ with $\ar$ yields the standard QCDCL system. The other five variants are introduced here for the first time.
	
	To actually show that these systems are sound and to proof-theoretically analyse their strength, we turn these six systems into formal refutational proof systems for QBF (Definition~\ref{DefinitionQCDCLrefutation}). 
	%
	%
	We establish that all these systems are sound and complete QBF proof systems (Thm~\ref{TheoremSystemsSimulatedByQRes} \& Thm~\ref{thm:qcdcl-complete}).
	
	Soundness is shown via efficiently constructing \ldqres proofs from QCDCL proofs. Crucially, when using the unit-propagation policy $\nr$, then no long-distance steps are actually needed and we just construct \qres proofs. The resulting simulations are depicted in Figure~\ref{fig:simulation-qcdcl}. Simulations between the QCDCL calculi follow by definition. We remark already here that this simulation order  simplifies further due to our results described in the next two subsections (cf.\ Figure~\ref{fig:qcdcl-sim-order}).

	From a theoretical point of view, formalising the QCDCL ingredients into proof systems enables a precise proof-theoretic analysis of the QCDCL systems and their comparison to \qres. This already was the underlying feature of our results in the previous two subsections, showing the incomparability of \qres and \qcdcl (Section~\ref{sec:qcdcl-qres-intro}) and the lower bounds for \qcdcl (Section~\ref{subsec:lb-qcdcl-intro}). We use it further to obtain a version of  QCDCL  that is even p-equivalent to \qres.

	\subsubsection{A QCDCL system that characterises Q-resolution}
	
	In one of our main results we obtain a QCDCL characterisation of \qres. Of course, given that \qres and \qcdcl are incomparable (Section~\ref{sec:qcdcl-qres-intro}), we cannot hope to achieve such a characterisation by simply strengthening some of the QCDCL policies.\footnote{Such hope might not have seemed totally implausible prior to this paper, e.g.\ \cite{DBLP:conf/sat/Janota16} states that `CDCL QBF solving appears to be quite weak compared to general \qres.'} As explained in the previous subsection, traditional QCDCL is using the decision policy $\lo$ and the unit-propagation policy $\ar$. To obtain a QCDCL system equivalent to \qres, we will have to change both policies. We will \emph{strengthen} the decision policy and replace $\lo$ by $\ao$ (we could also replace it with the intermediate version $\aso$). In addition, we will somewhat \emph{weaken} the unit propagation policy from $\ar$ to $\nr$.\footnote{While intuitively $\nr$ might indeed appear weaker then $\ar$ (it produces fewer unit propagations), we show in the next subsection that they are in fact incomparable, cf.~Figure~\ref{fig:qcdcl-sim-order}.} 
		This leads to a characterisation of \qres in terms of $\aonr$, and $\asonr$ (Theorem~\ref{TheoremSimComplete}).

	\subsubsection{The simulation order of QCDCL proof systems}
	
	We further analyse the simulation order of the defined QCDCL and QBF resolution systems, cf.\ Figure~\ref{fig:qcdcl-sim-order} which almost completely determines the simulations and separations between the systems involved (cf.\ Section~\ref{sec:conclusion} for the open cases).
	
	We highlight the most interesting findings (in addition to the results already described). 
	Firstly, we show that the unit-propagation policies $\ar$ and $\nr$ are incomparable when fixing the decision policy $\lo$ used in practical QCDCL.  
	%
	Secondly, we show that replacing the decision policy $\lo$ in \qcdcl with the more liberal decision policy $\asro$ yields exponentially shorter QCDCL runs, which we demonstrate on the  $\mathtt{Equality}_n$ formulas.
	Again, this theoretical result identifies potential for improvements in practical solving (cf.\ also the discussion in the concluding Section~\ref{sec:conclusion}).

	

	\begin{figure}\label{FigureAllConnections}
		\begin{center}
			
			
			\tikzset {_w0ng3dawj/.code = {\pgfsetadditionalshadetransform{ \pgftransformshift{\pgfpoint{0 bp } { 0 bp }  }  \pgftransformrotate{-270 }  \pgftransformscale{2 }  }}}
			\pgfdeclarehorizontalshading{_bvxyh9yzk}{150bp}{rgb(0bp)=(1,1,1);
				rgb(37.5bp)=(1,1,1);
				rgb(50bp)=(0.95,0.95,0.95);
				rgb(50.25bp)=(0.88,0.88,0.88);
				rgb(62.5bp)=(0.96,0.96,0.96);
				rgb(100bp)=(0.96,0.96,0.96)}
			
			
			\tikzset {_xufmlvnjs/.code = {\pgfsetadditionalshadetransform{ \pgftransformshift{\pgfpoint{0 bp } { 0 bp }  }  \pgftransformrotate{-270 }  \pgftransformscale{2 }  }}}
			\pgfdeclarehorizontalshading{_kqvw9gypw}{150bp}{rgb(0bp)=(1,1,1);
				rgb(37.5bp)=(1,1,1);
				rgb(50bp)=(0.95,0.95,0.95);
				rgb(50.25bp)=(0.88,0.88,0.88);
				rgb(62.5bp)=(0.96,0.96,0.96);
				rgb(100bp)=(0.96,0.96,0.96)}
			
			
			\tikzset {_s03qht2bn/.code = {\pgfsetadditionalshadetransform{ \pgftransformshift{\pgfpoint{0 bp } { 0 bp }  }  \pgftransformrotate{-270 }  \pgftransformscale{2 }  }}}
			\pgfdeclarehorizontalshading{_1x1c9muq4}{150bp}{rgb(0bp)=(1,1,1);
				rgb(37.5bp)=(1,1,1);
				rgb(50bp)=(0.95,0.95,0.95);
				rgb(50.25bp)=(0.88,0.88,0.88);
				rgb(62.5bp)=(0.96,0.96,0.96);
				rgb(100bp)=(0.96,0.96,0.96)}
			
			
			\tikzset {_r5s47byy9/.code = {\pgfsetadditionalshadetransform{ \pgftransformshift{\pgfpoint{0 bp } { 0 bp }  }  \pgftransformrotate{-270 }  \pgftransformscale{2 }  }}}
			\pgfdeclarehorizontalshading{_t36e4exrq}{150bp}{rgb(0bp)=(1,1,1);
				rgb(37.5bp)=(1,1,1);
				rgb(50bp)=(0.95,0.95,0.95);
				rgb(50.25bp)=(0.88,0.88,0.88);
				rgb(62.5bp)=(0.96,0.96,0.96);
				rgb(100bp)=(0.96,0.96,0.96)}
			
			
			\tikzset {_3vslq5rn6/.code = {\pgfsetadditionalshadetransform{ \pgftransformshift{\pgfpoint{0 bp } { 0 bp }  }  \pgftransformrotate{-270 }  \pgftransformscale{2 }  }}}
			\pgfdeclarehorizontalshading{_x55r0diw0}{150bp}{rgb(0bp)=(1,1,1);
				rgb(37.5bp)=(1,1,1);
				rgb(50bp)=(0.95,0.95,0.95);
				rgb(50.25bp)=(0.88,0.88,0.88);
				rgb(62.5bp)=(0.96,0.96,0.96);
				rgb(100bp)=(0.96,0.96,0.96)}
			
			
			\tikzset {_jip2nb5qg/.code = {\pgfsetadditionalshadetransform{ \pgftransformshift{\pgfpoint{0 bp } { 0 bp }  }  \pgftransformrotate{-270 }  \pgftransformscale{2 }  }}}
			\pgfdeclarehorizontalshading{_nd0h36q6n}{150bp}{rgb(0bp)=(1,1,1);
				rgb(37.5bp)=(1,1,1);
				rgb(50bp)=(0.95,0.95,0.95);
				rgb(50.25bp)=(0.88,0.88,0.88);
				rgb(62.5bp)=(0.96,0.96,0.96);
				rgb(100bp)=(0.96,0.96,0.96)}
			
			
			\tikzset {_2d7jsth3w/.code = {\pgfsetadditionalshadetransform{ \pgftransformshift{\pgfpoint{0 bp } { 0 bp }  }  \pgftransformrotate{-270 }  \pgftransformscale{2 }  }}}
			\pgfdeclarehorizontalshading{_2u2ycd5gz}{150bp}{rgb(0bp)=(1,1,1);
				rgb(37.5bp)=(1,1,1);
				rgb(50bp)=(0.95,0.95,0.95);
				rgb(50.25bp)=(0.88,0.88,0.88);
				rgb(62.5bp)=(0.96,0.96,0.96);
				rgb(100bp)=(0.96,0.96,0.96)}
			
			
			\tikzset {_697w8ohxo/.code = {\pgfsetadditionalshadetransform{ \pgftransformshift{\pgfpoint{0 bp } { 0 bp }  }  \pgftransformrotate{-270 }  \pgftransformscale{2 }  }}}
			\pgfdeclarehorizontalshading{_p7azjjclf}{150bp}{rgb(0bp)=(1,1,1);
				rgb(37.5bp)=(1,1,1);
				rgb(50bp)=(0.95,0.95,0.95);
				rgb(50.25bp)=(0.88,0.88,0.88);
				rgb(62.5bp)=(0.96,0.96,0.96);
				rgb(100bp)=(0.96,0.96,0.96)}
			
			
			\tikzset {_2q5z798r7/.code = {\pgfsetadditionalshadetransform{ \pgftransformshift{\pgfpoint{0 bp } { 0 bp }  }  \pgftransformrotate{-270 }  \pgftransformscale{2 }  }}}
			\pgfdeclarehorizontalshading{_4eq2ylx40}{150bp}{rgb(0bp)=(1,1,1);
				rgb(37.5bp)=(1,1,1);
				rgb(50bp)=(0.95,0.95,0.95);
				rgb(50.25bp)=(0.88,0.88,0.88);
				rgb(62.5bp)=(0.96,0.96,0.96);
				rgb(100bp)=(0.96,0.96,0.96)}
			
			
			\tikzset {_gr8uh6ycu/.code = {\pgfsetadditionalshadetransform{ \pgftransformshift{\pgfpoint{0 bp } { 0 bp }  }  \pgftransformrotate{-270 }  \pgftransformscale{2 }  }}}
			\pgfdeclarehorizontalshading{_6hqbi887a}{150bp}{rgb(0bp)=(1,1,1);
				rgb(37.5bp)=(1,1,1);
				rgb(50bp)=(0.95,0.95,0.95);
				rgb(50.25bp)=(0.88,0.88,0.88);
				rgb(62.5bp)=(0.96,0.96,0.96);
				rgb(100bp)=(0.96,0.96,0.96)}
			
			
			\tikzset {_9owjmjol5/.code = {\pgfsetadditionalshadetransform{ \pgftransformshift{\pgfpoint{0 bp } { 0 bp }  }  \pgftransformrotate{-270 }  \pgftransformscale{2 }  }}}
			\pgfdeclarehorizontalshading{_7oy3akdpb}{150bp}{rgb(0bp)=(1,1,1);
				rgb(37.5bp)=(1,1,1);
				rgb(50bp)=(0.95,0.95,0.95);
				rgb(50.25bp)=(0.88,0.88,0.88);
				rgb(62.5bp)=(0.96,0.96,0.96);
				rgb(100bp)=(0.96,0.96,0.96)}
			
			
			\tikzset {_d358x9nmc/.code = {\pgfsetadditionalshadetransform{ \pgftransformshift{\pgfpoint{0 bp } { 0 bp }  }  \pgftransformrotate{-270 }  \pgftransformscale{2 }  }}}
			\pgfdeclarehorizontalshading{_hpkgrp60p}{150bp}{rgb(0bp)=(1,1,1);
				rgb(37.5bp)=(1,1,1);
				rgb(50bp)=(0.95,0.95,0.95);
				rgb(50.25bp)=(0.88,0.88,0.88);
				rgb(62.5bp)=(0.96,0.96,0.96);
				rgb(100bp)=(0.96,0.96,0.96)}
			
			
			\tikzset {_z2l2iylow/.code = {\pgfsetadditionalshadetransform{ \pgftransformshift{\pgfpoint{0 bp } { 0 bp }  }  \pgftransformrotate{-270 }  \pgftransformscale{2 }  }}}
			\pgfdeclarehorizontalshading{_xcnnim56a}{150bp}{rgb(0bp)=(1,1,1);
				rgb(37.5bp)=(1,1,1);
				rgb(50bp)=(0.95,0.95,0.95);
				rgb(50.25bp)=(0.88,0.88,0.88);
				rgb(62.5bp)=(0.96,0.96,0.96);
				rgb(100bp)=(0.96,0.96,0.96)}
			
			
			\tikzset {_0nigzdy6z/.code = {\pgfsetadditionalshadetransform{ \pgftransformshift{\pgfpoint{0 bp } { 0 bp }  }  \pgftransformrotate{-270 }  \pgftransformscale{2 }  }}}
			\pgfdeclarehorizontalshading{_fo6h6lj6i}{150bp}{rgb(0bp)=(1,1,1);
				rgb(37.5bp)=(1,1,1);
				rgb(50bp)=(0.95,0.95,0.95);
				rgb(50.25bp)=(0.88,0.88,0.88);
				rgb(62.5bp)=(0.96,0.96,0.96);
				rgb(100bp)=(0.96,0.96,0.96)}
			\tikzset{every picture/.style={line width=0.75pt}} 
			
			\begin{tikzpicture}[x=0.75pt,y=0.75pt,yscale=-1.02,xscale=1.02]
			
			\path  [shading=_bvxyh9yzk,_w0ng3dawj] (50,169.83) .. controls (50,165.93) and (53.17,162.76) .. (57.07,162.76) -- (407.56,162.76) .. controls (411.47,162.76) and (414.63,165.93) .. (414.63,169.83) -- (414.63,191.05) .. controls (414.63,194.95) and (411.47,198.12) .. (407.56,198.12) -- (57.07,198.12) .. controls (53.17,198.12) and (50,194.95) .. (50,191.05) -- cycle ; 
			\draw   (50,169.83) .. controls (50,165.93) and (53.17,162.76) .. (57.07,162.76) -- (407.56,162.76) .. controls (411.47,162.76) and (414.63,165.93) .. (414.63,169.83) -- (414.63,191.05) .. controls (414.63,194.95) and (411.47,198.12) .. (407.56,198.12) -- (57.07,198.12) .. controls (53.17,198.12) and (50,194.95) .. (50,191.05) -- cycle ; 
			
			\path  [shading=_kqvw9gypw,_xufmlvnjs] (447.92,169.83) .. controls (447.92,165.93) and (451.08,162.76) .. (454.99,162.76) -- (602.93,162.76) .. controls (606.83,162.76) and (610,165.93) .. (610,169.83) -- (610,191.05) .. controls (610,194.95) and (606.83,198.12) .. (602.93,198.12) -- (454.99,198.12) .. controls (451.08,198.12) and (447.92,194.95) .. (447.92,191.05) -- cycle ; 
			\draw   (447.92,169.83) .. controls (447.92,165.93) and (451.08,162.76) .. (454.99,162.76) -- (602.93,162.76) .. controls (606.83,162.76) and (610,165.93) .. (610,169.83) -- (610,191.05) .. controls (610,194.95) and (606.83,198.12) .. (602.93,198.12) -- (454.99,198.12) .. controls (451.08,198.12) and (447.92,194.95) .. (447.92,191.05) -- cycle ; 

			\path  [shading=_1x1c9muq4,_s03qht2bn] (305.51,231.71) .. controls (305.51,227.81) and (308.67,224.64) .. (312.58,224.64) -- (407.02,224.64) .. controls (410.92,224.64) and (414.09,227.81) .. (414.09,231.71) -- (414.09,252.93) .. controls (414.09,256.83) and (410.92,260) .. (407.02,260) -- (312.58,260) .. controls (308.67,260) and (305.51,256.83) .. (305.51,252.93) -- cycle ; 
			\draw   (305.51,231.71) .. controls (305.51,227.81) and (308.67,224.64) .. (312.58,224.64) -- (407.02,224.64) .. controls (410.92,224.64) and (414.09,227.81) .. (414.09,231.71) -- (414.09,252.93) .. controls (414.09,256.83) and (410.92,260) .. (407.02,260) -- (312.58,260) .. controls (308.67,260) and (305.51,256.83) .. (305.51,252.93) -- cycle ; 
			
			\path  [shading=_t36e4exrq,_r5s47byy9] (360.61,46.07) .. controls (360.61,42.17) and (363.77,39) .. (367.68,39) -- (494.86,39) .. controls (498.77,39) and (501.94,42.17) .. (501.94,46.07) -- (501.94,67.29) .. controls (501.94,71.19) and (498.77,74.36) .. (494.86,74.36) -- (367.68,74.36) .. controls (363.77,74.36) and (360.61,71.19) .. (360.61,67.29) -- cycle ; 
			\draw   (360.61,46.07) .. controls (360.61,42.17) and (363.77,39) .. (367.68,39) -- (494.86,39) .. controls (498.77,39) and (501.94,42.17) .. (501.94,46.07) -- (501.94,67.29) .. controls (501.94,71.19) and (498.77,74.36) .. (494.86,74.36) -- (367.68,74.36) .. controls (363.77,74.36) and (360.61,71.19) .. (360.61,67.29) -- cycle ; 
			
			\path  [shading=_x55r0diw0,_3vslq5rn6] (447.92,107.95) .. controls (447.92,104.05) and (451.08,100.88) .. (454.99,100.88) -- (602.93,100.88) .. controls (606.83,100.88) and (610,104.05) .. (610,107.95) -- (610,129.17) .. controls (610,133.07) and (606.83,136.24) .. (602.93,136.24) -- (454.99,136.24) .. controls (451.08,136.24) and (447.92,133.07) .. (447.92,129.17) -- cycle ; 
			\draw   (447.92,107.95) .. controls (447.92,104.05) and (451.08,100.88) .. (454.99,100.88) -- (602.93,100.88) .. controls (606.83,100.88) and (610,104.05) .. (610,107.95) -- (610,129.17) .. controls (610,133.07) and (606.83,136.24) .. (602.93,136.24) -- (454.99,136.24) .. controls (451.08,136.24) and (447.92,133.07) .. (447.92,129.17) -- cycle ; 

			\draw [shading=_nd0h36q6n,_jip2nb5qg] [dash pattern={on 4.5pt off 4.5pt}]  (414.63,180.44) -- (447.37,180.44) ;
			\draw [shading=_2u2ycd5gz,_2d7jsth3w]   (360.07,224.64) -- (360.07,198.12) ;
			\draw [shading=_p7azjjclf,_697w8ohxo] [dash pattern={on 4.5pt off 4.5pt}]  (360.07,224.64) -- (530,200) ;
			\draw [shading=_4eq2ylx40,_2q5z798r7]   (529,136.48) -- (529,163) ;
			\draw [shading=_6hqbi887a,_gr8uh6ycu]   (425.55,74.36) .. controls (427.56,73.14) and (429.18,73.54) .. (430.4,75.55) .. controls (431.63,77.56) and (433.25,77.96) .. (435.26,76.74) .. controls (437.28,75.52) and (438.9,75.92) .. (440.11,77.94) .. controls (441.34,79.95) and (442.96,80.35) .. (444.97,79.13) .. controls (446.98,77.91) and (448.6,78.31) .. (449.83,80.32) .. controls (451.05,82.33) and (452.67,82.73) .. (454.68,81.51) .. controls (456.69,80.29) and (458.31,80.69) .. (459.54,82.7) .. controls (460.75,84.72) and (462.37,85.12) .. (464.39,83.9) .. controls (466.4,82.68) and (468.02,83.08) .. (469.25,85.09) .. controls (470.48,87.1) and (472.1,87.5) .. (474.11,86.28) .. controls (476.12,85.06) and (477.74,85.46) .. (478.96,87.47) .. controls (480.19,89.48) and (481.81,89.88) .. (483.82,88.66) .. controls (485.84,87.44) and (487.46,87.84) .. (488.67,89.86) .. controls (489.9,91.87) and (491.52,92.27) .. (493.53,91.05) .. controls (495.54,89.83) and (497.16,90.23) .. (498.38,92.24) .. controls (499.61,94.25) and (501.23,94.65) .. (503.24,93.43) .. controls (505.25,92.21) and (506.87,92.61) .. (508.1,94.62) .. controls (509.31,96.64) and (510.93,97.04) .. (512.95,95.82) .. controls (514.96,94.6) and (516.58,95) .. (517.81,97.01) .. controls (519.03,99.02) and (520.65,99.42) .. (522.66,98.2) .. controls (524.67,96.98) and (526.29,97.38) .. (527.52,99.39) -- (530,100) -- (530,100) ;
			\draw [shading=_7oy3akdpb,_9owjmjol5]   (425.55,74.36) -- (360.07,162.76) ;
			\draw  [fill={rgb, 255:red, 255; green, 255; blue, 255 }  ,fill opacity=1 ] (334.5,212) .. controls (334.5,207.03) and (338.98,203) .. (344.5,203) .. controls (350.02,203) and (354.5,207.03) .. (354.5,212) .. controls (354.5,216.97) and (350.02,221) .. (344.5,221) .. controls (338.98,221) and (334.5,216.97) .. (334.5,212) -- cycle ;
			\draw  [fill={rgb, 255:red, 255; green, 255; blue, 255 }  ,fill opacity=1 ] (440,227) .. controls (440,222.03) and (444.48,218) .. (450,218) .. controls (455.52,218) and (460,222.03) .. (460,227) .. controls (460,231.97) and (455.52,236) .. (450,236) .. controls (444.48,236) and (440,231.97) .. (440,227) -- cycle ;
			\draw  [fill={rgb, 255:red, 255; green, 255; blue, 255 }  ,fill opacity=1 ] (420.5,168) .. controls (420.5,163.03) and (424.98,159) .. (430.5,159) .. controls (436.02,159) and (440.5,163.03) .. (440.5,168) .. controls (440.5,172.97) and (436.02,177) .. (430.5,177) .. controls (424.98,177) and (420.5,172.97) .. (420.5,168) -- cycle ;
			\draw  [fill={rgb, 255:red, 255; green, 255; blue, 255 }  ,fill opacity=1 ] (371.5,109) .. controls (371.5,104.03) and (375.98,100) .. (381.5,100) .. controls (387.02,100) and (391.5,104.03) .. (391.5,109) .. controls (391.5,113.97) and (387.02,118) .. (381.5,118) .. controls (375.98,118) and (371.5,113.97) .. (371.5,109) -- cycle ;
			\draw  [fill={rgb, 255:red, 255; green, 255; blue, 255 }  ,fill opacity=1 ] (540,149) .. controls (540,144.03) and (544.48,140) .. (550,140) .. controls (555.52,140) and (560,144.03) .. (560,149) .. controls (560,153.97) and (555.52,158) .. (550,158) .. controls (544.48,158) and (540,153.97) .. (540,149) -- cycle ;
			
			\draw  [fill={rgb, 255:red, 255; green, 255; blue, 255 }  ,fill opacity=1 ] (53,20) -- (343,20) -- (343,130) -- (53,130) -- cycle ;
			\draw [shading=_hpkgrp60p,_d358x9nmc]   (83,30) -- (63,50) ;
			
			\draw [shading=_xcnnim56a,_z2l2iylow] [dash pattern={on 4.5pt off 4.5pt}]  (83,66) -- (63,86) ;
			\draw [shading=_fo6h6lj6i,_0nigzdy6z]   (83,101) .. controls (83,103.36) and (81.82,104.54) .. (79.46,104.54) .. controls (77.11,104.54) and (75.93,105.72) .. (75.93,108.07) .. controls (75.93,110.43) and (74.75,111.61) .. (72.39,111.61) .. controls (70.04,111.61) and (68.86,112.79) .. (68.86,115.14) .. controls (68.86,117.5) and (67.68,118.68) .. (65.32,118.68) -- (63,121) -- (63,121) ;
			
			\draw (232.32,180.44) node  [color={rgb, 255:red, 0; green, 0; blue, 0 }  ,opacity=1 ] [align=left] {\textbf{Q-resolution}$\displaystyle \equiv_p \asonr\displaystyle \equiv_p \aonr$};
			\draw (359.8,242.32) node  [color={rgb, 255:red, 0; green, 0; blue, 0 }  ,opacity=1 ] [align=left] {\textcolor[rgb]{0,0,0}{$\lonr$}};
			\draw (431.27,56.68) node   [align=left] {\textcolor[rgb]{0,0,0}{\textbf{LD-Q-resolution}}};
			\draw (340.14,206.72) node [anchor=north west][inner sep=0.75pt]  [font=\scriptsize] [align=left] {4};
			\draw (446,222) node [anchor=north west][inner sep=0.75pt]  [font=\scriptsize] [align=left] {5};
			\draw (426.14,162.72) node [anchor=north west][inner sep=0.75pt]  [font=\scriptsize] [align=left] {1};
			\draw (377.14,103.72) node [anchor=north west][inner sep=0.75pt]  [font=\scriptsize] [align=left] {2};
			\draw (94,26) node [anchor=north west][inner sep=0.75pt]  [font=\scriptsize] [align=left] {{\footnotesize strictly stronger}\\{\footnotesize (p-simulation + exponential separation)}};
			\draw (94,62) node [anchor=north west][inner sep=0.75pt]  [font=\scriptsize] [align=left] {{\footnotesize incomparable (exponential separations }\\{\footnotesize in both directions)}};
			\draw (94,97) node [anchor=north west][inner sep=0.75pt]  [font=\scriptsize] [align=left] {{\footnotesize p-simulation}\\{\footnotesize (equivalence/separation open)}};
			\draw (528.96,118.56) node   [align=left] {\textcolor[rgb]{0,0,0}{$\asroar$}};
			\draw (528.96,180.44) node  [color={rgb, 255:red, 0; green, 0; blue, 0 }  ,opacity=1 ] [align=left] {\textcolor[rgb]{0,0,0}{$\loar=$\:\qcdcl}};
			\draw (545.64,143.72) node [anchor=north west][inner sep=0.75pt]  [font=\scriptsize] [align=left] {3};

			\end{tikzpicture}

		\end{center}
		
		\begin{center}
			
			\tcbox[left=0mm,right=0mm,top=0mm,bottom=0mm,boxsep=0mm,
			toptitle=0.5mm,bottomtitle=0.5mm]{%
				\begin{tabular}{|c|p{\textwidth*3/7}|}
					\circled{1} & Theorem \ref{TheoremIncomparable}  ($\mathtt{QParity}_n$, $\mathtt{Trapdoor}_n$)  \\
					\circled{2} & \citespace\cite{ELW13,BBH19,BeyersdorffCJ19} ($\mathtt{Equality}_n$, $\mathtt{QParity}_n$, $\mathtt{KBKF}_n$) \\
					\circled{3} & Theorem \ref{CorASROARdstrongerThanLOAR} ($\mathtt{Equality}_n$)\\
					\circled{4} & Proposition \ref{PropLonsingsFormula} ($\Lon_n$)  \\
					\circled{5} & Theorem \ref{TheoremLONRandLOARincomparable} ($\mathtt{QParity}_n$, $\mathtt{Trapdoor}_n$)\\
			\end{tabular}}

		\end{center}
		
		\caption{The simulation order of QCDCL and QBF resolution systems. The table contains pointers to the separating formulas. \label{fig:qcdcl-sim-order}}
	\end{figure}

	\subsection{Organisation}
	
	The main part of the article is organised slightly differently from the order of results as outlined above. We start in Section~\ref{sec:prelim} with reviewing relevant notions concerning quantified Boolean logic, QBF proof systems, and QCDCL.
	
	Section~\ref{sec:qcdcl-ps} formalises QCDCL with a number of different policies for variable decision and unit propagation as proof systems. This constitutes the formal framework for the rest of the paper. The proof systems are shown to be sound and complete.
	
	Section~\ref{sec:separation-qcdcl-qres} shows the incomparability of \qcdcl and \qres by exponential separations. This is followed in Section~\ref{sec:hardness-qcdcl} by further hardness results for \qcdcl. 
	
	In Section~\ref{sec:qres-characterisation} we obtain the characterisation of \qres in terms of the new QCDCL proof system $\aonr$. We discuss implications of this result to the equivalence of propositional resolution and CDCL in
	Section~\ref{sec:prop-cdcl}.

	Section~\ref{sec:sim-order} reveals the full picture of the simulation order of the defined QCDCL proof systems. We conclude in Section~\ref{sec:conclusion} with some open questions and a discussion of the potential impact of our results for practice.

	\section{Preliminaries} \label{sec:prelim}
	
	\subsection{Propositional and quantified formulas}		
	We will consider propositional and quantified formulas over a countable set of variables. Variables and negations of variables are called  \emph{literals}, i.e., for a variable $x$ we can form two literals: $x$ and its negation $\bar{x}$. Sometimes we write $x^1$ instead of $x$ and $x^0$ instead of $\bar{x}$. We denote the corresponding variable as $\var(x):=\var(\bar{x}):=x$.
	
	A \emph{clause} is a disjunction $\ell_1\vee \ldots\vee \ell_m$ of some literals $\ell_1,\ldots,\ell_m$. We will sometimes view a clause as a set of literals, i.e., we will use the notation $\ell\in C$ if the literal $\ell$ is one of the literals in the clause $C$. If $m=1$, we will often write $(\ell_1)$ to emphasize the difference between literals and clauses. The \emph{empty clause} is the clause consisting of zero literals, denoted by $(\bot)$. For reasons of consistency it is helpful to define an \emph{empty literal}, denoted by $\bot$ in our case. As a consequence, we have $\bot\in (\bot)$, although we define the empty clause as a clause with zero literals.

	The negation of a clause $C=\ell_1\vee\ldots\vee\ell_m$ is called a term, i.e., terms are conjunctions $\bar{\ell}_m\wedge\ldots\wedge\bar{\ell}_m$ of literals. Similarly terms can be considered as sets of literals. A \emph{CNF} (\emph{conjunctive normal form}) is a conjunction of clauses. 
	
	Let $C=\ell_1\vee\ldots\vee \ell_m$. We define $\var(C):=\{ \var(\ell_1),\ldots,\var(\ell_m) \}$. For a CNF $\phi=C_1\wedge\ldots\wedge C_n$ we define $\var(\phi):=\bigcup_{i=1}^n\var(C_i)$.

	A clause or a set $C$ of literals is called \emph{tautological}, if there is a variable $x$ with $x,\bar{x}\in C$.
	
	An \emph{assignment} $\sigma$ of a set of variables $X$ is a non-tautological set of literals, such that for all $x\in X$ there is $\ell\in \sigma$ with $\var(\ell)=x$. The restriction of a clause $C$ by an assignment $\sigma$ is defined as
	\begin{align*}
	C|_{\sigma}:=\left\{ \begin{array}{ll} \top\text{ (true)}&\text{if $C\cap \sigma \neq \emptyset$,} \\ 
	\bigvee\limits_{\substack{\ell\in C\\\bar{\ell}\not\in \sigma}}\ell  &\text{otherwise.}\end{array}\right.
	\end{align*}
	For example, let $C=t\vee x\vee y\vee \bar{z}$ and define the assignment $\sigma:=\{ \bar{x},z,w \}$. Then we have $C|_{\sigma}=t\vee y$. Note that the set of assigned variables might differ from $\var(C)$. In our case, $\sigma$ is an assignment of the set $X:=\{ x,z,w \}$.

	One can interpret $\sigma$ as an operator that sets all literals from $\sigma$ to the boolean constant $1$. We denote the set of assignments of $X$ by $\langle X\rangle$. Assignments can also operate on CNFs in the natural sense. A CNF $\phi$ \emph{entails} another CNF $\psi$ if each assignment that satisfies $\phi$ also satisfies $\psi$ (denoted by $\phi \vDash \psi$).
	
	A \emph{QBF} (\emph{quantified Boolean formula}) $\Phi=\mathcal{Q}\cdot \phi$ is a propositional formula $\phi$ (also called \emph{matrix}) together with a \emph{prefix} $\mathcal{Q}$.   A prefix $Q_1x_1Q_2x_2\ldots Q_kx_k$ consists of variables $x_1,\ldots,x_k$ and quantifiers $Q_1,\ldots,Q_k\in \{\exists,\forall\}$. We obtain an equivalent formula if we unite adjacent quantifiers of the same type. Therefore we can always assume that our prefix is in the form of
	\begin{align*}
	\mathcal{Q}=Q'_1 X_1 Q'_2 X_2 \ldots Q'_s X_s
	\end{align*}
	with nonempty sets of variables $X_1,\ldots,X_s$ and quantifiers $Q'_1,\ldots,Q'_s\in \{\exists,\forall\}$ such that $Q'_{i}\neq Q'_{i+1}$ for $i\in [s-1]$.  For a variable $x$ in $\mathcal{Q}$ we denote the \emph{quantifier level} with respect to $\mathcal{Q}$ by $\lv(x)=\lv_{\Phi}(x)=i$, if $x\in X_i$. Variables from $\Phi$ are called \emph{existential}, if the corresponding quantifier is $\exists$, and \emph{universal} if the quantifier is $\forall$. We denote the set of existential variables from $\Phi$ by $\var_\exists(\Phi)$, and the set of universal variables by $\var_\forall(\Phi)$.
	
	A QBF whose matrix is a CNF is called a \emph{QCNF}. We require that all clauses from a matrix of a QCNF are non-tautological, otherwise we would just delete these clauses. This requirement is crucial for the correctness of the derivation rules we define later for our proof systems. Since we will only discuss refutational proof systems, we will always assume that all QCNFs we consider are false.
	
	A QBF can be interpreted as a game between two players: The $\exists$-player and the $\forall$-player. These players have to assign the respective variables one by one along the quantifier order from left to right. The $\forall$-player wins the game if and only if the matrix of the QBF gets falsified by this assignment. It is well known that for every false QBF $\Phi=\mathcal{Q}\cdot \phi$ there exists a winning strategy for the $\forall$-player.

	\subsection{Proof systems}
	
	A \emph{proof system} for a language $\mathcal{L}$ is a polynomial-time computable surjective function $f:\: \{ 0,1 \}^*\rightarrow \mathcal{L}$ \cite{CR79}. A \emph{proof} for $\phi \in \mathcal{L}$ is some $\pi\in \{0,1\}^*$ such that $f(\pi)=\phi$. In our case, the language $\mathcal{L}$ will be mostly  $\mathsf{UNSAT}$ (unsatisfiable formulas) or $\mathsf{FQBF}$ (false QBFs). For unsatisfiable or false formulas we often call the system \emph{refutational}. 
	
	To show that such a polynomial-time function $f$ is actually a proof system for $\mathcal{L}$, we have to verify two properties:
	\begin{itemize}
		\item \emph{soundness}: $f(\{ 0,1 \}^*)\subseteq \mathcal{L}$.
		\item \emph{completeness}: $f(\{ 0,1 \}^*)\supseteq \mathcal{L}$.
	\end{itemize}
	
	A proof system $f$ for a language $\mathcal{L}$ is \emph{simulated} by another proof system $g$ for $\mathcal{L}$, if there exists a function $h:\: \{0,1\}^*\rightarrow \{ 0,1 \}^*$ such that $g\circ h=f$ (denoted $f\leq g$) \cite{KP89}. If $h$ is polynomial-time computable, then we say that $g$ \emph{p-simulates} $f$ (denoted $f\leq_pg$) \cite{CR79}. If two systems p-simulate each other, they are \emph{p-equivalent} (denoted $f\equiv_p g$).

	\subsection{Q-resolution and long-distance Q-resolution}
	\label{subsec:qres-ldqres}
	
	Let $C_1$ and $C_2$ be two clauses of a QCNF $\Phi$ and let $\ell$ be an existential literal with $\var(\ell)\not\in \var(C_1)\cup\var(C_2)$. The \emph{resolvent} of  $C_1\vee \ell $ and $C_2\vee \bar{\ell}$ over $\ell $ is defined as
	\begin{align*}
	(C_1\vee \ell)\resop{\ell}(C_2\vee \bar{\ell})		
	:=C_1\vee C_2 \text{.}
	\end{align*}
	Let $C:=u_1\vee\ldots \vee u_m\vee x_1\vee\ldots \vee x_n \vee v_1\vee \ldots \vee v_s$ be a clause from $\Phi$, where $u_1,\ldots,u_m,v_1,\ldots,v_s$ are universal literals, $x_1,\ldots,x_n$ are existential literals and
	\begin{align*}
	\{ v\in C:\: \text{$v$ is universal and $\lv(v)>\lv(x_i)$ for all $i\in[n]$} \}=\{ v_1,\ldots,v_s \}\text{.}
	\end{align*}
	Then we can perform a \emph{reduction} step and obtain  
	\begin{align*}
	\red(C):=u_1\vee\ldots\vee u_m\vee x_1\vee \ldots \vee x_n\text{.}
	\end{align*}
	For a  CNF $\phi=\{ C_1,\ldots,C_k \}$ we define
	\begin{align*}
	\red(\phi):=\{\red(C_1),\ldots,\red(C_k)  \}\text{.}
	\end{align*}

	{\qres} \cite{DBLP:journals/iandc/BuningKF95} is a refutational proof system for false QCNFs. A \qres proof $\pi$ of a clause $C$ from a QCNF $\Phi=\mathcal{Q}\cdot \phi$ is a sequence of clauses $\pi=C_1,\ldots,C_m$ with $C_m=C$. Each $C_i$ has to be derived by one of the  following three rules:
	\begin{itemize}
		\item \emph{Axiom:} $C_i\in \phi$;
		\item \emph{Resolution:} $C_i=C_j\resop{x}C_k$ for some $j,k<i$ and $x\in \var_\exists(\Phi)$, and $C_i$ is non-tautological;
		\item \emph{Reduction:} $C_i=\red(C_j)$ for some $j<i$.
	\end{itemize}
	
	Note that none of our axioms are tautological by definition.
	A \emph{refutation} of a QCNF $\Phi$ is a proof of the empty clause $(\bot)$.

	For the simulation of the original version of QCDCL, the proof system \ldqres was introduced in \cite{ZM02,Balabanov12}. This extension of \qres allows to derive universal tautologies under specific conditions. As in \qres, there are three rules by which a clause $C_i$ can be derived. The axiom and reduction rules are identical to \qres, but the resolution rule is changed to 
	\begin{itemize}
		\item \emph{Resolution (long-distance):} $C_i= C_j\resop{x}C_k$ for some $j,k<i$ and $x\in \var_\exists(\Phi)$. The resolvent $C_i$ is allowed to contain a tautology $u\vee \bar{u}$ if $u$ is a universal variable. If $u\in \var(C_j)\cap \var(C_k)$, then we additionally require $\lv(u)>\lv(x)$.
	\end{itemize}
	
	
	Note that a \ldqres proof without tautologies is just a \qres proof.
	
	Creating universal tautologies without any assumptions is unsound in general. For example, consider the true QCNF $\Psi:=\forall u \exists x \cdot (u\vee \bar{x}) \wedge (\bar{u}\vee x)$. There is a winning strategy for the $\exists$-player by assigning $x$ equal to $u$. Hence, the step $\red\left( (u\vee \bar{x})\resop{x}(\bar{u}\vee x )    \right)=(\bot)$ is unsound since we resolved over an existential literal $x$ with $\lv_\Psi(x)>\lv_\Psi(u)$ while generating $u\vee \bar{u}$.
	
	\subsection{QCDCL}
	\label{subsec:qcdcl}
	
	Quantified conflict-driven clause learning (QCDCL) is the quantified version of the well-known \cdcl algorithm (see \cite{ZhangMMM01,DBLP:series/faia/SilvaLM09} for further details on \cdcl, and \cite{GiunchigliaNT06,LonsingDissertation,ZM02} for QCDCL). 
	Let $\Phi=\mathcal{Q}\cdot \phi$ be a false QCNF. 
	Roughly speaking, QCDCL consists of two processes: The \emph{propagation process} and the \emph{learning process}. 
	
	In the \emph{propagation process} we generate assignments  to the end that we obtain a conflict. We start with clauses from $\phi$ that force us to assign literals such that we do not falsify these clauses (subsequently called unit clauses). The underlying idea of this process is \emph{unit propagation}. One can think of a clause $x_1\vee \ldots\vee x_n$ as an implication $(\bar{x}_1\wedge\ldots\wedge \bar{x}_{n-1})\rightarrow x_n$. That is, if we already assigned the literals $\bar{x}_1,\ldots,\bar{x}_{n-1}$, then we are forced to assign $x_n$ in order to verify this clause. If $x_n$ was universal, this would already be a conflict since this clause must be true for both assignments of $x_n$ in order to not get falsified.  In general, we also have to insert reduction steps into this process. Hence we are interested in clauses that become unit after reduction. For example, the clause $(\bar{x}_1\wedge\ldots\wedge \bar{x}_{n-1})\rightarrow (x_n\vee u)$ for an existential literal $x_n$ and a universal literal $u$ with $\lv(x_n)<\lv(u)$ can also be used as an implication of $x_n$ for unit propagation.
	
	Of course we could also assign $\bar{x}_n$, immediately leading to a conflict by falsifying the clause $x_1\vee \ldots\vee x_n$, but this conflict would have been solely caused by this clause and would not give us any new information for the learning process. Our goal is to prolong a conflict as long as possible in the hope of learning something  helpful from it. However, it is not guaranteed that we can even perform any unit propagations by just starting with the formula.
	
	Therefore we will make \emph{decisions}, i.e., we assign literals without any solid reason. With the aid of these decisions (one can also think of assumptions) we can provoke further unit propagations. Since decision making is one of the non-deterministic components of the algorithm, we will try to keep its influence as low as possible. In detail, this means we will only make decisions if there are no more unit propagations available. In the classical QCDCL these decisions have to follow a level-order. This means we always have to decide the next available variable with the lowest quantifier level. However, we will later show that this condition is not necessary for soundness or completeness. There are even QCNFs whose hardness are based simply on this level-order, so leaving out this limitation would actually strengthen the algorithm. In our model these two policies will be denoted by $\lo$ and $\ao$.
	
	After we obtained a conflict, we can start the \emph{clause learning process}. Here the underlying idea is to use \qres resp.\ \ldqres. We start with the clause that caused our conflict and  resolve it with clauses that implied previous literals in the assignment in the reversed propagation order. At the end we  get a (hopefully) new clause such that each assignment that falsifies this clause also leads to a conflict. In addition we get a \ldqres-derivation of this learned clause from $\Phi$. We will add the learned clause to $\phi$, backtrack to a state before we assigned all literals of this clause  and start with the propagation process again. The algorithm ends as soon as we learn the empty clause $(\bot)$ and therefore obtain a refutation of $\Phi$. 
	
	In our work we want to formalize this algorithm as a proof system and slightly modify the system such that this new proof system becomes equivalent to  \qres. To  ensure we actually construct a \qres proof out of the learning process, we have to prevent the introduction of universal tautologies. As will become clear later, the reasons for these tautologies are reductions in the propagation process. This is the only way where both literals of universal variables can be introduced. This motivates disallowing these reductions in the propagation process. Later in the definitions we will denote these policies by $\ar$ and $\nr$.

	
	Usually QCDCL has to handle both refutations of false formulas as well as proving the validity of true formulas. For this purpose one would need to implement the so called \emph{cube learning} (or \emph{term learning}) for fulfilling assignments. But since we are  only interested in the refutation of formulas (otherwise we could not compare this system to \qres), we will omit this aspect of QCDCL.

	\section{Our framework: versions of QCDCL as proof systems}
	\label{sec:qcdcl-ps}
	
	In this section we define formal proof systems that capture QCDCL solving. 
	For this we need to formally define central ingredients of QCDCL solving, including trails, decision policies, unit propagation, and clause learning. For decisions and unit propagation we will consider different policies: those corresponding to QCDCL solving in practice and new policies, yet unexplored. We will show that the corresponding QCDCL proof systems are all sound and complete.
	
	We start with defining trails, decisions, unit propagations and our collection of policies. 	
	\begin{defi}[trails and policies for decision/unit propagation]\label{DefTrails}

		Let $\Phi=\mathcal{Q}\cdot \phi$ be a QCNF in $n$ variables.
		A \emph{trail} $\mathcal{T}$ for $\Phi$ is a sequence of literals (or $\bot$) of variables from $\Phi$  with some specific properties. We distinguish two types of literals in $\mathcal{T}$: \emph{decision literals}, that can be both existential and universal, and propagated literals, that are either existential or $\bot$.  Most of the time we write a  trail $\mathcal{T}$ as 
		\begin{align*}
		\mathcal{T}=(p_{(0,1)},\ldots, p_{(0,g_0)};\mathbf{d_{1}},p_{(1,1)},\ldots,p_{(1,g_{1})};\ldots; \mathbf{d_r},p_{(r,1)},\ldots ,p_{(r,g_r)}  )\text{.}
		\end{align*}
		We typically denote decision literals by $d_i$ and propagated literals by $p_{(i,j)}$. To emphasize decisions, we will set decision literals in the trail in \textbf{boldface} and put a semicolon at the end of each decision level. The literal $p_{(i,j)}$ represents the $j^{\text{th}}$ propagated literal in the $i^{\text{th}}$ decision level, determined by the corresponding decision $d_i$. The decision level $0$ is the only level where we do not have a decision literal. Similarly to clauses, we can view $\mathcal{T}$ as a set of literals or as an assignment and use the notation $x\in \mathcal{T}$ if the literal $x$ is contained in $\mathcal{T}$. 
		
		Let $s\in \{0,\ldots,r\}$ and $t\in \{0,\ldots,g_s\}$. The \emph{subtrail of $\mathcal{T}$ at the time $(s,t)$} is  the trail consisting of all literals from the leftmost literal in $\mathcal{T}$ up to (including) $p_{(s,t)}$, if $t\neq 0$, or $d_s$ otherwise. We denote this subtrail by $\mathcal{T}[s,t]$. The subtrail $\mathcal{T}[0,0]$ is defined as the empty trail.
		
		Now, we need some further requirements for $\mathcal{T}$ to be a trail for a QCNF $\Phi$.
		
		The decisions have to be non-tautological and non-repeating, i.e., we require $\var(d_i)\neq \var({d}_k)$ for each $i\neq k\in \{0,\dots,r\}$. If $\bot\in \mathcal{T}$, then this must be the last (rightmost) literal in $\mathcal{T}$. In this case we will say that $\mathcal{T}$ has \emph{run into a conflict}. 
		
		We define four policies, concerning the decision of literals, from which we can choose exactly one at a time:
		\begin{itemize}
			\item $\lo$\textbf{ - }For each $d_i\in \mathcal{T}$ we have $\lv(d_i)\leq \lv(x)$ for all $x\in \var(\phi)\backslash \var(\mathcal{T}[i-1,g_{i-1}])$. That means we have to decide the variables along the quantification order.
			\item $\aso$\textbf{ - }We can decide a literal $d_k$ if it is existential, or if it is universal and it holds $\lv(d_1)\leq \ldots \leq \lv(d_k)$. 
			\item $\asro$\textbf{ - }We can only decide an existential variable $x$ next, if and only if we already assigned all universal variables $u$ with $\lv(u)<\lv(x)$ before. 
			\item $\ao$\textbf{ - }We can choose any remaining literal as the next decision.
		\end{itemize}
		
		
		We define two more policies concerning  unit propagation. Again, we have to choose exactly one:
		\begin{itemize}
			\item $\ar$\textbf{ - }For each $p_{(i,j)}\in \mathcal{T}$ there has to be a clause $C\in \phi$ such that $\red(C|_{\mathcal{T}[i,j-1]})=(p_{(i,j)})$.
			\item $\nr$\textbf{ - }For each $p_{(i,j)}\in \mathcal{T}$ there has to be a clause $C\in \phi$ with $C|_{\mathcal{T}[i,j-1]}=(p_{(i,j)})$.
		\end{itemize}
		These clauses $C$ as described in the policies are called \emph{antecedent clauses}, which will be denoted by $\ante_\mathcal{T}(p_{(i,j)}):=C$. There could be more than one such suitable clause, in this case we will just choose one of them arbitrarily. These antecedent clauses clearly depend on the unit propagation policy we use.

		The size of a trail $\mathcal{T}$ can be measured by $|\mathcal{T}|$ (i.e., the cardinality of $\mathcal{T}$ as a set). Because each trail can at most contain all variables, we always have $|\mathcal{T}|\in \mathcal{O}(n)$.
		
	\end{defi}
	
	The policies $\ar$ and $\nr$ determine the notion of unit clauses, which are important for  unit propagation.

	\begin{defi}[unit clauses]
		Let $C$ be a clause. 		
		In the policy \ar, we call $C$ a \emph{unit clause} if $\red(C)=(x)$ for an existential literal $x$ or $x=\bot$.
		
		Otherwise, for \nr, we call $C$ a \emph{unit clause} if $C=(x)$ for an existential literal $x$ or $x=\bot$.
	\end{defi}
	
	Note that $(u)$ is not a unit clause under the policy $\nr$ for a universal literal $u$.
	
	Next we will formalise the process of clause learning from trails that run into a conflict. The idea is to resolve all antecedent clauses, starting from the end of the trail, until we stop at some point. We will always resolve over the corresponding propagated literal and skip literals not used for the implication of the conflict.

	\begin{defi}[learnable clauses]\label{DefLearnableClauses}
		Let $\Phi=\mathcal{Q}\cdot \phi$ be a QCNF and let 
		\begin{align*}
		\mathcal{T}=(p_{(0,1)},\ldots, p_{(0,g_0)};\mathbf{d_{1}},p_{(1,1)},\ldots,p_{(1,g_{1})};\ldots; \mathbf{d_r},p_{(r,1)},\ldots ,p_{(r,g_r)}  )
		\end{align*} be a trail with $p_{(r,g_r)}=\bot$ that follows policies $P\in \{ \lo,\aso,\ao  \}$ and $R\in \{ \ar,\nr \}$.
		We call a  clause \emph{learnable from} $\mathcal{T}$ if it appears in the  sequence
		\begin{align*}
		\mathcal{L_{\mathcal{T}}}:=(C_{(r,g_r)},\ldots,C_{(r,1)},\ldots,C_{(1,g_{1})},\ldots,C_{(1,1)},C_{(0,g_0)},\ldots,C_{(0,1)})
		\end{align*}
		where $C_{(r,g_r)}:=\red(\ante(p_{(r,g_r)}))$,
		\begin{align*}
		C_{(i,j)}:=\left\{ \begin{array}{ll}
		\red\left(   C_{(i,j+1)}\resop{p_{(i,j)}}\red(\ante(p_{(i,j)}))            \right)&\text{if ${\bar{p}_{(i,j)}}\in C_{(i,j+1)}$,} \\ 
		C_{(i,j+1)}&\text{otherwise}\end{array}\right.
		\end{align*}
		for $i\in \{0,\dots,r\}$, $j\in[g_i-1]$, and
		\begin{align*}
		C_{(i,g_i)}:=\left\{ \begin{array}{ll}
		\red\left(  C_{(i+1,1)}\resop{p_{(i,g_i)}}   \red(\ante(p_{(i,g_i)}))            \right)&\text{if $\bar{p}_{(i,g_i)}\in C_{(i+1,1)}$,} \\ 
		C_{(i+1,1)}&\text{otherwise}\end{array}\right.
		\end{align*}
		for $i\in \{0,\dots,r-1\}$.
	\end{defi}

	Note that clause learning works independently from the used policy. Even if we choose the policy \nr, we might have to make reduction steps in this process. After the construction of each trail $\mathcal{T}$ we will choose to learn exactly one clause from $\mathcal{L}_\mathcal{T}$. This choice represents a nondeterminism in the learning process.

	Next we formalise natural trails, where we are not allowed to skip unit propagations.
	
	\begin{defi}[natural trails]\label{DefNaturalTrails}
		We  call a trail $\mathcal{T}$ \emph{natural}, if the following holds:  For any time $(s,t)$, $s\in \{ 0,\ldots,r \}$ and $t\in [g_s]$, if $\{D_1,\ldots,D_h\}$ are all clauses from the corresponding QCNF that become unit under the assignment $\mathcal{T}[s,t-1]$ with literals $\ell_1,\ldots,\ell_h$, the next propagated literal has to be one of the $\ell_i$ together with $D_i$ as antecedent clause. If one of the $\ell_i$ is $\bot$, then we have to choose this $\ell_i$. I.e., conflicts have higher priority.
	\end{defi}

	The next definition presents the main framework for the whole paper. After having defined trails in a general sense, we want to specify the way a trail can be generated during a QCDCL run. We will give the notion of QCDCL-based proofs consisting of three components: the naturally created trails, the clauses we learned from each trail, and the proof of each learned clause.

	\begin{defi}[QCDCL proof systems]\label{DefinitionQCDCLrefutation}
		Let $\Phi=\mathcal{Q}\cdot \phi$ be a QCNF in $n$ variables.  We call a triple of sequences
		\begin{align*}
		\iota=({(\mathcal{T}_1,\ldots,\mathcal{T}_m)},{(C_1,\ldots,C_m)},{(\pi_1,\ldots,\pi_m)})
		\end{align*}
		a $\qcdcl^P_R$ \emph{proof} from $\Phi$ of a clause $C$  for $P\in \{ \lo,\aso,\asro,\ao \}$ and $R\in \{ \ar,\nr \}$, if for all $i\in [m]$ the trail $\mathcal{T}_i$ follows the policies $P$ and $R$ and uses the QCNF $\mathcal{Q}\cdot(\phi\cup \{ C_1,\ldots,C_{i-1} \})$, where $C_j\in \mathcal{L}_{\mathcal{T}_j}$ is a clause learnable  from $\mathcal{T}_j$ and $C_m=C$. Each $\pi_i$ is the derivation of the clause $C_i$ from $\mathcal{Q}\cdot(\phi\cup \{ C_1,\ldots,C_{i-1} \})$ as defined recursively in Definition~\ref{DefLearnableClauses}.  Note that all these trails need to run into a conflict in order to start clause learning. If $C=(\bot)$ we call $\iota$ a \emph{refutation}. 
		
		We also require that $\mathcal{T}_1$ is natural and for each $i\in\{ 2,\ldots,m \}$ there exist indices $(s,t)$ such that the following holds: 
		\begin{itemize}
			\item $\mathcal{T}_i[s,t]=\mathcal{T}_{i-1}[s,t]$.
			\item For each subtrail $\mathcal{T}_i[a,b]$ with $\mathcal{T}_i[s,t]\subseteq\mathcal{T}_i[a,b]$ and $\bot\not \in\mathcal{T}_i[a,b]$ let $D_1,\ldots,D_h$ be all the clauses in $\phi\cup \{ C_1,\ldots,C_{i-1} \}$ such that under the assignment $\mathcal{T}_i[a,b]$ these clauses get unit (under the policy~$R$) with corresponding literals $\ell_1,\ldots,\ell_h$. Then we have to propagate one of these literals next, i.e., $\ell_j\in \mathcal{T}_i[a,b+1]$ for some $j \in [h]$, and take the corresponding clause $D_j$ as antecedent. 
			\item In the situation above, if $\bot\in\{ \ell_1,\ldots,\ell_h \}$, then $\bot\in \mathcal{T}_i[a,b+1]$. I.e., we have to run into a conflict as soon as we find one. 
			
		\end{itemize}
		
		We call that \emph{backtracking} to $\mathcal{T}_i[s,t]$. Backtracking to $\mathcal{T}_i[0,0]$ is called \emph{restarting}.
		
		The size of such a proof $\iota$ is measured by $|\iota|:=\sum_{j=1}^{m}|T_j|\in\mathcal{O}(mn)$.

		The corresponding (refutational) proof system for false QCNFs is denoted $\qcdcl^P_R$. We will refer to these systems as QCDCL \emph{proof systems}. A trail $\mathcal{T}$ that follows the policies $P$ and $R$ is sometimes also called a \emph{$\qcdcl^P_R$ trail}.
		
	\end{defi}

	Note that the first trail $\mathcal{T}_1$ of each proof $\iota$ is always natural.
	
	In combination with $\ar$, the policy $\lo$  represents the original QCDCL algorithm without further modifications. The order policies $\aso$ and $\asro$ might seem slightly unintuitive at first sight. We will show later that these policies guarantee the learning of so-called asserting clauses (which will be defined in Definition~\ref{def:asserting}) in association with $\nr$ resp.\ $\ar$. Since $\aso$ (resp.\ $\asro$) will only unfold its impact in combination with $\nr$ (resp.\ $\ar$), we will not consider the systems $\asoar$ or $\asronr$.
	
	We will show later that $\pi_1,\ldots,\pi_m$ from $\iota$ in Definition~\ref{DefinitionQCDCLrefutation} are valid \ldqres proofs. In order to prove the correctness of these proofs, we will now argue that in proof systems with $\nr$ we cannot derive any tautologies, while with $\ar$ we can at most derive universal tautologies.

	\begin{exa}
		Let $\Phi:=\exists x \forall u \exists y,z \cdot (x\vee u\vee y)\wedge (\bar{y})\wedge (\bar{x}\vee z)\wedge (\bar{x}\vee \bar{z})$. We want to construct a natural trail $\mathcal{T}$ under the policies $\lo$ and $\ar$. Then $\mathcal{T}$ could represent a first trail in a $\loar$ proof from $\Phi$.
		
		Since we have a unit clause, we have to start by propagating $\bar{y}$. After a reduction of $u$ we can use $x\vee u \vee y$ to imply $x$. After this we can propagate $z$ via $\bar{x}\vee z$ and run into a conflict because of $\bar{x}\vee \bar{z}$.
		
		We now have constructed the following trail:
		
		\begin{align*}
		\mathcal{T}=(  \bar{y},x,z,\bot )
		\end{align*}
		
		Note that we did not need to make any decisions, hence swapping the decision policy $\lo$ would not change anything.
		
		The learnable clauses from $\mathcal{T}$ are, by definition, all clauses in the sequence
		
		\begin{align*}
		\mathcal{L}_\mathcal{T}=(C_{(0,4)},C_{(0,3)},C_{(0,2)},C_{(0,1)})
		\end{align*}
		with
		\begin{align*}
		C_{(0,4)}&=\ante_\mathcal{T}(\bot)=\bar{x}\vee \bar{z}\text{,}\\
		C_{(0,3)}&=\red\left(  C_{(0,4)}\resop{z}  \red(\ante_{\mathcal{T}}(z))     \right)=\red\left(    (\bar{x}\vee \bar{z})\resop{z} (  \bar{x}\vee z )      \right)=(\bar{x})\text{,}\\
		C_{(0,2)}&=\red\left(    C_{(0,3)}\resop{x} \red(\ante_{\mathcal{T}}(x))    \right)=\red\left(  (\bar{x})\resop{x}( x\vee u\vee y  )   \right)=u\vee y\text{,}\\
		C_{(0,1)}&=\red\left(    C_{(0,2)}\resop{y} \red(\ante_\mathcal{T}(\bar{y}))   \right)=\red\left( (u\vee y) \resop{y}  (\bar{y})      \right)=(\bot)\text{.}
		\end{align*}

		Let us now construct a natural trail $\mathcal{U}$ for the same formula $\Phi$ under the policies $\lo$ and $\nr$ (instead of $\ar$). Again, we have to start by propagating $\bar{y}$. However, since we are not allowed to use reduction in the propagation process, we do not gain a unit clause and therefore have to make a decision. Because of $\lo$, we are forced to decide $x$, which we will set to $1$ (we could also set $x$ to $0$, though). After this we can basically do the same steps as in $\mathcal{T}$, receiving the trail 
		\begin{align*}
		\mathcal{U}=( \bar{y};\mathbf{x},z,\bot )\text{.}
		\end{align*}
		
		The learnable clauses from $\mathcal{U}$ are
		\begin{align*}
		\mathcal{L}_\mathcal{U}=(D_{(1,2)},D_{(1,1)},D_{(0,1)})
		\end{align*}
		
		with 
		\begin{align*}
		D_{(1,2)}&=\red(\ante_\mathcal{U}(\bot))=\bar{x}\vee \bar{z}\text{,}\\
		D_{(1,1)}&=\red\left(       D_{(1,2)}\resop{z}\red( \ante_\mathcal{U}(z)        )\right)=\red\left(       (\bar{x}\vee \bar{z})\resop{z}(\bar{x}\vee z)     \right)=(\bar{x})\text{,}\\
		D_{(0,1)}&=D_{( 1,1)}\text{.}
		\end{align*}
		
		Note that $D_{(0,1)}$ is equal to $D_{(1,1)}$ since $D_{(1,1)}$ does not contain the literal $y$.		
	\end{exa}

	The above example could create the impression that using $\nr$ instead of $\ar$ weakens the underlying system. However, as we will show later, this is not the case.

	\begin{lem}\label{LemmaNoTautologies}
		It is not possible to create tautological clauses in any of the QCDCL proof systems with $\nr$ during the derivation of learnable clauses as described in Definition~\ref{DefLearnableClauses}. 
		If we use the policy $\ar$ instead, we are at least able to avoid any tautologies with existential literals.
	\end{lem}
	\begin{proof}
		Let $\mathcal{T}$ be a trail and $\mathcal{L}_\mathcal{T}$ be the sequence of learnable clauses as described in Definition~\ref{DefLearnableClauses}.   It suffices to concentrate on the case $C_{(i,j)}$, in particular
		\begin{align*}
		C_{(i,j)}=\red\left(    C_{(i,j+1)}\resop{p_{(i,j)}}\red(\ante(p_{(i,j)}))             \right)
		\end{align*}
		for $i\in \{0,\ldots,r\}$, $j\in[g_i-1]$. This is analogous to the case for $C_{(i,g_i)}$ (the case $C_{(r,g_r)}$ is trivial).

		Assume there is a literal $x$ with $\var(x)\neq \var(p_{(i,j)})$ such that $x\in C_{(i,j+1)}$ and $\bar{x}\in \red(\ante_\mathcal{T}(p_{(i,j)}))$. 
		If $\nr$ is chosen, we need $A|_{\mathcal{T}[i,j-1]}=(p_{(i,j)})$ for $A:=\ante_\mathcal{T}(p_{(i,j)})$ and therefore $x\in \mathcal{T}[i,j-1]$. 
		
		In case of $\ar$ we have $\red(A|_{\mathcal{T}[i,j-1]})=(p_{(i,j)})$. Because we can assume that $x$ is existential in this case, we also conclude $x\in \mathcal{T}[i,j-1]$ since existential literals cannot be reduced.
		
		On the other hand, we have $x\in C_{(i,j+1)}$, where $C_{(i,j+1)}$ is a learnable clause which is derived with the aid of antecedent clauses of literals occurring right of $p_{(i,j)}$ in the trail. In particular, we can find some $p_{(k,m)}$ right of $p_{(i,j)}$ in the trail with $x\in \ante_\mathcal{T}(p_{(k,m)})$.
		Because of $x\in \mathcal{T}[i,j-1]$, this gives a contradiction  since $\ante_\mathcal{T}(p_{(k,m)})$ must not become true before propagating $p_{(k,m)}$.	
	\end{proof}
	
	We have seen that systems with $\nr$ cannot contain tautologies in their extracted proofs. It remains to show that the derivation of tautological clauses in systems with $\ar$ fulfils the properties of \ldqres.
	
	\begin{prop}\label{PropSoundnessLongDistance}
		Let  $\Phi$ be a QCNF and let 
		\begin{align*}
		\mathcal{T}=(p_{(0,1)},\ldots, p_{(0,g_0)};\mathbf{d_{1}},p_{(1,1)},\ldots,p_{(1,g_{1})};\ldots; \mathbf{d_r},p_{(r,1)},\ldots ,p_{(r,g_r)}  )\text{.}
		\end{align*} be a trail under a QCDCL proof system with the policy $\ar$. Then the corresponding derivation of any learned clause is a valid \ldqres derivation.
	\end{prop}
	
	\begin{proof} Let
		\begin{align*}
		\mathcal{L_{\mathcal{T}}}:=(C_{(r,g_r)},\ldots,C_{(r,1)},\ldots,C_{(1,g_{1})},\ldots,C_{(1,1)},C_{(0,g_0)},\ldots,C_{(0,1)})
		\end{align*}
		be the sequence of learnable clauses. Because of Lemma \ref{LemmaNoTautologies} it remains to show that the derivation of clauses with universal tautologies are sound. 
		
		Assume otherwise. That means there are $i\in\{0,\ldots,r\}$, $j\in[g_i]$ and $u\in \var_\forall(\Phi)$ such that $\lv(u)<\lv(p_{(i,j)})$ and one of the following four cases holds, where we resolve over $p_{(i,j)}$:
		\begin{enumerate}
			\item $u\in C_{(i,j+1)}$ and $\bar{u}\in \ante_\mathcal{T}(p_{(i,j)})$\label{case1},
			\item $u\vee \bar{u}\in C_{(i,j+1)}$ and $\bar{u}\in \ante_\mathcal{T}(p_{(i,j)})$\label{case2},
			\item $u\in C_{(i,j+1)}$ and $u\vee \bar{u}\in \ante_\mathcal{T}(p_{(i,j)})$\label{case3},
			\item $u\vee \bar{u}\in C_{(i,j+1)}$ and $u\vee \bar{u}\in \ante_\mathcal{T}(p_{(i,j)})$\label{case4}.
		\end{enumerate}
		Consider case \ref{case1}. Since $u\in C_{(i,j+1)}$ there has to be a propagated literal $p_{(k,m)}$ right of $p_{(i,j)}$ in the trail such that $u\in \ante (p_{(k,m)})$. 
		In order to become unit, the $u$ in $\ante(p_{(k,m)})$ needed to vanish. We distinguish two cases:
		
		\underline{Case (i):} $\bar{u}$ was decided before $p_{(k,m)}$ was propagated, i.e., $\bar{u}\in \mathcal{T}[k,m-1]$.
		
		Then we have $u\not\in \mathcal{T}$ since each variable can occur at most once in $\mathcal{T}$. That means reducing $\bar{u}$ is the only way $\ante_\mathcal{T}(p_{(i,j)})$ could have become unit. But for the soundness of
		reduction we need $\lv(u)>\lv(p_{(i,j)})$. This gives a contradiction.

		\underline{Case (ii):} $u\in \ante_\mathcal{T}(p_{(k,m)})$ has vanished via reduction.
		
		Assume that $u$ was decided before $p_{(i,j)}$ was propagated, i.e., $u\in \mathcal{T}[i,j-1]$. But then $\ante(p_{(k,m)})$ would have become true under $\mathcal{T}[k,m-1]\supseteq\mathcal{T}[i,j-1]$. Therefore $\ante_\mathcal{T}(p_{(k,m)})$ could not have been used for unit propagation. 
		Thus $\bar{u}\in\ante_\mathcal{T}(p_{(i,j)})$ must have vanished via reduction, which implies $\lv(u)>\lv(p_{(i,j)})$, a contradiction.
		
		The same reasoning works for case \ref{case2}. Cases \ref{case3} and \ref{case4} are  easier since the only way for $u\vee \bar{u}$ to vanish in $\ante_\mathcal{T}(p_{(i,j)})$ is via reduction. Then we get the contradiction $\lv(u)>\lv(p_{(i,j)})$ as well.
		
		Also the same argumentation works if we consider $C_{(i+1,1)}$ and $\ante_\mathcal{T}(p_{(i,g_i)})$ instead of $C_{(i,j+1)}$ and $\ante_\mathcal{T}(p_{(i,j)})$.
	\end{proof}
	We combine the two results above to an argument of soundness of the defined QCDCL proof systems.
	
	\begin{thm}\label{TheoremSystemsSimulatedByQRes}
		All defined QCDCL proof systems are p-simulated by \ldqres and the systems with $\nr$ are even p-simulated by \qres. In fact,  for a proof $\iota$ of a clause $C$ in the   QCDCL system we can get a \ldqres (resp.\ \qres) proof $\pi$ with $|\pi|\in \mathcal{O}(|\iota|)$.
		
		In particular, all defined QCDCL proof systems are sound.
	\end{thm}

	\begin{proof} This follows directly from Lemma \ref{LemmaNoTautologies} and Proposition \ref{PropSoundnessLongDistance}. Let 
		\begin{align*}
		\iota=((\mathcal{T}_1,\ldots,\mathcal{T}_m),(C_1,\ldots,C_m),(\pi_1,\ldots,\pi_m))
		\end{align*}
		be a proof of a clause $C$ in the   QCDCL system. We get a \ldqres (resp.\ \qres) proof by sticking together the proofs $\pi_1,\ldots,\pi_m$ from $\iota$.
	\end{proof}

	Next we  introduce \emph{asserting learning schemes}. These are commonly used in practice since they guarantee a kind of progression in a run. These learning schemes are important to prevent a trail from backtracking too often (we will discuss this later). 
	
	\begin{defi}[asserting clauses and asserting learning schemes] \label{def:asserting}	
		Let $\Phi:=\mathcal{Q}\cdot \phi$ be a QCNF in any of the defined QCDCL systems. Let 
		\begin{align*}
		\mathcal{T}=(p_{(0,1)},\ldots, p_{(0,g_0)};\mathbf{d_{1}},p_{(1,1)},\ldots,p_{(1,g_{1})};\ldots; \mathbf{d_r},p_{(r,1)},\ldots ,p_{(r,g_r)} =\bot )
		\end{align*} 
		be a trail which follows the corresponding policies and $\mathcal{L}_{\mathcal{T}}$ the sequence of learnable clauses. A nonempty clause $ C\in \mathcal{L}_{\mathcal{T}}$ is called an \emph{asserting clause}, if it becomes unit after backtracking, i.e., there exists a time $(s,t)$ with $s\in \{  0,\ldots,r-1 \}$ and $t\in [g_s]$ such that $C|_{\mathcal{T}[s,t]}$ is a unit clause under the corresponding system.

		Let  $\mathbb{T}$ be the set of trails $\mathcal{T}$ for $\Phi$ such that $\bot\in \mathcal{T}$.  
		A \emph{learning scheme} $\xi$ is a map with domain $\mathbb{T}$, which maps each $\mathcal{T}$ to a clause $\xi(  \mathcal{T})\in \mathcal{L}_{\mathcal{T}}$.

		A learning scheme $\xi$ is called \emph{asserting} if it maps to asserting clauses or $(\bot)$ as long as $\mathcal{L}_{\mathcal{T}}$ contains such.
		
	\end{defi}

	\begin{rem} \label{rem:no-asserting}
		It is not guaranteed that we will always find asserting clauses for our trails. For example consider the false QCNF $\forall u\exists x\cdot (u\vee x)\wedge (u\vee \bar{x})\wedge(\bar{u}\vee x)\wedge(\bar{u}\vee \bar{x})$ and the trail $\mathcal{T}=(\textbf{x};\textbf{u},\bot)$ under the system $\aonr$. We can only learn the clause $(\bar{u}\vee \bar{x})$, which is non-unit under $\mathcal{T}[0,0]=\emptyset$. 
	\end{rem}
	Therefore despite allowing any decision order, we still need some restrictions in order to be able to learn asserting clauses.
	
	Next we show that asserting learning schemes return new clauses that have not been learned before. 
	\begin{lem}\label{LemmaNewClauseAsserting}
		Let $\mathcal{T}$ be a natural trail for a QCNF $\Phi=\mathcal{Q}\cdot \phi$ in any of the defined QCDCL systems and $A\in \mathcal{L}_\mathcal{T}$ be an asserting clause. Then we have $A\not\in \phi$.
	\end{lem}	
	\begin{proof} Let $\mathcal{T}$ be the trail 
		\begin{align*}
		\mathcal{T}=(p_{(0,1)},\ldots, p_{(0,g_0)};\mathbf{d_{1}},p_{(1,1)},\ldots,p_{(1,g_{1})};\ldots; \mathbf{d_r},p_{(r,1)},\ldots ,p_{(r,g_r)}=\bot  )\text{.}
		\end{align*} 
		
		Since we learned an asserting clause, we obviously have $r>0$ by definition.

		Assume that $A\in \phi$. Then there is a time $(s,t)$ with $s\in \{ 0,\ldots,r-1 \}$ and $t\in[g_s]$ such that $A|_{\mathcal{T}[s,t]}$ is a unit clause with a literal $\ell$. This literal is either existential or $\bot$. Because of Definition \ref{DefNaturalTrails} we are forced to propagate $\ell$ not later than level $s$. Since $\mathcal{T}$ has run into a conflict at level $r$, we could not get a conflict before level $r$, hence $\ell\neq \bot$. Therefore $\ell=p_{(s,b)}\in \mathcal{T}$ for some $b\in[g_s]$. However, this is a contradiction because we need $\ell\in A$, which is only possible if $\bar{\ell}\in \mathcal{T}$ by definition of clause learning.
	\end{proof}	
	
	\begin{rem}
		In general, the above result is not true for arbitrary trails. For example, let $\mathcal{T}=(\mathbf{\bar{x}_1};\mathbf{\bar{x}_2},\bot)$ be a trail for the false QCNF $\exists x_1,x_2\forall u\cdot (x_1\vee x_2)\wedge (u)$. Then $x_1\vee x_2$ is an asserting learnable clause, which is already an axiom.
	\end{rem}

	Now we know that learning schemes always return new clauses.  We also need to examine the circumstances under which we are actually able to learn such asserting clauses.

	\begin{lem}\label{LemmaAssertingClauseASONR}
		Let $\Phi=\mathcal{Q}\cdot \phi$ be a QCNF and let 
		\begin{align*}
		\mathcal{T}=(p_{(0,1)},\ldots, p_{(0,g_0)};\mathbf{d_{1}},p_{(1,1)},\ldots,p_{(1,g_{1})};\ldots; \mathbf{d_r},p_{(r,1)},\ldots ,p_{(r,g_r)} =\bot )
		\end{align*} be a trail under the policies $\aso$ and $\nr$. If $(\bot)\not\in \mathcal{L_{\mathcal{T}}}$, then there exists an asserting clause $D\in \mathcal{L}_\mathcal{T}$.

	\end{lem}

	\begin{proof} Consider the sequence of learnable clauses:
		\begin{align*}
		\mathcal{L_{\mathcal{T}}}:=(C_{(r,g_r)},\ldots,C_{(r,1)},\ldots,C_{(1,g_{1})},\ldots,C_{(1,1)},C_{(0,g_0)},\ldots,C_{(0,1)})
		\end{align*}
		The learning scheme 
		that maps each trail to the rightmost clause in $\mathcal{L_{\mathcal{T}}}$ is asserting. In particular we learn a clause $D:=\red(C)$ such that $C$ is a subclause of $\bar{d}_1\vee\ldots\vee\bar{d}_r$. Note that, because of $\nr$, the clause $C$ does not contain further (universal) literals that were reduced in the trail.  Suppose that $D\neq (\bot)$. Write   $D=\bar{d}_{h_1}\vee \ldots \vee \bar{d}_{h_a}$ for some $h_1<\ldots<h_a$. Because of $\aso$ the last literal $d_{h_a}$ will always be existential. Then we can backtrack to the time $(h_{a-1},0)$ (resp.\ $(0,0)$ if $a=1$) and get $D|_{\mathcal{T}[h_{a-1},0]}=(\bar{d}_{h_a})$.
	\end{proof}
	
	The above result is not true for  QCDCL systems with the policies $\aso$ and $\ar$. Consider the following counterexample: $\Phi:=\exists x\forall u \exists z\cdot (\bar{x}\vee u\vee \bar{z})$ and the trail $\mathcal{T}=(\textbf{x};\textbf{z},\bot)$. We can only learn the clause $(\bar{x}\vee u\vee \bar{z})$. But since this is already an axiom, this clause cannot become unit at an earlier level.

	However, the next lemma shows that under the policy $\asro$ we can always learn asserting clauses.

	\begin{lem}\label{LemmaAssertingClauseASROAR}
		Let $\Phi=\mathcal{Q}\cdot \phi$ be a QCNF and let 
		\begin{align*}
		\mathcal{T}=(p_{(0,1)},\ldots, p_{(0,g_0)};\mathbf{d_{1}},p_{(1,1)},\ldots,p_{(1,g_{1})};\ldots; \mathbf{d_r},p_{(r,1)},\ldots ,p_{(r,g_r)}=\bot  )
		\end{align*} be a trail under the policies $\asro$ and $\ar$. If $(\bot)\not\in \mathcal{L_{\mathcal{T}}}$, then there exists an asserting clause $D\in \mathcal{L}_\mathcal{T}$.
		
	\end{lem}
	\begin{proof}  Consider the sequence of learnable clauses:
		\begin{align*}
		\mathcal{L_{\mathcal{T}}}:=(C_{(r,g_r)},\ldots,C_{(r,1)},\ldots,C_{(1,g_{1})},\ldots,C_{(1,1)},C_{(0,g_0)},\ldots,C_{(0,1)})
		\end{align*}
		If $r=0$ or if all decision literals are universal, we can just take the rightmost clause in $\mathcal{L_{\mathcal{T}}}$, which is $(\bot)$. Hence we can assume	 that $r>0$ and there exists at least one existential decision literal.
		
		Let $k\in [r]$ be maximal with respect that $\bar{d}_k$ is contained in some clause from $\mathcal{L}_\mathcal{T}$. This must exist since otherwise we could resolve over all propagation literals $p_{(i,j)}$ and reduce all universal literals during the learning process. Then we would be able to learn $(\bot)$.
		
		Intuitively, $d_k$ is the last existential decision that contributed to the conflict.  Let $p_{(\ell,m)}$ be the next propagated literal right of $d_k$ in $\mathcal{T}$ (this does not necessarily have to be $p_{(k,1)}$). Set $D:=C_{(\ell,m)}$.
		
		We claim that this clause $D$ is asserting. Let us learn this clause and backtrack to $\mathcal{T}[k-1,g_{k-1}]$. It is easy to see that for all existential literals $y\in D\backslash\{ \bar{d}_k \}$ we need $\bar{y}\in\mathcal{T}[k-1,g_{k-1}]$. All universal variables $z$ with $\lv_\Phi(z)<\lv_\Phi(d_k)$ are assigned earlier than $d_k$ in $\mathcal{T}$ because of the policy $\asro$.
		
		Now consider $E:=D|_{\mathcal{T}[k-1,g_{k-1}]}$.  The only type of literals that can occur in $E$, aside from $\bar{d_k}$, are universal literals $u$ with $\lv_\Phi(u)>\lv_\Phi(d_k)$. 
		Suppose there are such literals $u_1,\ldots,u_m$ with $C=\bar{d}_k\vee u_1\vee\ldots\vee u_m$ and $\lv_\Phi(d_k)<\lv_\Phi(u_i)$ for all $i\in [m]$. But then we can conclude $\red(E)=(\bar{d}_k)$ and therefore $E$ is a unit clause.
	\end{proof}
	
	Now that we have clarified how to gain asserting clauses, we can finally prove the completeness of all systems.
	
	\begin{thm} \label{thm:qcdcl-complete}
		All defined QCDCL proof systems are complete.
		
	\end{thm}	
	\begin{proof} We will  concentrate on the systems $\loar$ and $\lonr$ since all the other systems are just strengthenings of these two proof systems.
		
		Let $\Phi=\mathcal{Q}\cdot \phi$ be a false QCNF over the set of variables $V$. W.l.o.g.\ we write the prefix  as 
		\begin{align*}
		\mathcal{Q}=	\exists X_1 \forall U_1 \exists X_2 \forall U_2 \ldots \exists X_m \forall U_m
		\end{align*}
		with $X_1\uplus\ldots\uplus X_m\uplus U_1\uplus \ldots \uplus U_m=V$. This is possible if we allow empty sets.
		
		Now because $\Phi$ is false, there exists a winning strategy for the $\forall$-player. That means we can find functions $f_i:\:\langle X_1\cup \ldots \cup X_i\rangle \rightarrow \langle U_i \rangle$ such that $\phi$ gets falsified under any assignment 
		\begin{align*}
		\sigma_1\cup f_1(\sigma_1)\cup \sigma_2 \cup f_2(\sigma_1\cup \sigma_2)\cup\ldots \cup \sigma_m \cup f_m(\sigma_1\cup \ldots \cup \sigma_m)\text{.}
		\end{align*}
		We can now construct a trail for $\Phi$ in the respective system. All the unit propagations can be done automatically. When we have to make decisions, we distinguish two cases:
		
		\underline{Case 1:} We have to decide an existential variable $x$. 
		
		We can choose an arbitrary polarity for this decision of $x$ (e.g. set all variables to $1$).
		
		\underline{Case 2:} We have to decide a universal variable.
		
		Suppose we have to handle the variable $u\in U_i$. Since our trail follows the policy $\lo$, all variables from $X_1\cup \ldots \cup X_i$ are already assigned. Let $\sigma$ be the corresponding assignment. Then we decide $u$ in the same polarity as it occurs in $f_i(\sigma)$.\\
		
		After making the decisions we continue with unit propagation as usual. The trail that is generated this way represents an assignment as it gets created in the game between the two players. The $\forall$-player was able to use the winning strategy, therefore we falsify the matrix $\phi$ at some point and run into a conflict. After this we start clause learning where we can always learn an asserting clause $C$ by Lemma~\ref{LemmaAssertingClauseASONR} (resp.\ Lemma~\ref{LemmaAssertingClauseASROAR}) until we learn the empty clause.  By Lemma \ref{LemmaNewClauseAsserting} we conclude that $C$ is in fact a new clause. We add $C$ to $\phi$, restart (i.e., backtrack to $(0,0)$) and start again. Since there are only finitely many clauses with variables in $V$, this process will end after finitely many runs. In the last run we have to learn the empty clause, hence we created a (possibly exponential-size) refutation 
		\begin{align*}
		\iota=((\mathcal{T}_1,\ldots,\mathcal{T}_n),(C_1,\ldots,C_n=(\bot)),(\pi_1,\ldots,\pi_n))
		\end{align*}
		of $\Phi$ in the QCDCL system.
	\end{proof}
	
	Proving that QCDCL decisions do not necessarily need to follow the order of quantification (as is done in practical QCDCL with policy $\lo$), might be a somewhat surprising discovery. It seems to us that inside the QBF community there is the wide-spread belief that following the quantification order in decisions is needed for soundness (cf.\ e.g.\ \cite{LonsingEG13,ZM02,GiunchigliaMN09}).\footnote{In fact we thought so too, prior to this paper.} While this is true for QDPLL \cite{GiunchigliaMN09,CadoliGS98},\footnote{The fact that the earlier QDPLL algorithm \cite{CadoliGS98} needs to obey the quantifier order might have been one reason why this policy was adopted in QCDCL as well \cite{ZM02}. Additionally, following prefix order in decisions allows to learn asserting constraints, which otherwise cannot be always guaranteed (cf.\ \cite{BPB22}).} it is actually not needed in QCDCL: the quantification order is immaterial for the decisions as long as the quantification order is correctly taken into account when deriving learned clauses.\footnote{We note, however, that the approach of \emph{dependency learning} \cite{PeitlSS19} starts with an empty set of dependency conditions (cf.\ \cite{SlivovskyS16,BeyersdorffBCSS19} for background on dependencies) and incrementally learns new dependencies. As decisions only need to respect the learned dependencies, they can initially be made out of order \cite{PeitlSS19}.}
	Hence our theoretical work also opens the door towards \emph{new solving approaches in practice} (cf.\ the discussion in the 
	concluding Section~\ref{sec:conclusion}).

	\section{Separating classic QCDCL and Q-resolution}
	\label{sec:separation-qcdcl-qres}
	
	In this section we will show that classic \qcdcl (i.e., $\loar$) and \qres are incomparable. This requires exponential separations in both directions between the two systems. This also motivates that when searching for a QCDCL algorithm that exploits the full power of \qres (a topic we will address in Section~\ref{sec:qres-characterisation}), we will have to modify the policies in $\loar$. 
	We start the separation by introducing QBFs that will turn out to be easy for $\loar$, but hard for \qres.
	
	\begin{defiC}[\cite{BeyersdorffCJ19}]
		The QCNF $\mathtt{QParity}_n$ consists of the prefix $\exists x_1 \ldots x_n\forall z\exists t_2 \ldots t_n$ and the matrix
		\begin{align*}
		&x_1\vee x_2\vee \bar{t}_2,\
		x_1\vee \bar{x}_2\vee t_2,\
		\bar{x}_1\vee x_2 \vee t_2,\
		\bar{x}_1\vee \bar{x}_2\vee \bar{t}_2,\\
		&x_i\vee t_{i-1}\vee \bar{t}_i,\
		x_i\vee \bar{t}_{i-1}\vee t_i,\
		\bar{x}_i\vee t_{i-1}\vee t_i,\
		\bar{x}_i\vee \bar{t}_{i-1}\vee \bar{t}_i,\\
		&t_n\vee z,\
		\bar{t}_n\vee \bar{z},
		\end{align*} 
		for $i\in \{ 2,\ldots,n \}$.
	\end{defiC}

	One can interpret $\mathtt{QParity}_n$ as follows: Each $t_i$ symbolizes the partial sum $x_1\oplus x_2 \oplus\ldots \oplus x_i$. The last two clauses can only be true if $t_n$, which represents $x_1 \oplus x_2 \oplus\ldots \oplus x_n$, is neither $0$ nor~$1$. Hence this formula is obviously false. It was shown in \cite{BeyersdorffCJ19} that $\mathtt{QParity}_n$ needs exponential-size \qres (and even \textsf{QU-resolution}) refutations. 
	
	We will demonstrate that $\mathtt{QParity}_n$ is in fact easy for $\loar$, i.e., for  classical \qcdcl.
	
	\begin{prop}\label{PropQParityHardForQCDCL}
		$\mathtt{QParity}_n$ has polynomial-size $\loar$ refutations.
	\end{prop}	
	\begin{proof} We construct the trails $(\mathcal{T}_n,\mathcal{U}_n,\mathcal{T}_{n-1},\mathcal{U}_{n-1},\ldots,\mathcal{T}_2,\mathcal{U}_2,\mathcal{T}_1,\mathcal{U}_1)$. 
		
		Let $\mathcal{T}_n$ be the trail
		\begin{align*}
		\mathcal{T}_n=( \mathbf{\bar{x}_1};
		\mathbf{\bar{x}_2},\bar{t}_2;
		\mathbf{\bar{x}_3},\bar{t}_3;
		\ldots;
		\mathbf{\bar{x}_{n-1}},\bar{t}_{n-1};
		\mathbf{\bar{x}_{n}},\bar{t}_{n},\bot)
		\text{}
		\end{align*}
		with antecedent clauses
		\begin{align*}
		\ante_{\mathcal{T}_n}(\bar{t}_2)&={x}_1\vee x_2\vee \bar{t}_2\text{,}\\
		\ante_{\mathcal{T}_n}(\bar{t}_i)&=x_{i}\vee {t}_{i-1}\vee \bar{t}_i\text{ for $i=3,\ldots,n$,}\\
		\ante_{\mathcal{T}_n}(\bot)&=t_n\vee z\text{.}
		\end{align*}
		After resolving over $\bar{t}_n$ we can derive $C_n:= x_n\vee t_{n-1}\vee z$ as a learnable clause and backtrack to $\mathcal{T}_n[n-2,1]$.
		
		Further, let $\mathcal{U}_n$ be the trail 
		\begin{align*}
		\mathcal{U}_n=( \mathbf{\bar{x}_1};
		\mathbf{\bar{x}_2},\bar{t}_2;
		\mathbf{\bar{x}_3},\bar{t}_3;
		\ldots;
		\mathbf{\bar{x}_{n-2}},\bar{t}_{n-2};
		\mathbf{{x}_{n-1}},{t}_{n-1};
		\mathbf{\bar{x}_{n}},{t}_{n},\bot)
		\text{}
		\end{align*}
		with antecedent clauses 
		\begin{align*}
		\ante_{\mathcal{U}_n}(\bar{t}_2)&={x}_1\vee x_2\vee \bar{t}_2\text{,}\\
		\ante_{\mathcal{U}_n}(\bar{t}_i)&=x_{i}\vee {t}_{i-1}\vee \bar{t}_i\text{ for $i=3,\ldots,n-2$,}\\
		\ante_{\mathcal{U}_n}({t}_{n-1})&=\bar{x}_{n-1}\vee {t}_{n-2}\vee {t}_{n-1}\text{,}\\
		\ante_{\mathcal{U}_n}({t}_{n})&={x}_{n}\vee \bar{t}_{n-1}\vee {t}_{n}\text{,}\\
		\ante_{\mathcal{U}_n}(\bot)&=\bar{t}_n\vee \bar{z}\text{.}
		\end{align*}
		Symmetrically to $\mathcal{T}_n$, we can learn the clause $D_n:=x_n\vee \bar{t}_{n-1}\vee \bar{z}$ by resolving over $t_n$ and backtrack up to $\mathcal{U}_n[n-2,1]$.
		
		We continue with the construction of $\mathcal{T}_{n-1}$.
		\begin{align*}
		\mathcal{T}_{n-1}=( \mathbf{\bar{x}_1};
		\mathbf{\bar{x}_2},\bar{t}_2;
		\mathbf{\bar{x}_3},\bar{t}_3;
		\ldots;
		\mathbf{\bar{x}_{n-2}},\bar{t}_{n-2};
		\mathbf{\bar{x}_{n-1}},\bar{t}_{n-1},x_n,t_n,\bot)
		\text{}
		\end{align*}
		with antecedent clauses
		\begin{align*}
		\ante_{\mathcal{T}_{n-1}}(\bar{t}_2)&={x}_1\vee x_2\vee \bar{t}_2\text{,}\\
		\ante_{\mathcal{T}_{n-1}}(\bar{t}_i)&=x_{i}\vee {t}_{i-1}\vee \bar{t}_i\text{ for $i=3,\ldots,n-1$,}\\
		\ante_{\mathcal{T}_{n-1}}(x_n)&=C_n=x_n\vee t_{n-1}\vee z\text{,}\\
		\ante_{\mathcal{T}_{n-1}}({t}_n)&=\bar{x}_n\vee t_{n-1}\vee t_n\text{,}\\
		\ante_{\mathcal{T}_{n-1}}(\bot)&=\bar{t}_n\vee \bar{z}\text{.}
		\end{align*}
		After resolving over $t_n$, $x_n$ and finally $\bar{t}_{n-1}$ we get the learned clause $C_{n-1}:=x_{n-1}\vee t_{n-2}\vee z\vee \bar{z}$ and backtrack up to $\mathcal{T}_{n-1}[n-3,1]$.

		As before, we can symmetrically create the next trail $\mathcal{U}_{n-1}$.
		\begin{align*}
		\mathcal{U}_{n-1}=( \mathbf{\bar{x}_1};
		\mathbf{\bar{x}_2},\bar{t}_2;
		\mathbf{\bar{x}_3},\bar{t}_3;
		\ldots;
		\mathbf{\bar{x}_{n-3}},\bar{t}_{n-3};
		\mathbf{{x}_{n-2}},{t}_{n-2};
		\mathbf{\bar{x}_{n-1}},{t}_{n-1},x_n,\bar{t}_n,\bot)
		\text{}
		\end{align*}
		with antecedent clauses 
		\begin{align*}
		\ante_{\mathcal{U}_{n-1}}(\bar{t}_2)&={x}_1\vee x_2\vee \bar{t}_2\text{,}\\
		\ante_{\mathcal{U}_{n-1}}(\bar{t}_i)&=x_{i}\vee {t}_{i-1}\vee \bar{t}_i\text{ for $i=3,\ldots,n-3$,}\\
		\ante_{\mathcal{U}_{n-1}}({t}_{n-2})&=\bar{x}_{n-2}\vee {t}_{n-3}\vee {t}_{n-2}\text{,}\\
		\ante_{\mathcal{U}_{n-1}}({t}_{n-1})&={x}_{n-1}\vee \bar{t}_{n-2}\vee {t}_{n-1}\text{,}\\
		\ante_{\mathcal{U}_{n-1}}(x_n)&=D_n=x_n\vee \bar{t}_{n-1}\vee \bar{z}\text{,}\\
		\ante_{\mathcal{U}_{n-1}}(\bar{t}_n)&=\bar{x}_n\vee \bar{t}_{n-1}\vee \bar{t}_n\text{,}\\
		\ante_{\mathcal{U}_{n-1}}(\bot)&={t}_n\vee {z}\text{.}
		\end{align*}
		By resolving over $\bar{t}_n$, $x_n$ and $t_{n-1}$ we can derive $D_{n-1}:=x_{n-1}\vee \bar{t}_{n-2} \vee z\vee \bar{z}$ and backtrack up to $\mathcal{U}_{n-1}[n-3,1]$.

		We now describe the general step. Let $j\in \{ 3,\ldots,n-2 \}$ and suppose we already learned the clauses $C_n,D_n,C_{n-1},D_{n-1},\ldots,C_{j+1},D_{j+1}$ with the aid of the earlier trails $\mathcal{T}_n,\mathcal{U}_n,\ldots,\mathcal{T}_{j+1},\mathcal{U}_{j+1}$, where $C_k:= x_{k}\vee t_{k-1}\vee z\vee \bar{z}$ and $D_k:=x_k\vee \bar{t}_{k-1}\vee z\vee\bar{z}$ for $k\in \{ j+1,\ldots,n-1 \}$. Then we can construct $\mathcal{T}_j$ as follows:
		\begin{align*}
		\mathcal{T}_j=( \mathbf{\bar{x}_1};
		\mathbf{\bar{x}_2},\bar{t}_2;
		\mathbf{\bar{x}_3},\bar{t}_3;
		\ldots;
		\mathbf{\bar{x}_{j-1}},\bar{t}_{j-1};
		\mathbf{\bar{x}_j},  \bar{t}_j,{x}_{j+1},{t}_{j+1},{x}_{j+2},\bar{t}_{j+2},\ldots,{x}_n,{t}^{e}_n, \bot  )\text{}
		\end{align*}
		with $e\in \{0,1\}$ such that $e\equiv n-j$ (mod $2$) and antecedent clauses
		\begin{align*}
		\ante_{\mathcal{T}_{j}}(\bar{t}_2)&={x}_1\vee x_2\vee \bar{t}_2\text{,}\\
		\ante_{\mathcal{T}_{j}}(\bar{t}_i)&=x_{i}\vee {t}_{i-1}\vee \bar{t}_i\text{ for $i=3,\ldots,j$,}\\
		\ante_{\mathcal{T}_{j}}(x_{j+1})&=C_{j+1}=x_{j+1}\vee t_{j}\vee z\vee \bar{z}\text{,}\\
		\ante_{\mathcal{T}_{j}}({t}_{j+1})&=  \bar{x}_{j+1}\vee t_j\vee t_{j+1} \text{,}\\
		\ante_{\mathcal{T}_{j}}(x_{j+2})&=D_{j+2}=x_{j+2}\vee \bar{t}_{j+1}\vee z\vee \bar{z}\text{,}\\
		\ante_{\mathcal{T}_{j}}(\bar{t}_{j+2})&=  \bar{x}_{j+2}\vee \bar{t}_{j+1}\vee \bar{t}_{j+2} \text{,}\\
		\vdots&\\
		\ante_{\mathcal{T}_{j}}(x_{n})&=\left\{ \begin{array}{ll}
		C_n&=x_n\vee t_{n-1}\vee \bar{z}\text{, if $n-j$ is odd,}\\
		D_n&=x_n\vee \bar{t}_{n-1}\vee {z}\text{, if $n-j$ is even,}\\          \end{array}\right.\\
		\ante_{\mathcal{T}_{j}}(t_n^e)&=  \bar{x}_n\vee t_{n-1 }^{e}  \vee t_n^e \text{,}\\
		\ante_{\mathcal{T}_{j}}(\bot)&=t_n^{1-e}\vee {z}^{1-e}\text{.}
		\end{align*}
		We resolve over $t_n^e,x_n,\ldots,t_{j+1},x_{j+1}$ and $\bar{t}_j$, derive $C_j:=x_j\vee t_{j-1}\vee z\vee \bar{z}$ and backtrack up to $\mathcal{T}_j[j-2,1]$.
		
		In a similar way we will create $\mathcal{U}_j$:
		\begin{align*}
		\mathcal{U}_j=&( \mathbf{\bar{x}_1};
		\mathbf{\bar{x}_2},\bar{t}_2;
		\mathbf{\bar{x}_3},\bar{t}_3;
		\ldots;
		\mathbf{\bar{x}_{j-2}},\bar{t}_{j-2};
		\mathbf{{x}_{j-1}},{t}_{j-1};
		\mathbf{\bar{x}_j},  {t}_j,{x}_{j+1},\bar{t}_{j+1},{x}_{j+2},{t}_{j+2},\ldots,{x}_n,\\&{t}^{f}_n, \bot  )
		\text{}
		\end{align*}
		with $f\in \{0,1\}$ such that $f\equiv n-j+1$ (mod $2$) and antecedent clauses 
		\begin{align*}
		\ante_{\mathcal{U}_{j}}(\bar{t}_2)&={x}_1\vee x_2\vee \bar{t}_2\text{,}\\
		\ante_{\mathcal{U}_{j}}(\bar{t}_i)&=x_{i}\vee {t}_{i-1}\vee \bar{t}_i\text{ for $i=3,\ldots,j-2$,}\\
		\ante_{\mathcal{U}_{j}}({t}_{j-1})&=\bar{x}_{j-1}\vee {t}_{j-2}\vee {t}_{j-1}\text{,}\\
		\ante_{\mathcal{U}_{j}}({t}_{j})&={x}_{j}\vee \bar{t}_{j-1}\vee {t}_{j}\text{,}\\
		\ante_{\mathcal{U}_{j}}(x_{j+1})&=D_{j+1}=x_{j+1}\vee \bar{t}_{j}\vee z\vee \bar{z}\text{,}\\
		\ante_{\mathcal{U}_{j}}(\bar{t}_{j+1})&=  \bar{x}_{j+1}\vee \bar{t}_j\vee \bar{t}_{j+1} \text{,}\\
		\ante_{\mathcal{U}_{j}}(x_{j+2})&=C_{j+2}=x_{j+2}\vee {t}_{j+1}\vee z\vee \bar{z}\text{,}\\
		\ante_{\mathcal{U}_{j}}({t}_{j+2})&=  \bar{x}_{j+2}\vee {t}_{j+1}\vee \bar{t}_{j+2} \text{,}\\
		\vdots&\\
		\ante_{\mathcal{U}_{j}}(x_{n})&=\left\{ \begin{array}{ll}
		C_n&=x_n\vee t_{n-1}\vee \bar{z}\text{, if $n-j$ is even,}\\
		D_n&=x_n\vee \bar{t}_{n-1}\vee {z}\text{, if $n-j$ is odd,}\\          \end{array}\right.\\
		\ante_{\mathcal{U}_{j}}(t_n^f)&=  \bar{x}_n\vee t_{n-1 }^{f}  \vee t_n^f \text{,}\\
		\ante_{\mathcal{U}_{j}}(\bot)&=t_n^{1-f}\vee {z}^{1-f}\text{.}
		\end{align*}
		Resolving over $t_n^f,x_n,\ldots\bar{t}_{j+1},x_{j+1}$ and $t_j$ gives us $D_j:=x_j\vee \bar{t}_{j-1}\vee z\vee \bar{z}$. We backtrack to $\mathcal{U}_j[j-2,1]$.

		We end the refutation with the following four trails whose antecedent clauses are almost as before:
		\begin{align*}
		\mathcal{T}_2=( \mathbf{\bar{x}_1};
		\mathbf{\bar{x}_2},\bar{t}_2,x_3,t_3,x_4,\bar{t}_4,\ldots,\bot )\text{, }
		\end{align*}
		from which we learn $C_2:=x_1\vee x_2$, and
		\begin{align*}
		\mathcal{U}_2=( &\mathbf{\bar{x}_1},{x}_2,{t}_2,x_3,\bar{t}_4,\ldots,\bot)
		\text{, }
		\end{align*}
		where we can derive the unit clause $(x_1)$. We continue with 
		\begin{align*}
		\mathcal{T}_1=( {{x}_1};
		\mathbf{\bar{x}_2},{t}_2,x_3,\bar{t}_3,x_4,{t}_4,\ldots,\bot )\text{, }
		\end{align*}
		learn $(x_2)$, and can finally derive the empty clause $(\bot)$ via the last trail 
		\begin{align*}
		\mathcal{U}_1=( {{x}_1}, {{x}_2},\bar{t}_2,x_3,{t}_3,x_4,\bar{t}_4,\ldots,\bot )\text{. } \tag*{\qedhere}
		\end{align*}
	\end{proof}
	
	For the separation in the other direction we need CNFs that are hard for general \resolution. One of the famous examples is the pigeonhole principle, but also any other formula that is hard for propositional resolution would serve the purpose. 
	
	\begin{defi}
		The pigeonhole principle $\mathtt{PHP}_n^m$ is a propositional CNF consisting of the variables $x_{i,j}$, for $i\in [m]$ and $j\in [n]$, and the clauses
		\begin{align*}
		\bigvee_{k\in[n]}x_{i,k}\\
		\bar{x}_{i_1,j}\vee \bar{x}_{i_2,j}
		\end{align*}
		for all $i,i_1,i_2\in [m]$, $i_1\neq i_2$ and $j\in [n]$.
	\end{defi}
	
	\begin{propC}[\cite{Hak85}]\label{PropHakenPHP}
		The CNFs $\mathtt{PHP}^{n+1}_n$ are unsatisfiable and require exponential-size \resolution refutations.	
	\end{propC}
	
	We embed $\mathtt{PHP}^{n+1}_n$ into a QCNF which we will call $\mathtt{Trapdoor}_n$. Intuitively, if we have chosen $\ar$, we are forced to reach a conflict in the propositionally hard formula $\mathtt{PHP}^{n+1}_n$. However, forbidding reduction by choosing the policy $\nr$ allows us to avoid the pigeonhole principle and instead derive a conflict in part that is easier to refute. 
	
	\begin{defi} \label{def:trapdoor}
		Let $\mathtt{PHP}_n^{n+1}$ be the set of clauses for the pigeonhole principle with parameters $n$ and $n+1$ in the variables $x_1,\ldots,x_{s_n}$. Let $\mathtt{Trapdoor}_n$ be the QCNF in the variables $x_1,\ldots,x_{s_n},y_1,\ldots,y_{s_n},u,t,w$ with the prefix
		\begin{align*}
		\exists y_1,\ldots,y_{s_n}\forall w\exists t, x_1,\ldots,x_{s_n}\forall u
		\end{align*}
		and the matrix
		\begin{align*}
		&\mathtt{PHP}_n^{n+1}(x_1,\ldots,x_{s_n})\\
		&\bar{y}_i\vee x_i \vee u,\ 
		y_i\vee \bar{x_i}\vee u\\
		&y_i\vee w\vee t, \
		y_i\vee w\vee \bar{t},\ 
		\bar{y}_i\vee w\vee t ,\ 
		\bar{y}_i\vee w\vee \bar{t}
		\end{align*}
		for $i=1,\ldots,s_n$.
	\end{defi}
	
	
	The next result shows the hardness for $\mathtt{Trapdoor}_n$ in $\loar$. It is clear that this hardness is directly caused by $\mathtt{PHP}^{n+1}_n$. However, it remains to prove that it is still retained by the embedding.
	
	\begin{prop}\label{PropTrapdoorHardForQCDCL}
		The QCNFs $\mathtt{Trapdoor}_n$ require exponential-size $\loar$ refutations.
	\end{prop}
	
	\begin{proof} Each $\loar$ trail for $\mathtt{Trapdoor}_n$ starts with some decisions of $y$ variables. Before reaching the variables $w$ and $t$, we can only have propagated $x$ or $y$ variables. Since $\mathtt{PHP}_n^{n+1}$ is unsatisfiable, we will reach a conflict before deciding $w$. The learned clauses will only contain $x$, $y$ and $u$ variables, hence the same situation will happen after restarting which leads to a refutation that does not use the last four types of axioms in $\mathtt{Trapdoor}_n$. 
		
		Therefore we have gained a \ldqres refutation $\pi''$ of $\mathtt{Trapdoor}_n$ which only makes use of the axioms
		\begin{align*}
		&\mathtt{PHP}_n^{n+1}(x_1,\ldots,x_{s_n})\\
		&\bar{y}_i\vee x_i \vee u \\
		&y_i\vee \bar{x_i}\vee {u}\text{.}
		\end{align*} 
		Define $\psi':=\{ \bar{y}_i\vee x_i \vee u,y_i\vee \bar{x_i}\vee {u}:\: i=1,\ldots,s_n  \}$.
		We can construct a \qres refutation $\pi'$ by reducing the $u$-variable right after any introduction of a clause from $\psi'$ and performing the same resolution steps as in $\pi''$ afterwards. Then $|\pi'|\in \mathcal{O}(|\pi''|)$.
		
		Now deleting any axiom of $\psi'$ from $\pi'$ results in a \resolution refutation $\pi$ of the unsatisfiable CNF
		\begin{align*}
		&\mathtt{PHP}_n^{n+1}(x_1,\ldots,x_{s_n})\\
		&\bar{y}_i\vee x_i \\
		&y_i\vee \bar{x_i}\text{}
		\end{align*}
		with $|\pi|\in \mathcal{O}(|\pi'|)$.
		
		Let $\psi=\{ \bar{y}_i\vee x_i , y_i\vee \bar{x_i}:\: i=1,\ldots,s_n \}$. Next we will create a \resolution refutation $\mu$ of $\mathtt{PHP}_n^{n+1}$.
		
		Let $\pi$ consist of the clauses $C'_1,\ldots,C'_m$. W.l.o.g.\ we can assume that none of these clauses are of the form $D\vee\bar{y}_i\vee x_i$ or $D\vee  y_i\vee \bar{x_i}$ for a non-empty  subclause $D$, because otherwise we can shorten the refutation by taking the corresponding axioms in $\psi$ instead.
		
		Also let $f$ be the function on clauses that replaces all occurrences of $y_i$ (resp.\ $\bar{y}_i$) in a clause with $x_i$ (resp.\ $\bar{x}_i$). We will show that the proof $\mu$ we get by deleting all clauses of $\pi$ contained in $\psi$ and replacing all other clauses $C$ by $f(C)$ is a correct \resolution refutation. 
		
		Let $C_\ell'$ be a clause in $\pi$. If $C_\ell'$ is an axiom, then either $C_\ell'\in \mathtt{PHP}_n^{n+1}$ and therefore $C_\ell:=f(C_\ell')=C_\ell'$, or $C_\ell'\in   \psi$, in which case we delete $C_\ell'$ (or replace it with a placeholder in $\mu$).
		
		If $C_\ell'$ was not an axiom, then we can find two parental clauses $C_j'$ and $C_k'$ of $C_\ell'$. We distinguish two cases.
		
		\underline{Case 1:} One of the clauses $C_j'$, $C_k'$ is from $\psi$.
		
		W.l.o.g. let $C_j'\in \psi$ and $C_k'\not\in \psi$. Note that it is not possible for both clauses to be contained in $\psi$. By induction we know that $C_k:=f(C_k')\in \mu$. Since $C_j'$ was deleted during the transition into $\mu$, we can only set $C_\ell:=C_k=f(C_k')$. Because resolving with clauses of $\psi$ just swap $x_i$-variable and $y_i$-variables, we immediately get $f(C_k')=f(C_\ell')$, hence $C_\ell=f(C_\ell')$.
		
		\underline{Case 2:} None of the clauses $C_j'$, $C_k'$ are contained in $\psi$.
		
		Then there exists $C_j:=f(C_j')\in \mu$ and $C_k:=f(C_k')\in \mu$. 
		Let  $C_\ell'=C_j'\resop{z}C_k' $ with $z\in \{ x_i,y_i \}$ for an $i\in [s_n]$. 
		We set $C_\ell:= C_j\resop{x_i}C_k  $. The only chance for this resolution to become unsound (with respect to \resolution) is w.l.o.g. $x_a\in C_j'$ and $\bar{y}_a\in C_k'$ for an $i\neq a\in \{1,\ldots, s_n \}$. Then we would receive $x_a\vee \bar{y}_a$ in $C_\ell'$ and a tautology $x_a\vee \bar{x}_a$ in $C_\ell$. However, this cannot happen due to our assumption that  the clauses of $\psi$ are no subclauses of $C_j',C_k',C_\ell'$. It is easy to see that $C_\ell=f(C_\ell')$.

		We now have constructed a \textsf{resolution} refutation $\mu$ with $|\mu|\in \mathcal{O}(|\pi|)=\mathcal{O}(|\pi''|)$ which only uses the clauses of $\mathtt{PHP}_n^{n+1}$ as axioms. By Proposition \ref{PropHakenPHP} the formula $\mathtt{PHP}_n^{n+1}$ needs exponential sized \textsf{resolution} refutations. Therefore $|\pi''|\in 2^{\Omega(n)}$, which is the size of the corresponding $\loar$ refutation as well.
	\end{proof}
	
	We remark that the hardness of $\mathtt{Trapdoor}_n$ crucially depends on propositional hardness. Note that it is not possible to just substitute $\mathtt{PHP}_n^{n+1}$ in $\mathtt{Trapdoor}_n$ by some QCNF that is hard for \ldqres or $\loar$ since it is not guaranteed that the conflict in the trail will occur in this embedded formula.
	
	In contrast, the $\mathtt{Trapdoor}_n$ formulas are easy in \qres.
	\begin{prop}  \label{prop:trapdoor-easy}
		The QCNFs $\mathtt{Trapdoor}_n$ have constant-size \qres refutations.
	\end{prop}	
	
	\begin{proof} The refutation is given in Figure~\ref{fig:trapdoor}.
		\begin{figure}

			
			\tikzset {_b67p9bkvt/.code = {\pgfsetadditionalshadetransform{ \pgftransformshift{\pgfpoint{0 bp } { 0 bp }  }  \pgftransformrotate{-270 }  \pgftransformscale{2 }  }}}
			\pgfdeclarehorizontalshading{_aqs9dmaug}{150bp}{rgb(0bp)=(1,1,1);
				rgb(37.5bp)=(1,1,1);
				rgb(50bp)=(0.95,0.95,0.95);
				rgb(50.25bp)=(0.88,0.88,0.88);
				rgb(62.5bp)=(0.96,0.96,0.96);
				rgb(100bp)=(0.96,0.96,0.96)}
			
			
			\tikzset {_dwcu8ya73/.code = {\pgfsetadditionalshadetransform{ \pgftransformshift{\pgfpoint{0 bp } { 0 bp }  }  \pgftransformrotate{-270 }  \pgftransformscale{2 }  }}}
			\pgfdeclarehorizontalshading{_q4jfvwa3y}{150bp}{rgb(0bp)=(1,1,1);
				rgb(37.5bp)=(1,1,1);
				rgb(50bp)=(0.95,0.95,0.95);
				rgb(50.25bp)=(0.88,0.88,0.88);
				rgb(62.5bp)=(0.96,0.96,0.96);
				rgb(100bp)=(0.96,0.96,0.96)}
			
			
			\tikzset {_wd5jfg42q/.code = {\pgfsetadditionalshadetransform{ \pgftransformshift{\pgfpoint{0 bp } { 0 bp }  }  \pgftransformrotate{-270 }  \pgftransformscale{2 }  }}}
			\pgfdeclarehorizontalshading{_ssm8l1hsb}{150bp}{rgb(0bp)=(1,1,1);
				rgb(37.5bp)=(1,1,1);
				rgb(50bp)=(0.95,0.95,0.95);
				rgb(50.25bp)=(0.88,0.88,0.88);
				rgb(62.5bp)=(0.96,0.96,0.96);
				rgb(100bp)=(0.96,0.96,0.96)}
			
			
			\tikzset {_vpotmfoom/.code = {\pgfsetadditionalshadetransform{ \pgftransformshift{\pgfpoint{0 bp } { 0 bp }  }  \pgftransformrotate{-270 }  \pgftransformscale{2 }  }}}
			\pgfdeclarehorizontalshading{_dzeqq20tp}{150bp}{rgb(0bp)=(1,1,1);
				rgb(37.5bp)=(1,1,1);
				rgb(50bp)=(0.95,0.95,0.95);
				rgb(50.25bp)=(0.88,0.88,0.88);
				rgb(62.5bp)=(0.96,0.96,0.96);
				rgb(100bp)=(0.96,0.96,0.96)}
			
			
			\tikzset {_zj4bcwbdv/.code = {\pgfsetadditionalshadetransform{ \pgftransformshift{\pgfpoint{0 bp } { 0 bp }  }  \pgftransformrotate{-270 }  \pgftransformscale{2 }  }}}
			\pgfdeclarehorizontalshading{_1rb0hxmp1}{150bp}{rgb(0bp)=(1,1,1);
				rgb(37.5bp)=(1,1,1);
				rgb(50bp)=(0.95,0.95,0.95);
				rgb(50.25bp)=(0.88,0.88,0.88);
				rgb(62.5bp)=(0.96,0.96,0.96);
				rgb(100bp)=(0.96,0.96,0.96)}
			
			
			\tikzset {_lan23thwm/.code = {\pgfsetadditionalshadetransform{ \pgftransformshift{\pgfpoint{0 bp } { 0 bp }  }  \pgftransformrotate{-270 }  \pgftransformscale{2 }  }}}
			\pgfdeclarehorizontalshading{_r9fgf10oq}{150bp}{rgb(0bp)=(1,1,1);
				rgb(37.5bp)=(1,1,1);
				rgb(50bp)=(0.95,0.95,0.95);
				rgb(50.25bp)=(0.88,0.88,0.88);
				rgb(62.5bp)=(0.96,0.96,0.96);
				rgb(100bp)=(0.96,0.96,0.96)}
			
			
			\tikzset {_mmfu4e1ux/.code = {\pgfsetadditionalshadetransform{ \pgftransformshift{\pgfpoint{0 bp } { 0 bp }  }  \pgftransformrotate{-270 }  \pgftransformscale{2 }  }}}
			\pgfdeclarehorizontalshading{_obl1fw1gx}{150bp}{rgb(0bp)=(1,1,1);
				rgb(37.5bp)=(1,1,1);
				rgb(50bp)=(0.95,0.95,0.95);
				rgb(50.25bp)=(0.88,0.88,0.88);
				rgb(62.5bp)=(0.96,0.96,0.96);
				rgb(100bp)=(0.96,0.96,0.96)}
			
			
			\tikzset {_lrf37g14m/.code = {\pgfsetadditionalshadetransform{ \pgftransformshift{\pgfpoint{0 bp } { 0 bp }  }  \pgftransformrotate{-270 }  \pgftransformscale{2 }  }}}
			\pgfdeclarehorizontalshading{_lgn8c6nvu}{150bp}{rgb(0bp)=(1,1,1);
				rgb(37.5bp)=(1,1,1);
				rgb(50bp)=(0.95,0.95,0.95);
				rgb(50.25bp)=(0.88,0.88,0.88);
				rgb(62.5bp)=(0.96,0.96,0.96);
				rgb(100bp)=(0.96,0.96,0.96)}
			
			
			\tikzset {_qdnn8ubu5/.code = {\pgfsetadditionalshadetransform{ \pgftransformshift{\pgfpoint{0 bp } { 0 bp }  }  \pgftransformrotate{-270 }  \pgftransformscale{2 }  }}}
			\pgfdeclarehorizontalshading{_2idvnz66g}{150bp}{rgb(0bp)=(1,1,1);
				rgb(37.5bp)=(1,1,1);
				rgb(50bp)=(0.95,0.95,0.95);
				rgb(50.25bp)=(0.88,0.88,0.88);
				rgb(62.5bp)=(0.96,0.96,0.96);
				rgb(100bp)=(0.96,0.96,0.96)}
			
			
			\tikzset {_wbxhettr6/.code = {\pgfsetadditionalshadetransform{ \pgftransformshift{\pgfpoint{0 bp } { 0 bp }  }  \pgftransformrotate{-270 }  \pgftransformscale{2 }  }}}
			\pgfdeclarehorizontalshading{_yf0fwe76h}{150bp}{rgb(0bp)=(1,1,1);
				rgb(37.5bp)=(1,1,1);
				rgb(50bp)=(0.95,0.95,0.95);
				rgb(50.25bp)=(0.88,0.88,0.88);
				rgb(62.5bp)=(0.96,0.96,0.96);
				rgb(100bp)=(0.96,0.96,0.96)}
			
			
			\tikzset {_06dz5ey8y/.code = {\pgfsetadditionalshadetransform{ \pgftransformshift{\pgfpoint{0 bp } { 0 bp }  }  \pgftransformrotate{-270 }  \pgftransformscale{2 }  }}}
			\pgfdeclarehorizontalshading{_6t83t1rjz}{150bp}{rgb(0bp)=(1,1,1);
				rgb(37.5bp)=(1,1,1);
				rgb(50bp)=(0.95,0.95,0.95);
				rgb(50.25bp)=(0.88,0.88,0.88);
				rgb(62.5bp)=(0.96,0.96,0.96);
				rgb(100bp)=(0.96,0.96,0.96)}
			
			
			\tikzset {_kvjkvr9hh/.code = {\pgfsetadditionalshadetransform{ \pgftransformshift{\pgfpoint{0 bp } { 0 bp }  }  \pgftransformrotate{-270 }  \pgftransformscale{2 }  }}}
			\pgfdeclarehorizontalshading{_294rqfp8o}{150bp}{rgb(0bp)=(1,1,1);
				rgb(37.5bp)=(1,1,1);
				rgb(50bp)=(0.95,0.95,0.95);
				rgb(50.25bp)=(0.88,0.88,0.88);
				rgb(62.5bp)=(0.96,0.96,0.96);
				rgb(100bp)=(0.96,0.96,0.96)}
			
			
			\tikzset {_r9lkvwfhd/.code = {\pgfsetadditionalshadetransform{ \pgftransformshift{\pgfpoint{0 bp } { 0 bp }  }  \pgftransformrotate{-270 }  \pgftransformscale{2 }  }}}
			\pgfdeclarehorizontalshading{_bu24oqrgm}{150bp}{rgb(0bp)=(1,1,1);
				rgb(37.5bp)=(1,1,1);
				rgb(50bp)=(0.95,0.95,0.95);
				rgb(50.25bp)=(0.88,0.88,0.88);
				rgb(62.5bp)=(0.96,0.96,0.96);
				rgb(100bp)=(0.96,0.96,0.96)}
			
			
			\tikzset {_npnc1sey6/.code = {\pgfsetadditionalshadetransform{ \pgftransformshift{\pgfpoint{0 bp } { 0 bp }  }  \pgftransformrotate{-270 }  \pgftransformscale{2 }  }}}
			\pgfdeclarehorizontalshading{_tsx65n2kg}{150bp}{rgb(0bp)=(1,1,1);
				rgb(37.5bp)=(1,1,1);
				rgb(50bp)=(0.95,0.95,0.95);
				rgb(50.25bp)=(0.88,0.88,0.88);
				rgb(62.5bp)=(0.96,0.96,0.96);
				rgb(100bp)=(0.96,0.96,0.96)}
			
			
			\tikzset {_q66hd1r5b/.code = {\pgfsetadditionalshadetransform{ \pgftransformshift{\pgfpoint{0 bp } { 0 bp }  }  \pgftransformrotate{-270 }  \pgftransformscale{2 }  }}}
			\pgfdeclarehorizontalshading{_5mgklpi68}{150bp}{rgb(0bp)=(1,1,1);
				rgb(37.5bp)=(1,1,1);
				rgb(50bp)=(0.95,0.95,0.95);
				rgb(50.25bp)=(0.88,0.88,0.88);
				rgb(62.5bp)=(0.96,0.96,0.96);
				rgb(100bp)=(0.96,0.96,0.96)}
			
			
			\tikzset {_n3git8u72/.code = {\pgfsetadditionalshadetransform{ \pgftransformshift{\pgfpoint{0 bp } { 0 bp }  }  \pgftransformrotate{-270 }  \pgftransformscale{2 }  }}}
			\pgfdeclarehorizontalshading{_x0dm62i22}{150bp}{rgb(0bp)=(1,1,1);
				rgb(37.5bp)=(1,1,1);
				rgb(50bp)=(0.95,0.95,0.95);
				rgb(50.25bp)=(0.88,0.88,0.88);
				rgb(62.5bp)=(0.96,0.96,0.96);
				rgb(100bp)=(0.96,0.96,0.96)}
			
			
			\tikzset {_b4rpf18pj/.code = {\pgfsetadditionalshadetransform{ \pgftransformshift{\pgfpoint{0 bp } { 0 bp }  }  \pgftransformrotate{-270 }  \pgftransformscale{2 }  }}}
			\pgfdeclarehorizontalshading{_rh3modo1b}{150bp}{rgb(0bp)=(1,1,1);
				rgb(37.5bp)=(1,1,1);
				rgb(50bp)=(0.95,0.95,0.95);
				rgb(50.25bp)=(0.88,0.88,0.88);
				rgb(62.5bp)=(0.96,0.96,0.96);
				rgb(100bp)=(0.96,0.96,0.96)}
			\tikzset{every picture/.style={line width=0.75pt}} 
			
			\begin{tikzpicture}[x=0.75pt,y=0.75pt,yscale=-1,xscale=0.9]
			
			\path  [shading=_aqs9dmaug,_b67p9bkvt] (50.5,46) .. controls (50.5,42.69) and (53.19,40) .. (56.5,40) -- (174,40) .. controls (177.31,40) and (180,42.69) .. (180,46) -- (180,64) .. controls (180,67.31) and (177.31,70) .. (174,70) -- (56.5,70) .. controls (53.19,70) and (50.5,67.31) .. (50.5,64) -- cycle ; 
			\draw  [color={rgb, 255:red, 0; green, 0; blue, 0 }  ,draw opacity=1 ] (50.5,46) .. controls (50.5,42.69) and (53.19,40) .. (56.5,40) -- (174,40) .. controls (177.31,40) and (180,42.69) .. (180,46) -- (180,64) .. controls (180,67.31) and (177.31,70) .. (174,70) -- (56.5,70) .. controls (53.19,70) and (50.5,67.31) .. (50.5,64) -- cycle ; 
			
			\path  [shading=_q4jfvwa3y,_dwcu8ya73] (190,46) .. controls (190,42.69) and (192.69,40) .. (196,40) -- (313.5,40) .. controls (316.81,40) and (319.5,42.69) .. (319.5,46) -- (319.5,64) .. controls (319.5,67.31) and (316.81,70) .. (313.5,70) -- (196,70) .. controls (192.69,70) and (190,67.31) .. (190,64) -- cycle ; 
			\draw  [color={rgb, 255:red, 0; green, 0; blue, 0 }  ,draw opacity=1 ] (190,46) .. controls (190,42.69) and (192.69,40) .. (196,40) -- (313.5,40) .. controls (316.81,40) and (319.5,42.69) .. (319.5,46) -- (319.5,64) .. controls (319.5,67.31) and (316.81,70) .. (313.5,70) -- (196,70) .. controls (192.69,70) and (190,67.31) .. (190,64) -- cycle ; 
			
			\path  [shading=_ssm8l1hsb,_wd5jfg42q] (120.5,106) .. controls (120.5,102.69) and (123.19,100) .. (126.5,100) -- (244,100) .. controls (247.31,100) and (250,102.69) .. (250,106) -- (250,124) .. controls (250,127.31) and (247.31,130) .. (244,130) -- (126.5,130) .. controls (123.19,130) and (120.5,127.31) .. (120.5,124) -- cycle ; 
			\draw  [color={rgb, 255:red, 0; green, 0; blue, 0 }  ,draw opacity=1 ] (120.5,106) .. controls (120.5,102.69) and (123.19,100) .. (126.5,100) -- (244,100) .. controls (247.31,100) and (250,102.69) .. (250,106) -- (250,124) .. controls (250,127.31) and (247.31,130) .. (244,130) -- (126.5,130) .. controls (123.19,130) and (120.5,127.31) .. (120.5,124) -- cycle ; 
			
			\path  [shading=_dzeqq20tp,_vpotmfoom] (341,46) .. controls (341,42.69) and (343.69,40) .. (347,40) -- (464.5,40) .. controls (467.81,40) and (470.5,42.69) .. (470.5,46) -- (470.5,64) .. controls (470.5,67.31) and (467.81,70) .. (464.5,70) -- (347,70) .. controls (343.69,70) and (341,67.31) .. (341,64) -- cycle ; 
			\draw  [color={rgb, 255:red, 0; green, 0; blue, 0 }  ,draw opacity=1 ] (341,46) .. controls (341,42.69) and (343.69,40) .. (347,40) -- (464.5,40) .. controls (467.81,40) and (470.5,42.69) .. (470.5,46) -- (470.5,64) .. controls (470.5,67.31) and (467.81,70) .. (464.5,70) -- (347,70) .. controls (343.69,70) and (341,67.31) .. (341,64) -- cycle ; 
			
			\path  [shading=_1rb0hxmp1,_zj4bcwbdv] (480.5,46) .. controls (480.5,42.69) and (483.19,40) .. (486.5,40) -- (604,40) .. controls (607.31,40) and (610,42.69) .. (610,46) -- (610,64) .. controls (610,67.31) and (607.31,70) .. (604,70) -- (486.5,70) .. controls (483.19,70) and (480.5,67.31) .. (480.5,64) -- cycle ; 
			\draw  [color={rgb, 255:red, 0; green, 0; blue, 0 }  ,draw opacity=1 ] (480.5,46) .. controls (480.5,42.69) and (483.19,40) .. (486.5,40) -- (604,40) .. controls (607.31,40) and (610,42.69) .. (610,46) -- (610,64) .. controls (610,67.31) and (607.31,70) .. (604,70) -- (486.5,70) .. controls (483.19,70) and (480.5,67.31) .. (480.5,64) -- cycle ; 
			
			\path  [shading=_r9fgf10oq,_lan23thwm] (410,106) .. controls (410,102.69) and (412.69,100) .. (416,100) -- (533.5,100) .. controls (536.81,100) and (539.5,102.69) .. (539.5,106) -- (539.5,124) .. controls (539.5,127.31) and (536.81,130) .. (533.5,130) -- (416,130) .. controls (412.69,130) and (410,127.31) .. (410,124) -- cycle ; 
			\draw  [color={rgb, 255:red, 0; green, 0; blue, 0 }  ,draw opacity=1 ] (410,106) .. controls (410,102.69) and (412.69,100) .. (416,100) -- (533.5,100) .. controls (536.81,100) and (539.5,102.69) .. (539.5,106) -- (539.5,124) .. controls (539.5,127.31) and (536.81,130) .. (533.5,130) -- (416,130) .. controls (412.69,130) and (410,127.31) .. (410,124) -- cycle ; 
			
			\path  [shading=_obl1fw1gx,_mmfu4e1ux] (120.5,166) .. controls (120.5,162.69) and (123.19,160) .. (126.5,160) -- (244,160) .. controls (247.31,160) and (250,162.69) .. (250,166) -- (250,184) .. controls (250,187.31) and (247.31,190) .. (244,190) -- (126.5,190) .. controls (123.19,190) and (120.5,187.31) .. (120.5,184) -- cycle ; 
			\draw  [color={rgb, 255:red, 0; green, 0; blue, 0 }  ,draw opacity=1 ] (120.5,166) .. controls (120.5,162.69) and (123.19,160) .. (126.5,160) -- (244,160) .. controls (247.31,160) and (250,162.69) .. (250,166) -- (250,184) .. controls (250,187.31) and (247.31,190) .. (244,190) -- (126.5,190) .. controls (123.19,190) and (120.5,187.31) .. (120.5,184) -- cycle ; 
			
			\path  [shading=_lgn8c6nvu,_lrf37g14m] (410,166) .. controls (410,162.69) and (412.69,160) .. (416,160) -- (533.5,160) .. controls (536.81,160) and (539.5,162.69) .. (539.5,166) -- (539.5,184) .. controls (539.5,187.31) and (536.81,190) .. (533.5,190) -- (416,190) .. controls (412.69,190) and (410,187.31) .. (410,184) -- cycle ; 
			\draw  [color={rgb, 255:red, 0; green, 0; blue, 0 }  ,draw opacity=1 ] (410,166) .. controls (410,162.69) and (412.69,160) .. (416,160) -- (533.5,160) .. controls (536.81,160) and (539.5,162.69) .. (539.5,166) -- (539.5,184) .. controls (539.5,187.31) and (536.81,190) .. (533.5,190) -- (416,190) .. controls (412.69,190) and (410,187.31) .. (410,184) -- cycle ; 
			
			\path  [shading=_2idvnz66g,_qdnn8ubu5] (267,226) .. controls (267,222.69) and (269.69,220) .. (273,220) -- (390.5,220) .. controls (393.81,220) and (396.5,222.69) .. (396.5,226) -- (396.5,244) .. controls (396.5,247.31) and (393.81,250) .. (390.5,250) -- (273,250) .. controls (269.69,250) and (267,247.31) .. (267,244) -- cycle ; 
			\draw  [color={rgb, 255:red, 0; green, 0; blue, 0 }  ,draw opacity=1 ] (267,226) .. controls (267,222.69) and (269.69,220) .. (273,220) -- (390.5,220) .. controls (393.81,220) and (396.5,222.69) .. (396.5,226) -- (396.5,244) .. controls (396.5,247.31) and (393.81,250) .. (390.5,250) -- (273,250) .. controls (269.69,250) and (267,247.31) .. (267,244) -- cycle ; 
			
			\draw [shading=_yf0fwe76h,_wbxhettr6]   (120,70) -- (180,100) ;
			\draw [shading=_6t83t1rjz,_06dz5ey8y]   (250,70) -- (190,100) ;
			\draw [shading=_294rqfp8o,_kvjkvr9hh]   (410,70) -- (470,100) ;
			\draw [shading=_bu24oqrgm,_r9lkvwfhd]   (540,70) -- (480,100) ;
			
			\draw [shading=_tsx65n2kg,_npnc1sey6]   (185,160) -- (185,141) -- (185,130) ;
			\draw [shading=_5mgklpi68,_q66hd1r5b]   (476,160) -- (476,141) -- (476,130) ;
			\draw [shading=_x0dm62i22,_n3git8u72]   (330,220) -- (190,190) ;
			\draw [shading=_rh3modo1b,_b4rpf18pj]   (330,220) -- (480,190) ;
			
			\draw (115.25,55) node  [color={rgb, 255:red, 255; green, 255; blue, 255 }  ,opacity=1 ] [align=left] {$\displaystyle \textcolor[rgb]{0,0,0}{\mathbf{y_{1} \lor w\lor t}}$};
			\draw (254.75,55) node  [color={rgb, 255:red, 255; green, 255; blue, 255 }  ,opacity=1 ] [align=left] {$\displaystyle \mathbf{\textcolor[rgb]{0,0,0}{y_{1} \lor w\lor \overline{t}}}$};
			\draw (185.25,115) node  [color={rgb, 255:red, 255; green, 255; blue, 255 }  ,opacity=1 ] [align=left] {$\displaystyle \mathbf{\textcolor[rgb]{0,0,0}{y_{1} \lor w}}$};
			\draw (545.25,55) node  [color={rgb, 255:red, 255; green, 255; blue, 255 }  ,opacity=1 ] [align=left] {$\displaystyle \mathbf{\textcolor[rgb]{0,0,0}{\overline{y}_{1} \lor w\lor \overline{t}}}$};
			\draw (405.75,55) node  [color={rgb, 255:red, 255; green, 255; blue, 255 }  ,opacity=1 ] [align=left] {$\displaystyle \mathbf{\textcolor[rgb]{0,0,0}{\overline{y}_{1} \lor w\lor t}}$};
			\draw (474.75,115) node  [color={rgb, 255:red, 255; green, 255; blue, 255 }  ,opacity=1 ] [align=left] {$\displaystyle \mathbf{\textcolor[rgb]{0,0,0}{\overline{y}_{1} \lor w}}$};
			\draw (185.25,175) node  [color={rgb, 255:red, 255; green, 255; blue, 255 }  ,opacity=1 ] [align=left] {$\displaystyle \mathbf{\textcolor[rgb]{0,0,0}{(y_{1})}}$};
			\draw (474.75,175) node  [color={rgb, 255:red, 255; green, 255; blue, 255 }  ,opacity=1 ] [align=left] {$\displaystyle \mathbf{\textcolor[rgb]{0,0,0}{(\overline{y}_{1})}}$};
			\draw (330.75,235) node  [color={rgb, 255:red, 255; green, 255; blue, 255 }  ,opacity=1 ] [align=left] {$\displaystyle \textcolor[rgb]{0,0,0}{\mathbf{(\bot)}}$};
			
			\end{tikzpicture}
			\caption{Short refutations of $\mathtt{Trapdoor}_n$ in  \qres  \label{fig:trapdoor}}
		\end{figure}
	\end{proof}

	These two results immediately lead to the following separation.
	
	\begin{thm}\label{CorIncomparable}\label{TheoremIncomparable}
		The systems \qres and $\loar$ are incomparable.
	\end{thm}
	
	From Theorem \ref{TheoremIncomparable}  we can conclude not only that we have to modify the QCDCL proof system if we aim to characterise \qres, it also implies that a simple strengthening (or weakening) of one of the two systems cannot result in the desired equivalence. We mentioned earlier that the policies $\ar$ and $\nr$ seem to operate orthogonally to each other. In Section~\ref{sec:sim-order} we will get back to this point and formally prove this intuition. This also motivates why switching from policy $\ar$ to $\nr$ can be helpful in obtaining a QCDCL system that characterises \qres.

	\section{Hard formulas for  QCDCL} \label{sec:hardness-qcdcl}
	As we have shown in Section \ref{sec:separation-qcdcl-qres}, \qres is incomparable to \qcdcl. This leaves open the question of what formulas are hard for \qcdcl, without relying on hardness of \qres or \ldqres. 
	\begin{defi}
		We call a \ldqres proof $\pi$ of a clause $C$ from a QCNF $\Phi$ \emph{a \ldqcdcl proof of $C$ from $\Phi$}, if there exists a $\loar$ proof $\iota$ of $C$ from $\Phi$ such that the \ldqres proof $\pi$ is obtained by pasting together the sub-proofs $(\pi_1,\ldots,\pi_m)$ from~$\iota$ (cf.\  Definition~\ref{DefinitionQCDCLrefutation}).
	\end{defi}
	The system \ldqcdcl identifies a fragment of \ldqres, which collects all \ldqres proofs that appear in $\loar$ derivations.  By definition therefore, \ldqcdcl and $\loar$ are p-equivalent proof systems.
	
	Our next goal is to identify a whole class of QCNFs that witness the hardness of \qcdcl.
	
	\begin{defi}  \label{def:XT-property}
		Let $\Phi$ be a QCNF of the form 
		\begin{align*}
		\exists X \forall U \exists T \cdot \phi
		\end{align*}
		with sets of variables $X=\{ x_1, \ldots,x_a \}$, $U=\{u_1,\ldots,u_b\}$ and $T=\{ t_1,\ldots,t_c \}$. 
		
		We call a clause $C$ in the variables of $\Phi$
		\begin{itemize}
			\item \emph{$T$-clause}, if $\var(C)\cap X=\emptyset$, $ \var(C)\cap U=\emptyset$ and $\var(C)\cap T\neq \emptyset$,
			\item \emph{$XT$-clause}, if $\var(C)\cap X\neq\emptyset$, $ \var(C)\cap U=\emptyset$ and $\var(C)\cap T\neq \emptyset$,
			\item \emph{$XUT$-clause}, if $\var(C)\cap X\neq\emptyset$, $ \var(C)\cap U\neq\emptyset$ and $\var(C)\cap T\neq \emptyset$.
		\end{itemize}
		
		We say that $\Phi$ fulfils the \emph{$XT$-property} if $\phi$ contains no $XT$-clauses as well as no unit $T$-clauses and there do not exist two $T$-clauses $C_1,C_2\in \phi$ that are resolvable.
	\end{defi}

The next lemma shows that, under the XT-property, we cannot derive any XT-clauses or new T-clauses.

\begin{lem}\label{LemmaNoTclause}\label{LemmaNoXTclause}
	Let $\Phi$ be a QBF with the XT-property and let $C$ be a clause derived by \ldqres from $\Phi$ such that $C$ contains a $T$-literal. Then $C$ is an axiom or it contains a $U$-literal. 
\end{lem}
\begin{proof} 
	We prove this by induction. Let $\pi$ be the \ldqres proof of $C$. If $C$ is an axiom, then there is nothing to prove. Therefore, let us assume that $C$ is not an axiom. Then the last step in the derivation must have been a resolution step, because we cannot reduce literals when a $T$-literal is present.
	
	
	Hence there are clauses $D,E\in \pi$ and a literal $\ell$ with $\ell\in D$ and $\bar{\ell}\in E$ such that $C=D\resop{\ell}E$.  By our induction hypothesis, $D$ (as well as $E$) is an axiom or it contains a $U$-literal. If one of them contains a $U$-literal, then $C$ also contains one and we are done. Therefore, let us consider the case where both $D$ and $E$ are axioms, but do not contain $U$-literals. Since XT-clauses are forbidden by the XT-property, $D$ and $E$ are either $X$-clauses or $T$-clauses. Again, by the XT-property, $\Phi$ does not contain two resolvable T-clauses, therefore $D$ and $E$ must be X-clauses. But then $C$ cannot contain a $T$-literal, which is a contradiction. 
\end{proof}
	
	We will show later that we need to resolve two $XUT$-clauses over an $X$-literal in order to introduce tautologies. Now we prove that this is not possible in \ldqcdcl under the $XT$-property.
	
	\begin{lem}\label{LemmaResolveTwoXUTclausesNotPossible}
		It is not possible to resolve two $XUT$-clauses over an $X$-literal in a \ldqcdcl proof of a QCNF $\Phi$ that fulfils the $XT$-property.
	\end{lem}
	
	\begin{proof} Assume there is a \ldqcdcl proof $\pi$ that contains such a resolution step over an $X$-literal $x$. Let $C_1$ and $C_2$ be the corresponding $XUT$-clauses. One of these clauses, say $C_1$, had to be an antecedent clause in a $\loar$ trail $\mathcal{T}$ that implied $x$. Since our decisions in the trail are level-ordered and we did not skip any decisions, either $x$ was propagated at decision level $0$, or at a decision level in which we decided another $X$-literal. 
		
		Because $C_1$ is an $XUT$-clause, we can find a $T$-literal $t\in C_1$. The literal $\bar{t}$ must have been propagated before we implied $x$ ($\bar{t}$ could not have been decided because the decisions are level-ordered). That means that  for the same trail we can find $E:=\ante_\mathcal{T}(\bar{t})$. Now, $E$ cannot be a unit $T$-clause by the $XT$-property and Lemma \ref{LemmaNoTclause}. Hence $E$ must contain further $X$-, $U$-, or $T$-literals. If $E$  contains a $U$-literal, then we would have had to decide this $U$-literal before we use $E$ as an antecedent clause, contradicting the level-order of our decisions. Also, this $U$-literal cannot be reduced since we want to imply a $T$-literal with the help of $E$. Therefore we conclude that $E$ contains an $X$-literal or a $T$-literal. If $E$ contains an $X$-literal, then $E$ is an $XT$-clause, which is not possible by Lemma \ref{LemmaNoXTclause}.
		
		Therefore $E$ contains at least another $T$-literal $\ell\in E$. As before, the literal $\bar{\ell}$ was propagated before we implied $\bar{t}$ and $x$. We set $E':=\ante_\mathcal{T}(\bar{\ell})$ and argue in the same way as with $E$. This process would repeat endlessly, which is a contradiction since we only have finitely many $T$-variables.  
	\end{proof}
	
	We combine the above results and show that the hardness of QCNFs with the $XT$-property lifts from \qres to \ldqcdcl.

	\begin{thm}\label{TheoremXTpropertyRequiresLongDistanceProofsSizeS}
		If $\Phi$ fulfils the $XT$-property and requires \qres refutations of size~$s$, then
		each \ldqcdcl refutation (and therefore also each $\loar$ refutation) of $\Phi$ has at least size $s$ as well. In detail, each \ldqcdcl refutation of $\Phi$ is in fact a \qres refutation.
	\end{thm}	
	\begin{proof} Let $\pi$ be a \ldqcdcl refutation of $\Phi$. We show
		that $\pi$ does not contain any tautological clause C and hence $\pi$ is in fact a
		Q-resolution refutation. 
		
		Assume that $\pi$ contains some tautological clause $C$. W.l.o.g.\ let $C$ be the first tautological clause in $\pi$. Clearly, $C$ has to be derived by a resolution step over an $X$-literal. Let $C_1$ and $C_2$ be the parent clauses of $C$. Both of them contain some $X$-literals and some $U$-literals. They also have to contain $T$-literals, otherwise we would reduce all $U$-literals (in the learning process we reduce as soon as possible). Therefore $C_1$ and $C_2$ are both $XUT$-clauses that are resolved over an $X$-literal, which is not possible by Lemma \ref{LemmaResolveTwoXUTclausesNotPossible}.
		
		Therefore such a clause $C$ cannot exist. Hence each \ldqcdcl refutation of $\Phi$ is even a \qres refutation and  the result follows.
	\end{proof}
	
	One can conclude that QCNFs with the $XT$-property that are hard for \qres are also hard for \ldqres and therefore \qcdcl. We will give some examples for these cases.
	
	\begin{defiC}[\cite{BBH19}]
		The formula $\mathtt{Equality}_n$  is defined as the QCNF
		$$
		\exists x_1\ldots x_n\forall u_1\ldots u_n \exists t_1\ldots t_n \cdot (\bar{t}_1\vee\ldots\vee \bar{t}_n)\wedge \bigwedge_{i=1}^n((\bar{x}_i\vee \bar{u}_i\vee t_i)\wedge (x_i\vee u_i\vee t_i))\text{.}
		$$
		
	\end{defiC}
	
	This QCNF is obviously false since the $\forall$-player has a winning strategy by assigning each $u_i$ equal to the assignment of $x_i$.
	
	\begin{thmC}[\cite{BBH19}]
		$\mathtt{Equality}_n$ requires \qres refutations of size $2^n$.
	\end{thmC}
	
	It is easy to see that $\mathtt{Equality}_n$ fulfils the $XT$-property for $n\geq 2$. Therefore we obtain:
	
	\begin{cor}  \label{cor:equality-hardness}
		$\mathtt{Equality}_n$ requires $\loar$ refutations of size $2^n$.
	\end{cor}
	
	Since the equality formulas are easy in \ldqres \cite{BeyersdorffBM19}, we obtain an exponential separation between \qcdcl and \ldqres.
	\begin{cor} \label{cor:qcdcl-ldqres}
		\Ldqres is exponentially stronger than \qcdcl, i.e., \ldqres p-simulates \qcdcl and there are QCNFs that require exponential-size proofs in \qcdcl, but admit polynomial-size proofs in \ldqres.
	\end{cor}
	
	Next we will define a whole class of randomly generated QCNFs. With high probability, they also serve as hard examples for \qcdcl. 
	\begin{defiC}[\cite{BBH19}]
		For each $1\leq i\leq n$ let $C_i^{(1)},\ldots,C_i^{(cn)}$ be clauses picked uniformly at random from the set of clauses containing 1 literal from the set $U_i=\{ u_i^{(1)},\ldots,u_i^{(m)} \}$ and 2 literals from $X_i=\{ x_i^{(1)},\ldots,x_i^{(n)} \}$. Define the randomly generated QCNF $Q(n,m,c)$ as:
		\begin{align*}
		Q(n,m,c):=\exists X_1,\ldots,X_n\forall U_1,\ldots,U_n\exists t_1,\ldots,t_n \cdot \bigwedge_{i=1}^n\bigwedge_{j=1}^{cn}(\bar{t}_i\vee C_i^{(j)})\wedge (t_1\vee \ldots \vee t_n)
		\end{align*}
	\end{defiC}
	
	Suitably choosing the parameters $c$ and $m$, we gain false and indeed hard formulas. 
	
	\begin{thmC}[\cite{BBH19}]\label{PropQnmcFalseProbability}
		Let $1<c<2$ be a constant and $m\leq (1-\epsilon)\log_2n$ for some constant $\epsilon>0$. With probability $1-o(1)$ the random QCNF $Q(n,m,c)$ is false and requires  \qres refutations of size $2^{\Omega(n^\epsilon)}$. 
	\end{thmC}
	
	Again, it is easy to see that all $Q(n,m,c)$-formulas fulfil the $XT$-property.

	\begin{cor} \label{cor:random-hardness}
		Let $1<c<2$ be a constant and $m\leq (1-\epsilon)\log_2n$ for some constant $\epsilon>0$. With probability $1-o(1)$   the random QCNF $Q(n,m,c)$ is false and requires  $\loar$ refutations of size $2^{\Omega(n^\epsilon)}$. 
	\end{cor}

	The hardness of $\mathtt{Equality}_n$ and the random formulas does not rely on propositional hardness. In order to make the notion of `propositional hardness' formal, we will use a strengthening of \qres that allows oracle calls. This framework was introduced in \cite{BHP20} and tailored towards \qres in \cite{BBM20}. The oracle allows to collapse arbitrary propositional sub-derivations into just one inference step.
	

	\begin{defiC}[\cite{BBM20}]
		\qnpres is defined as \qres, but the resolution rule is replaced by the following:
		\begin{itemize}
			\item \emph{$\Sigma^\exists_1$-rule:} For some $\mathcal{G}\subseteq \{ C_1,\ldots,C_{i-1} \}$,
			\begin{enumerate}
				\item $\bigwedge_{B\in \mathcal{G}}B^\exists \vDash C_i^\exists$, and
				\item for each $B\in \mathcal{G}$, $B^\forall$ is a subclause from $C_i^\forall$,
			\end{enumerate}
			where $C^\exists$ and $C^\forall$ denote the existential and universal subclauses of any clause $C$. 
		\end{itemize}
		We use the notation $C_1\wedge\ldots\wedge C_{i-1}\entailsSigma C_i$ for referring to such a step.
	\end{defiC}
	
	In fact, the lower bounds for the equality and random formulas above hold in this stronger model.
	\begin{thmC}[\cite{BBM20,BBH19}]\label{PropQNPResOfEqualityIsLarge}
		$\mathtt{Equality}_n$ requires \qnpres refutations of size $2^n$. Likewise, the random formulas $Q(n,m,c)$ (with the same parameters as in Theorem~\ref{PropQnmcFalseProbability}) are false and require  $\qnpres$ refutations of size $2^{\Omega(n^\epsilon)}$. 
	\end{thmC}
	
	An equivalent way of stating this theorem is to say that the equality and random formulas require exponentially many reduction steps in \qres proofs. This measure is also applicable to QCDCL trails and in particular to \ldqcdcl proofs. 
	
	\begin{prop} \label{prop:equality-random-oracle-hardness}
		The number of reduction steps in each \ldqcdcl refutation (and also each $\loar$ refutation)  of $\mathtt{Equality}_n$  is at least $2^n$. The same holds for the false formulas $Q(n,m,c)$ with $2^{\Omega(n^\epsilon)}$ reduction steps.
	\end{prop}
	\begin{proof} 
		Let $\pi=C_1,\ldots,C_m$ be a \ldqcdcl refutation of $\mathtt{Equality}_n$ or $Q(n,m,c)$ and let $r$ be the number of clauses in $\pi$ that were directly derived by reduction of some preceding clause. By Proposition \ref{TheoremXTpropertyRequiresLongDistanceProofsSizeS}, $\pi$ is even a \qres refutation. W.l.o.g.\ all axioms are introduced at the beginning of $\pi$. I.e., for an $\ell\in \N$ we have that $C_1,\ldots,C_\ell$ are all axioms from $\mathtt{Equality}_n$. This does not change the size of $\pi$ as we need all axioms anyway. 
		
		We can assume that $\pi$ is in the form of 
		\begin{align*}
		\pi=&C_1,\ldots,C_\ell,R_{(1,1)},\ldots,R_{(1,s_1)},T_1,R_{(2,1)},\ldots,R_{(2,s_2)},T_2,\ldots , R_{(r,1)},\ldots,R_{(r,s_r)},T_r,\\&R_{(r+1,1)},\ldots,R_{(r+1,s_{r+1})}\text{}
		\end{align*}
		where each $R_{(i,j)}$ is (directly) derived by the resolution rule, and for each $T_i$ we have $T_i=\red(R_{(i,s_i)})$. One can inductively show that for each $R_{(c,d)}$ we have
		\begin{align*}
		\bigwedge_{i=1}^\ell C_i \wedge \bigwedge_{j=1}^{c-1} T_j \entailsSigma R_{(c,d)}\text{}
		\end{align*}
		due to 
		\begin{align*}
		\bigwedge_{i=1}^\ell C_i \wedge \bigwedge_{j=1}^{c-1} T_j \entailsSigma R_{(a,b)}\text{}
		\end{align*}
		for every $R_{(a,b)}$ left of $R_{(c,d)}$ in $\pi$. 

		We can replace consecutive resolution steps with the $\Sigma_1^\exists$-rule and obtain a \qnpres refutation $\pi'$ with the same number of reductions as $\pi$. We have $|\pi'|\in 2^{\Omega(n)}$ (resp.\ $2^{\Omega(n^\epsilon)}$ for $Q(n,m,c)$) by Proposition \ref{PropQNPResOfEqualityIsLarge}. This new refutation $\pi'$ is  of the form 
		\begin{align*}
		\pi'=C_1,\ldots,C_\ell,R_{(1,s_1)},T_{1},R_{(2,s_2)},T_2,\ldots,R_{(r,s_r)},T_r,R_{(r+1,s_{r+1})}\text{.}
		\end{align*}  
		I.e., these two kinds of derivation steps are alternating after $C_\ell$. Therefore the number of reductions $r$ of $\pi'$ (and also $\pi$) can be estimated by
		\begin{align*}
		r\in \Omega\left(\frac{|\pi'|-\ell}{2}\right)=\Omega(|\pi'|)\text{.} \tag*{\qedhere}
		\end{align*}
	\end{proof}
	
	Compared to $\mathtt{Equality}_n$, the formulas $\mathtt{Trapdoor}_n$ just need a linear number of reduction steps in $\loar$.

	\section{A QCDCL system equivalent to Q-resolution}\label{SectionMainResult}
	\label{sec:qres-characterisation}

	We now show that the QCDCL system $\asonr$ has exactly the right strength to characterise \qres. 
	
	Our central notion in this section will be that of the \emph{reliability} of a clause. The motivation for this definition is as follows: In our QCDCL systems we generate trails and proofs in a natural way, i.e., we are not allowed to skip propagations or conflicts. Nevertheless, when aiming to simulate \qres proofs, we also want  to  prescribe a sequence of literals as decisions in order to create trails along these decisions. However, it is not guaranteed that we can even make all these decisions without problems. We could run into situations that prevent us from continuing with the desired decisions. To classify these situations, we use the notion of reliability.
	
	\begin{defi} \label{def:reliable}
		Let $\Phi=\mathcal{Q}\cdot \phi$ be a QCNF and $C$ be a non-tautological clause. If there is a $\asonr$ trail $\mathcal{T}$, an existential literal $\ell\in C$ and a set of literals $\alpha\subseteq \bar{C}\backslash\{ \bar{\ell} \}$ such that $\alpha$ is the set of decision literals in $\mathcal{T}$ and $\ell\in \mathcal{T}$, then $C$ is called \emph{unreliable with respect to} $\Phi$. Alternatively, we say that the decisions $\bar{C}$ are \emph{blocking each other}. 
		
		If $C$ is not unreliable, we call $C$ \emph{reliable}.

	\end{defi}

	\begin{rem}\label{RemConstructTrailWithDecisions} \hfill 
\begin{enumerate}
						\item Let us give an intuitive interpretation of unreliability: Let us assume the two clauses $(x\vee y\vee a)$ and $(\bar{x}\vee z\vee b)$ are given (we ignore quantifiers for the moment). These two clauses can be resolved over $x$, and we receive $(x\vee y\vee a)\resop{x}(\bar{x}\vee z\vee b)=(y\vee z\vee a \vee b)$. We would like to simulate this resolution step via trails. The most obvious (and naive) approach would be to decide $\bar{y}$ and $\bar{a}$, propagate $x$, decide $\bar{z}$ and $\bar{b}$, receive a conflict. From that trail $(\mathbf{\bar{y}};\mathbf{\bar{a}},x;\mathbf{\bar{z}};\mathbf{\bar{b}},\bot)$ we can learn the clause $(y\vee z\vee a \vee b)$. However, this would not work in practice: After deciding $\bar{z}$, we would detect the unit clause $(b)$, which will trigger the propagation of $b$ and prevent us from finding a conflict. Hence, we would actually create the trail $(\mathbf{\bar{y}};\mathbf{\bar{a}},x;\mathbf{\bar{z}},b)$. This shows that the clause $(y\vee z\vee a \vee b)$ is unreliable (with respect to our formula). Although this might look like a disadvantage, making clauses unreliable or detecting their unreliability is the aim of our simulation method. Basically, unreliability means that, even though this clause might not be contained in our formula as an axiom, we can still (or already) make use of its impact to trigger a unit propagation as if we would know the clause directly. In our example, although we have never learned the clause $(y\vee z\vee a \vee b)$, we have found out how we can use the implication $(\bar{y}\wedge \bar{z}\wedge \bar{a})\rightarrow b$ as an indirect unit propagation.
			
			We will show later that unreliable clauses are as good as learned clauses. Also, it is not important which literal got implied by the witness of unreliability. For our example that means that, even if we later need the clause $(y\vee z\vee a \vee b)$ for a propagation of $a$, it suffices to know how we can imply $b$, or any other literal of the formula, by deciding the negations of the  remaining literals.
			
			\item 	Let $\Phi=\mathcal{Q}\cdot \phi$ be a QCNF and let $C$ be reliable with respect to $\Phi$. Then we can construct a trail $\mathcal{T}$ by choosing $\bar{C}$ as decisions one by one and doing the propagations more or less automatically. These decisions cannot block each other. However, it is possible that we propagate a literal from $\bar{C}$ in the same polarity before deciding it. In this case we have to skip the decision. Also, we could reach a conflict before deciding all literals, then we abort the trail as usual. Both of these  cases are still fine for our purposes. We stop the construction when we either reach a conflict, or if all literals from $\bar{C}$ are assigned and we cannot make further propagations.
			
			We will use this technique in our later proofs and refer to it as \emph{constructing $\mathcal{T}$ with decisions $\bar{C}$}. Note that all trails we construct in this way are natural. 
		\end{enumerate}
	\end{rem}

	\begin{exas}
		We give some minimal examples of the situations of Remark \ref{RemConstructTrailWithDecisions} in the system $\asonr$. In all four examples, we want to check if the clause $(\bar{x}\vee\bar{u} \vee \bar{y})$ is reliable with respect to the corresponding formula $\Phi$.
		\begin{enumerate}
			\item Let $\Phi:=\exists x\forall u\exists y,z \cdot (\bar{x}\vee \bar{u}\vee \bar{y}\vee z)$.  Constructing a trail with decisions $(x,u,y)$ in this order gives us $\mathcal{T}:=( \textbf{x};\textbf{u};\textbf{y},z  )$. In this situation the decisions $(x,u,y)$ are not blocking each other. In fact, the clause $(\bar{x}\vee \bar{u}\vee \bar{y})$ is reliable with respect to $\Phi$ since we would need to decide $\bar{z}$ in order to propagate $\bar{x}$ or $\bar{y}$.
			\item Let $\Phi:=\exists x\forall u\exists y,z \cdot (\bar{x}\vee \bar{u}\vee \bar{y}\vee z)\wedge (\bar{x}\vee \bar{u}\vee y)$.  After constructing a trail $\mathcal{T}$ with decisions $(x,u,y)$ we get  $\mathcal{T}=(\textbf{x};\textbf{u},y,z)$. This is still fine even if we did not actually use $y$ as a decision. The clause $(\bar{x}\vee\bar{u} \vee \bar{y})$ is still reliable with respect to $\Phi$ because we would need to decide $\bar{z}$ or $\bar{y}$ in order to imply $\bar{x}$.
			\item Let $\Phi:=\exists x\forall u\exists y,z \cdot (\bar{x}\vee \bar{u}\vee \bar{y}\vee z)\wedge (\bar{x}\vee \bar{u}\vee z)\wedge (\bar{x}\vee \bar{u}\vee \bar{z})$. When we want to construct a trail with decisions $(x,u,y)$ we get $\mathcal{T}=\{  \textbf{x},\textbf{u},z,\bot \}$. This is fine even though $y$ does not appear in $\mathcal{T}$, since we run into a conflict beforehand. The clause $(\bar{x}\vee\bar{u} \vee \bar{y})$ is reliable with respect to $\Phi$, because the only way to achieve a situation of unreliability would be implying $\bar{x}$ with decisions $u$ and $y$, or implying $\bar{y}$ with decisions $x$ and $u$. In the first case we would not get any unit clauses, and in the second one we would propagate $\bot$ before $\bar{y}$.
			\item Let $\Phi:=\exists x\forall u\exists y,z \cdot (\bar{x}\vee \bar{u}\vee z)\wedge(\bar{z}\vee \bar{y})$. When we try to construct a trail $\mathcal{T}$ with decisions $(x,u,y)$ we get stuck at $\mathcal{T}=(\textbf{x};\textbf{u},z,\bar{y})$. Now we have propagated $\bar{y}$ although we wanted to decide $y$. This shows that these decisions block each other and the clause $(\bar{x}\vee \bar{u}\vee \bar{y})$ is unreliable with respect to $\Phi$ and $\mathcal{T}$ serves as a witness.
		\end{enumerate}
	\end{exas}
	
	The next lemma shows that we can basically `copy' a trail that was created at a previous point by deciding or propagating all of its decisions. This will help us later during the simulation of resolution or reductions steps, where we want to copy the trails that serve as a witness for the unreliability of the parent clauses.
	
	\begin{lem}\label{LemmaFindEveryPropagation}
		Let $\Phi=\mathcal{Q}\cdot \phi$ and $\Psi=\mathcal{Q}\cdot \psi$ be two QCNFs with the same prefix such that $\phi\subseteq \psi$. Let $\mathcal{T}$ be a $\asonr$ trail for $\Phi$ and $\mathcal{U}$ be a natural trail for $\Psi$. Let $\alpha$ be the decisions in $\mathcal{T}$ and $\alpha\subseteq \mathcal{U}$. If $\mathcal{U}$ does not run into a conflict, then every propagated literal from $\mathcal{T}$ is contained in~$\mathcal{U}$.
	\end{lem}	
	
	\begin{proof}  Suppose that $\mathcal{U}$ did not run into a conflict. 
		
		Write $\mathcal{T}$ as
		\begin{align*}
		\mathcal{T}=(p_{(0,1)},\ldots, p_{(0,g_0)};\mathbf{d_{1}},p_{(1,1)},\ldots,p_{(1,g_{1})};\ldots; \mathbf{d_r},p_{(r,1)},\ldots ,p_{(r,g_r)}  )\text{.}
		\end{align*}
		Assume that there are some propagated literals from $\mathcal{T}$ that are not contained in $\mathcal{U}$.
		Let $p_{(i,j)}$ be the first (leftmost) literal in $\mathcal{T}$ of this kind. 
		Then we know that
		\begin{align*}
		\ante_\mathcal{T}(p_{(i,j)})|_{\mathcal{T}[i,j-1]}=(p_{(i,j)})
		\end{align*}
		with $\ante_\mathcal{T}(p_{(i,j)})\in \phi\subseteq \psi$. Since $p_{(i,j)}$ was the first propagated literal not contained in $\mathcal{U}$ and all decisions $\alpha$ for $\mathcal{T}$ are in $\mathcal{U}$, all literals from $\mathcal{T}[i,j-1]$ are contained in $\mathcal{U}$. This means that at some point in $\mathcal{U}$ the clause $\ante_\mathcal{T}(p_{(i,j)})$ could have been used to imply $p_{(i,j)}$. Because we are not allowed to skip propagations, we must have propagated $p_{(i,j)}$ somewhere in $\mathcal{U}$, contradicting our assumption.	
	\end{proof}
	
	In the following proposition we will give a simple counting argument. This gives another motivation for using asserting schemes. Informally, it says that we only have to backtrack polynomially often under a specific decision sequence until we reach a desired state of unreliability.
	
	\begin{prop}\label{PropMakeUnreliable}
		Let $\Phi:=\mathcal{Q}\cdot \phi$ be a QCNF in $n$ variables, $\xi$ be an asserting learning scheme, $D$ be a clause  and let $\mathcal{T}$ be a natural $\asonr$ trail for $\Phi$ with decision set $\bar{D}$ and $\bot\in \mathcal{T}$. Then there exists a clause $E$ and a $\asonr$-proof 
		\begin{align*}
		\iota=((\mathcal{T}_1,\ldots,\mathcal{T}_{f_n}),(\xi(\mathcal{T}_1),\ldots,\xi(\mathcal{T}_{f_n})),(\pi_1,\ldots,\pi_{f_n}))
		\end{align*}
		from $\Phi$ of $E$ such that $|\iota|\in \mathcal{O}(n^3)$.  If $E\neq (\bot)$, then $D$ is unreliable with respect to $\mathcal{Q}\cdot (\phi\cup\{  \xi(\mathcal{T}_1),\ldots,\xi(\mathcal{T}_{f_n}) \})$. 
		
		In particular, in case $E\neq (\bot)$ we can find a trail $\mathcal{W}$ that witnesses the unreliability of $D$ after having learnt the clauses $\xi(\mathcal{T}_1),\ldots,\xi(\mathcal{T}_{f_n})$ in the proof $\iota$, i.e., there is a trail $\mathcal{W}$ for $\mathcal{Q}\cdot (\phi\cup\{  \xi(\mathcal{T}_1),\ldots,\xi(\mathcal{T}_{f_n}) \})$, an existential literal $\ell\in D$ with $\ell\in \mathcal{W}$ and decisions $\alpha\subseteq\bar{D}\backslash\{ \bar{\ell} \}$ that follow the same order as the decisions of $\mathcal{T}$.
	\end{prop}

	
	\begin{proof}  Set $\mathcal{T}_1:=\mathcal{T}$. We will now construct a proof $\iota$ with a sequence of trails $\mathcal{T}_1,\ldots,\mathcal{T}_{f_n}$ for some $f_n\in \N$. Arguing inductively, suppose that $\mathcal{T}_i$ was created with decisions $\bar{D}$ and runs into a conflict. We start clause learning and derive $\xi(\mathcal{T}_i)$, which is an asserting clause. After that we will backtrack to a point at which $\xi(\mathcal{T}_i)$ becomes unit, say $\mathcal{T}_i[s_i,t_i]$, and go on constructing the next trail $\mathcal{T}_{i+1}$. From there on we will complete the trail by choosing the same decision literals in the same order as before (while still considering the situations described in Remark~\ref{RemConstructTrailWithDecisions}).    
		Either $\mathcal{T}_{i+1}$ also runs into a conflict, in which case we repeat this whole process, or the decisions $\bar{D}$ block each other in $\mathcal{T}_{i+1}$. If this happens, or if we derive $(\bot)$, we stop. Note that we will always follow the same decision order, even if we skip some decisions. This proves the last statement of the proposition. 
		
		We will argue that the number of these backtracking steps is polynomially bounded, in fact $f_n\in \mathcal{O}(n^2)$. For a variable $z\in\text{Var}(\Phi)$ and a trail $\mathcal{U}$ we define $\nu_\mathcal{U}(z)$ as the level in which $z$ was propagated or decided in $\mathcal{U}$, regardless of the polarity of $z$. If $z$ does not occur in $\mathcal{U}$ in either polarity, then we set $\nu_\mathcal{U}(z):=\infty$. Let $\delta$ be the map defined as follows: 
		\begin{align*}
		\delta:&\:[f_n]\times \var(\phi)\longrightarrow \{0,\ldots,n,\infty\}\\
		\delta(1,z)&:=\nu_{\mathcal{T}_1}(z),\\
		\delta(i,z)&:=\min(\delta(i-1,z),\nu_{\mathcal{T}_{i}}(z)) \quad\text{ for $i\in \{ 2,\ldots,f_n \}$.}
		\end{align*}
		Intuitively, $\delta(i,z)$ returns the smallest decision level in which the variable $z$ occurred in the trails $\{ \mathcal{T}_1,\ldots,\mathcal{T}_i \}$ in any polarity.  By construction, we get $\delta(i+1,z)\leq \delta(i,z)$ for all $i\in [f_n-1]$ and $z\in \var(\phi)$. Furthermore, because $\xi$ is asserting,   we can find $y_i\in \var(\phi)$ (e.g.\ the asserting literal in step $i$) with $\delta(i+1,y_i)< \delta(i,y_i)$ for each $i\in [f_n-1]$. This follows from the fact that  by definition we have to backtrack at least one level. We have to finish after at most $\mathcal{O}(n^2)$ steps since after that $\delta$ would return $0$ for each variable. Therefore $f_n\in \mathcal{O}(n^2)$ and $|\iota|\in \mathcal{O}(n^3)$.
	\end{proof}
	
	The proof of this proposition not only confirms the existence of such a proof $\iota$, it also gives us an algorithm for creating it (although that is not essential for us here). The connection between $\mathcal{T}$ and $\iota$ might not seem obvious at first sight, as $\mathcal{T}$ only serves as a witness in order to start the process of constructing $\iota$.
	
	Let us outline the rest of this section: we aim to construct a $\asonr$ refutation from a \qres refutation. For this we will go through all clauses in the \qres refutation and make them unreliable using Proposition~\ref{PropMakeUnreliable}. For this purpose we have to repeatedly find a natural trail $\mathcal{T}$ with decision set $\bar{D}$ that fulfils the postulated properties. There are three ways a clause could have been derived: as an axiom, via resolution or via reduction. Therefore the next three lemmas will concentrate on this goal.

	In the following results we will assume that all clauses, which we intend to make unreliable, are in fact reliable at the beginning. This guarantees that all construction steps can be correctly carried out. However, since we want to p-simulate \qres (which includes computing the simulation in polynomial time), we have to argue that we can efficiently find witnesses of unreliability. For example, axioms will typically be unreliable (purely universal clauses are not, for instance), but we might not know a witness in advance (although this witness exists by definition). What we can do is to act as if  the axiom is reliable and try to perform the construction steps as described below. Dropping the assumption of reliability, we would lose the 	guarantee that these construction steps work correctly. However, if not we will obtain a witness of unreliability just by constructing this trail, which similarly serves our purpose.  We will discuss this in greater detail when proving Theorem~\ref{TheoremSimComplete}.

	We start with considering the axiom case.

	\begin{lem}\label{LemmaSimAxiom}
		Let $\Phi:=\mathcal{Q}\cdot \phi$ be a QCNF in $n$ variables and $C\in \phi$.  If $C$ is reliable with respect to $\Phi$, 
		there exists a $\asonr$-proof $\iota$ with  trails $\mathcal{T}_1,\ldots,\mathcal{T}_{f_n}$ from $\Phi$ of some clause $E$ that uses the learning scheme $\xi$ such that $|\iota|\in \mathcal{O}(n^3)$. If $E\neq(\bot)$, then $C$ is unreliable with respect to $\mathcal{Q}\cdot (\phi\cup\{  \xi(\mathcal{T}_1),\ldots,\xi(\mathcal{T}_{f_n}) \})$.
	\end{lem}
	
	\begin{proof} Construct a (natural) trail $\mathcal{T}$ by choosing $\bar{C}$ as level-ordered decisions. If we do not run into a conflict, we get a contradiction since we falsified $C$ and are not allowed to skip conflicts. After applying Proposition \ref{PropMakeUnreliable} we have either derived $(\bot)$, or $C$ becomes unreliable.
	\end{proof}

	The next two results cover the simulation of resolution and reduction steps. As before, we assume that the clause for which we intend to find a witness of unreliability  is actually reliable at the beginning. Dropping this assumption removes the guarantee that the decisions will not block each other in the construction. But then we might find a witness of unreliability even earlier.  
	
	For these next results, we will use the notation $\theta(\iota):=(\mathcal{T}_1,\ldots,\mathcal{T}_m)$, $\lambda(\iota):=(C_1,\ldots,C_m)$ and $\rho(\iota):=(\pi_1,\ldots,\pi_m)$ for a $\asonr$ proof
	\begin{align*}
		\iota=((\mathcal{T}_1,\ldots,\mathcal{T}_m),(C_1,\ldots,C_m),(\pi_1,\ldots,\pi_m)).
	\end{align*} 
	
	\begin{lem}\label{PropSimResolution}
		Let $\Phi:=\mathcal{Q}\cdot \phi$ be a QCNF in $n$ variables. Also let $C_1\vee x$ be a clause that is unreliable with respect to $\Psi:=\mathcal{Q}\cdot\psi$ with $\psi\subseteq\phi$ and $C_2\vee \bar{x}$ unreliable with respect to $\Upsilon:=\mathcal{Q}\cdot \tau$ with $\tau\subseteq\phi$, such that $C_1\vee C_2$ is non-tautological. Let $\xi$ be an asserting learning scheme. If $C_1\vee C_2$ is reliable with respect to $\Phi$, there exists a $\asonr$-proof $\iota$ with  $\theta(\iota)=\mathcal{T}_1,\ldots,\mathcal{T}_{f_n}$ from $\Phi$ of some clause $E$ that uses the learning scheme $\xi$ such that $|\iota|\in \mathcal{O}(n^3)$. If $E\neq(\bot)$, then $C_1\vee C_2$ is unreliable with respect to $\mathcal{Q}\cdot (\phi\cup\{  \xi(\mathcal{T}_1),\ldots,\xi(\mathcal{T}_{f_n}) \})$.
		
	\end{lem}
	
	\begin{proof} Since $C_1\vee x$ is unreliable, there is a literal $\ell_1\in C_1\vee x$ and a trail $\mathcal{U}_1$ for $\Psi$ with decisions $\alpha_1\subseteq \overline{(C_1\vee x)}\backslash \{ \bar{\ell}_1 \}$ such that $\ell_1\in \mathcal{U}_1$. The same is true for $C_2\vee  \bar{x}$, thus we get a trail $\mathcal{U}_2$ for $\Upsilon$, a literal $\ell_2$ and decisions $\alpha_2\subseteq\overline{(C_2\vee \bar{x})}\backslash\{ \bar{\ell_2} \}$.

		We now distinguish three cases. In each case we will define a set of decisions that are not blocking each other, construct a natural trail with these decisions (see Remark \ref{RemConstructTrailWithDecisions}) and run into a conflict. After that we will unite all three cases and start clause learning. The created trail will serve as a starting point for Proposition \ref{PropMakeUnreliable}.
		
		\underline{Case 1:} $\ell_1=x$ and $\ell_2=\bar{x}$.
		
		We choose the set $\alpha_1\cup \alpha_2\subseteq \bar{C}$ as level-ordered decisions for a new, natural trail $\mathcal{T}$. These decisions cannot block each other since $C$ is still reliable. If we assume that $\mathcal{T}$ does not run into a conflict, then we get $\alpha_1\cup \alpha_2\subseteq \mathcal{T}$. We can now apply Lemma \ref{LemmaFindEveryPropagation} and conclude that all propagated literals from $\mathcal{U}_1$ and $\mathcal{U}_2$ are contained in $\mathcal{T}$. However, this is a contradiction since we would need to propagate both $\ell_1=x$ and $\ell_2=\bar{x}$. Therefore $\mathcal{T}$ has to run into a conflict.
		
		\underline{Case 2:} $\ell_1=x$ and $\ell_2\neq \bar{x}$ (or analogously $\ell_1\neq x$ and $\ell_2=\bar{x}$).
		
		We choose the set $(\{ \bar{\ell}_2 \} \cup \alpha_1\cup\alpha_2)\backslash\{ x \}\subseteq \bar{C}$ as level-ordered decisions for a natural trail $\mathcal{T}$. As before, these decisions are not blocking each other and we can at least do the same propagations as in $\mathcal{U}_1$ as long as we do not run into a conflict because of $\alpha_1\subseteq\mathcal{T}$. In particular we are able to propagate $x=\ell_1$ (e.g. via $\ante_{\mathcal{U}_1}(\ell_1)$). Now we have decided or propagated all decisions of $\mathcal{U}_2$, i.e., $\alpha_2\subseteq \mathcal{T}$. Hence after applying Lemma \ref{LemmaFindEveryPropagation} again we can do at least all propagations of $\mathcal{U}_2$. But then we would need to propagate $\ell_2$ in $\mathcal{T}$. That is contradictory to $\bar{\ell_2}\in \mathcal{T}$.  This means likewise $\mathcal{T}$ has to run into a conflict.

		\underline{Case 3:} $\ell_1\neq x$ and $\ell_2\neq \bar{x}$.
		
		We know that $\bar{x}\in \alpha_1$, otherwise $C$ would be unreliable. We choose $ \{\bar{\ell}_1\} \cup\alpha_1$ as decisions for the natural trail $\mathcal{T}$, but we demand that $\bar{x}$ will be decided last and all the other decisions are level-ordered. 
		
		The decisions before $\bar{x}$ cannot block each other, since its negations are literals in $C$. 
		
		If we propagated $x$ somewhere, we can instead start with a new natural trail $\mathcal{T}'$ using the non-blocking level-ordered decisions $(\{ \bar{\ell}_1, \bar{\ell}_2 \}\cup\alpha_1\cup \alpha_2)\backslash \{ x,\bar{x} \}\subseteq \bar{C}$.  Provided that $\mathcal{T'}$ does not run into a conflict, applying Lemma \ref{LemmaFindEveryPropagation} gives us $x\in \mathcal{T}'$ since we already implied $x$ in $\mathcal{T}$.  Therefore we have decided or propagated all decisions of $\mathcal{U}_2$, i.e., $\alpha_2\subseteq \mathcal{T}'$. Hence we get a contradiction by Lemma \ref{LemmaFindEveryPropagation} because we would need to propagate $\ell_2$. Therefore $\mathcal{T}'$ runs into a conflict.

		If we actually decide or propagate $\bar{x}$, we will run into a conflict afterwards. Otherwise we would have made all propagations from $\mathcal{U}_1$ (since $\alpha_1\subseteq \mathcal{T}$), receiving a contradiction by Lemma \ref{LemmaFindEveryPropagation} again.\\

		Using Proposition \ref{PropMakeUnreliable}, for Case 1, Case 2 and the first part of Case 3 (where we propagated $x$ before we could decide $\bar{x}$) we can construct a $\asonr$-proof $\iota$ with  $|\iota|\in \mathcal{O}(n^3)$ that fulfils the desired properties.  Note that in each case we have followed the policy $\aso$. 
		
		The last part of Case 3 (where we actually decide or propagate $\bar{x}$) works slightly different. Here our decisions were $\{ \bar{\ell}_1 \}\cup \alpha_1\subseteq \overline{(C_1\vee x)}$, therefore we also apply Proposition \ref{PropMakeUnreliable} and construct a $\asonr$-proof $\iota'$ with  $\theta(\iota')=\mathcal{T}'_1,\ldots,\mathcal{T}'_{f_n'}$ from $\Phi$ of a clause $E'$ that uses $\xi$ such that $|\iota'|\in \mathcal{O}(n^3)$. If we have $E'=(\bot)$, we are done. If not, then 
		$C_1\vee x$ became unreliable with respect to $\mathcal{Q}\cdot (\phi\cup\{  \xi(\mathcal{T}'_1),\ldots,\xi(\mathcal{T}'_{f_n'}) \})$. Here we need the last statement of Proposition \ref{PropMakeUnreliable}: There is a trail $\mathcal{W}$ for $\mathcal{Q}\cdot (\phi\cup\{  \xi(\mathcal{T}'_1),\ldots,\xi(\mathcal{T}'_{f_n'}) \})$, an existential literal $\ell\in C_1\vee x$  with $\ell\in \mathcal{W}$ and decisions $\alpha\subseteq (\overline{C_1\vee x})\backslash\{ \bar{\ell} \}$ that follow the same order as the decisions of $\mathcal{T}$. Since in this order the literal $\bar{x}$ was always the last literal, we conclude $\bar{x}\not\in \alpha$ (otherwise we would not have a chance to achieve a situation where the decisions block each other). There are two remaining possibilities: If $\ell\neq x$, then $\mathcal{W}$ is a witness for the unreliability of $C_1$ (and therefore also $C_1\vee C_2$), which gives us the desired result. 
		
		However, if $\ell=x$, we can go back to Case 2 and use the trails $\mathcal{W}$ and $\mathcal{U}_2$. We create another $\asonr$-proof $\iota''$ with  $\theta(\iota'')=\mathcal{T}_1'',\ldots,\mathcal{T}''_{f''_n}$ from  $\mathcal{Q}\cdot (\phi\cup\{  \xi(\mathcal{T}'_1),\ldots,\xi(\mathcal{T}'_{f_n'}) \})$ of a clause $E$ that uses the learning scheme $\xi$ such that $|\iota''|\in \mathcal{O}(n^3)$. We can combine these two proofs $\iota'$ and $\iota''$ into a proof $\iota$ with $\theta(\iota)=\mathcal{T}'_1,\ldots,\mathcal{T}'_{f_n'},\mathcal{T}''_{1},\ldots,\mathcal{T}''_{f_n''}$ from $\Phi$ of the clause $E$ by connecting the components $\theta(\iota')$ with $\theta(\iota'')$, $\lambda(\iota')$ with $\lambda(\iota'')$ and $\rho(\iota')$ with $\rho(\iota'')$ such that $|\iota|\in \mathcal{O}(n^3)$. Between the trails $\mathcal{T}'_{f'_n}$ and $\mathcal{T}''_1$ we backtrack to the time $(0,0)$ (i.e., we restart the trail). If $E\neq (\bot)$, then $C_1\vee C_2$ became unreliable with respect to $\mathcal{Q}\cdot (\phi\cup\{  \xi(\mathcal{T}'_1),\ldots,\xi(\mathcal{T}'_{f_n'}),\xi(\mathcal{T}''_{1}),\ldots, \xi(\mathcal{T}''_{f''_n}) \})$.
	\end{proof}

	It remains to show a similar result for the reduction step.
	
	\begin{lem}\label{PropSimReduction}
		Let $\Phi:=\mathcal{Q}\cdot \phi$ be a QCNF in $n$ variables, let $D:=C\vee u_1\vee\ldots\vee u_m$ be a non-tautological clause with universal literals $u_1,\ldots,u_m$ and $\red(D)=C$, such that $D$ is unreliable with respect to a QCNF $\Psi=\mathcal{Q}\cdot \psi$ with $\psi\subseteq \phi$. 
		Let $\xi$ be an asserting learning scheme. If $C$ is reliable with respect to $\Phi$, there exists a $\asonr$-proof $\iota$ with  $\theta(\iota)=\mathcal{T}_1,\ldots,\mathcal{T}_{f_n}$ from $\Phi$ of some clause $E$ that uses the learning scheme $\xi$ such that $|\iota|\in \mathcal{O}(n^3)$. If $E\neq(\bot)$, then $C$ is unreliable with respect to $\mathcal{Q}\cdot (\phi\cup\{  \xi(\mathcal{T}_1),\ldots,\xi(\mathcal{T}_{f_n}) \})$.
	\end{lem}
	
	\begin{proof} Because $D$ is unreliable, we can find a literal $\ell\in D$ and a trail $\mathcal{U}$ for $\Psi$ with decisions $\alpha\subseteq \bar{D}\backslash\{ \bar{\ell}  \}$ and $\ell\in \mathcal{U}$. The literal $\ell$ has to be existential, hence $\ell\in C$.
		
		Now we choose the set $\{ \bar{\ell} \}\cup \alpha\subseteq \bar{D}$ as level-ordered decisions for the natural trail $\mathcal{T}$, i.e., we decide the literals in $(\{ \bar{\ell} \}\cup \alpha)\cap \{  \bar{u}_1,\ldots,\bar{u}_m \}$ at the end.  The decisions before    $(\{ \bar{\ell} \}\cup \alpha)\cap \{  \bar{u}_1,\ldots,\bar{u}_m \}$  cannot block each other since $C$ is reliable. After this we would only decide universal literals, which can not be propagated. Therefore, in this order, the decisions $\{ \bar{\ell} \}\cup \alpha$ are not blocking each other. If $\mathcal{T}$ does not run into a conflict, we can at least do the same propagations as in $\mathcal{U}$ by Lemma \ref{LemmaFindEveryPropagation}. Then we would get the contradiction $\ell,\bar{\ell}\in \mathcal{T}$.
		
		Now we can apply Proposition \ref{PropMakeUnreliable} and construct a proof $\iota$ with  $\theta(\iota)=\mathcal{T}_1,\ldots,\mathcal{T}_{f_n}$ from $\Phi$ of a clause $E$ such that $|\iota|\in \mathcal{O}(n^3)$. If $E\neq (\bot)$, then $D$ is unreliable with respect to $\mathcal{Q}\cdot (\phi\cup\{  \xi(\mathcal{T}_1),\ldots,\xi(\mathcal{T}_{f_n}) \})$. In this case there is a trail $\mathcal{W}$ for $\mathcal{Q}\cdot (\phi\cup\{  \xi(\mathcal{T}_1),\ldots,\xi(\mathcal{T}_{f_n}) \})$, an existential literal $y\in D$ with $y\in \mathcal{W}$ and decisions $\beta\subseteq \bar{D}\backslash\{  \bar{y} \}$ that follows the same order as the decisions of $\mathcal{T}$ (level-order). Since $y$ is existential, we have $y\in C$. The blocking situation in $\mathcal{W}$ had to occur before we could decide $\bar{u}_1,\ldots,\bar{u}_m$. We conclude $\beta\subseteq \bar{C}\backslash \{  \bar{y} \}$ and thus $C$ is unreliable with respect to $\mathcal{Q}\cdot (\phi\cup\{  \xi(\mathcal{T}_1),\ldots,\xi(\mathcal{T}_{f_n}) \})$.
	\end{proof}
	
	Now we can combine these three auxiliary results into the main theorem of this section.
	
	\begin{thm}\label{TheoremSimComplete}
		$\asonr$ p-simulates \qres. I.e., each \qres refutation $\pi$ of a QCNF in $n$ variables can be transformed into a $\asonr$-refutation of size $\mathcal{O}(n^3\cdot |\pi|)$ that uses an arbitrary asserting learning scheme $\xi$.
		
		In particular, \qres, $\aonr$ and $\asonr$ are p-equivalent proof systems.
		
	\end{thm}	
	\begin{proof} 
		First we show that $\asonr$ simulates \qres. 
		Let $\Phi:=\mathcal{Q}\cdot \phi$ be a QCNF in $n$ variables with a \qres refutation $\pi=C_1,\ldots,C_k$. We will go through this proof from left to right and check whether or not the clauses are reliable with respect to the corresponding current QCNF.  Suppose that $C_1,\ldots,C_{i-1}$ are already unreliable with respect to QCNFs $\mathcal{Q}\cdot \phi_1,\ldots,\mathcal{Q}\cdot \phi_{i-1}$ with $\phi \subseteq \phi_1 \subseteq\ldots\subseteq\phi_{i-1}$. If $C_i$ is unreliable with respect to $\mathcal{Q}\cdot \phi_{i-1}$, then we can set $\phi_i:=\phi_{i-1}$ (let $\phi_0:=\phi$), $\iota_i:=(\emptyset,\emptyset,\emptyset)$ (empty proof) and continue with $C_{i+1}$.
		
		However, if $C_i$ is reliable with respect to $\mathcal{Q}\cdot \phi_{i-1}$, then we can apply either Lemma \ref{LemmaSimAxiom} or Lemma~\ref{PropSimResolution} or Lemma~\ref{PropSimReduction}, depending whether $C_i$ is an axiom, or was derived via resolution or reduction. We construct a $\asonr$-proof $\iota_i$ with  $\theta(\iota_i)=\mathcal{T}_1^{(i)},\ldots,\mathcal{T}_{f^{(i)}_n}^{(i)}$ of a clause $E_i$ from $\mathcal{Q}\cdot ( \phi_{i-1}\cup\{  \xi(\mathcal{T}_1^{(i)}),\ldots,\xi(\mathcal{T}_{f^{(i)}_n}^{(i)}) \}  )$ with $|\iota_i|\in \mathcal{O}(n^3)$ that uses the learning scheme $\xi$. If $E_i=(\bot)$, we are done and connect the components of the proofs $\iota_1,\ldots,\iota_i$, in particular $\theta(\iota)=\theta(\iota_1),\ldots,\theta(\iota_i)$. Otherwise $C_i$ became unreliable with respect to $\mathcal{Q}\cdot ( \phi_{i-1}\cup\{  \xi(\mathcal{T}_1^{(i)}),\ldots,\xi(\mathcal{T}_{f^{(i)}_n}^{(i)}) \}  )$. In this case we set $\phi_i:=\phi_{i-1}\cup\{  \xi(\mathcal{T}_1^{(i)}),\ldots,\xi(\mathcal{T}_{f^{(i)}_n}^{(i)}) \}$ and continue with the next clause $C_{i+1}$. 
		
		At the latest when we reach $C_k=(\bot)$ in $\pi$, we will create a refutation $\iota$ since $(\bot)$ can never become unreliable by definition. The $\asonr$-refutation $\iota$ consists of some $\asonr$-proofs $\iota_1,\ldots,\iota_j$, $1\leq j\leq k$. Between two proofs $\iota_a$ and $\iota_{a+1}$ we will just restart, i.e., between the trails $\mathcal{T}_{f_n^{(a)}}^{(a)}$ and $\mathcal{T}_1^{(a+1)}$ we will backtrack to the point $(0,0)$. At the end we have $|\iota|\in \mathcal{O}(n^3\cdot |\pi|)$.
		
		We have now shown \qres$\leq \asonr$. In order to actually prove that this simulation is polynomial-time computable, one should actually argue that the described construction steps are in fact polynomial-time computable. However, we needed to decide whether or not a clause in the given proof is reliable. This might not be computable in polynomial time. Alternatively, we can pretend that a clause $C$ is reliable unless we have found a witness that proves the opposite. We can do the exact same steps as described in the above results. In detail, we can still try to create the trail $\mathcal{T}$ from Lemma~\ref{LemmaSimAxiom}, Lemma~\ref{PropSimResolution} or Lemma~\ref{PropSimReduction}. If the decisions do not block each other, we proceed as if the clause $C$ was reliable. Otherwise, we immediately receive a witness for our unreliability, even if we were not able to take the steps as described.
		Therefore  $\asonr$ p-simulates \qres.
		
		By Theorem \ref{TheoremSystemsSimulatedByQRes} the systems $\aonr$ and $\asonr$ are p-simulated by \qres. Obviously, $\asonr$ is p-simulated by $\aonr$. As a consequence, all of these three systems are p-equivalent.	
	\end{proof}

	\section{Comparison to the correspondence between\texorpdfstring{\\}{ }propositional re\-solution and CDCL}
	\label{sec:prop-cdcl}

	There are a few similarities between our definitions of  \emph{unreliable}/\emph{reliable} and the notions \emph{1-empowering}/\emph{absorbed} used in \cite{DBLP:journals/ai/PipatsrisawatD11}. Pipatsrisawat and Darwiche  \cite{DBLP:journals/ai/PipatsrisawatD11} focused on propositional logic and CNFs, so let us for this section restrict our attention to the propositional case (e.g.\ by considering only QCNFs with existential quantifiers) and recall the definition that was used in the work of Pipatsrisawat and Darwiche.
	
	\begin{defiC}[\cite{DBLP:journals/ai/PipatsrisawatD11}] \label{def:1-empowering}
		Let $\alpha \Rightarrow \ell$ be a clause where $\ell$ is some literal and $\alpha$ is a conjunction of literals. The clause is \emph{1-empowering} with respect to CNF $\Delta$ iff
		\begin{enumerate}
			\item $\Delta \vDash (\alpha \Rightarrow \ell)$: the clause is implied by $\Delta$.
			\item $\Delta \wedge \alpha$ is 1-consistent: asserting $\alpha$ does not result in a conflict that is detectable by unit resolution.
			\item $\Delta \wedge \alpha \nvdash_1\ell$: the literal $\ell$ cannot be derived from $\Delta\wedge \alpha$ using unit resolution.
		\end{enumerate}
		In this case, $\ell$ is called an \emph{empowering literal} of a clause. On the other hand, a clause implied by $\Delta$ that contains no empowering literals is said to be \emph{absorbed} by $\Delta$.
		
	\end{defiC}
	
	For further details concerning the notations see \cite{DBLP:journals/ai/PipatsrisawatD11}. 
	
	Translating the third point into our framework could lead to the following interpretation: $\ell$ cannot be propagated in a trail for $\Delta$ with decisions $\alpha$. This is related to our definition of reliability. Simply put, for reliability we require that no literal in the clause is derivable by unit propagation. In Definition~\ref{def:1-empowering}, however, it suffices to find at least one literal $\ell$ with this property.  
	
	Consequently the differences between `unreliable' and `absorbing' could be formulated as follows: for a clause $C$ to be unreliable, we need at least one literal $\ell\in C$ such that $\ell$ is `accidentally' propagated in a trail with decisions contained in $\bar{C}\setminus\{\bar{\ell}\}$. In an absorbed clause, on the other hand, each literal has to be propagated accidentally in this way. 
	
	In \cite{DBLP:journals/ai/PipatsrisawatD11} this difference in definition caused an additional factor $n$ in the complexity of the \cdcl simulation of resolution. Roughly speaking, their idea consists of searching for 1-empowering (and 1-provable) clauses in a given \textsf{resolution} refutation $\pi$ of a CNF $\Delta$. They described how these clauses get absorbed after $\mathcal{O}(n^4)$ \cdcl steps, where the last $n$-factor is incurred by the fact that they have to handle all empowering literals, not just one as in our results. The remaining factor $n^3$ can be explained in a similar way as in Proposition \ref{PropMakeUnreliable} (cf.\ \cite[Prop.~3]{DBLP:journals/ai/PipatsrisawatD11}).
	
	Consequently, we obtain a slight quantitative improvement of the simulation of resolution by CDCL \cite{DBLP:journals/ai/PipatsrisawatD11} from $\mathcal{O}(n^4|\pi|)$ to $\mathcal{O}(n^3|\pi|)$.
	\begin{thm} \label{thm:cdcl-res}
		Let $\phi$ be a CNF in $n$ variables and let $\pi$ be a resolution refutation of $\phi$. Then $\phi$ has a CDCL refutation of size $\mathcal{O}(n^3 |\pi|)$.
	\end{thm}

	An advantage of the `1-empowering/absorbed' notion is the simplification when it comes to cover the resolution steps in the given proof $\pi$. In the proof of Lemma~\ref{PropSimResolution} we had to distinguish three cases depending on the literal that witnessed the unreliability. However, this is not necessary when using the definition of `absorption'. Since in this case all literals from a clause shall be propagated accidentally, we can pick an arbitrary literal that simplifies the following reasoning steps. In fact,  it then suffices to consider Case 1 in Lemma \ref{PropSimResolution}. 
	
	Furthermore, in \cite{DBLP:journals/ai/PipatsrisawatD11} the authors refrained from using the concept of trails or introducing algorithm-based proof systems. Instead, they relied on the notions of unit resolution, 1-consistency, and 1-provability. However, these notions cannot be fully translated  into our framework that enables the construction of trails as it is done in practical CDCL. For example, consider the following:
	
	\smallskip
	Let $\Delta:=(\bar{x}_1\vee\bar{x}_2\vee \bar{x}_3)$ be a CNF consisting only of one clause. Then clearly this CNF together with decisions $x_1,x_2,x_3$ is 1-inconsistent (one could think of 1-inconsistent \cdcl states, that are in the form of $\Delta\wedge \ell_1\wedge \ldots\wedge \ell_m$ with a CNF $\Delta$ and decision literals $\ell_1,\ldots,\ell_m$, as formulas which are refutable via unit propagation). However, this conflict cannot occur in practice when creating trails as described in our work (i.e., in a natural sense). We would obtain a situation where the decisions block each other:
	\begin{align*}
	\mathcal{T}= (\mathbf{x_1};\mathbf{x_2},\bar{x}_3 )\enspace.
	\end{align*}
	This blocking can only be resolved by skipping the propagation $\bar{x}_3$. Of course, we can observe this artificial conflict by our notion of trails as well:
	\begin{align*}
	\mathcal{T}'= (\mathbf{x_1};\mathbf{x_2};\mathbf{{x}_3},\bot )\enspace.
	\end{align*}
	But since we want to avoid this inherent kind of non-determinism (cf.\ also \cite{DBLP:journals/jair/BeameKS04}), we defined our proof systems in such a way as  to circumvent these situations. As a consequence, it is impossible for the trail $\mathcal{T}'$ to appear in a proof under one of our QCDCL systems.  This distinction between non-deterministic, more liberal systems and algorithm-based, stricter versions of these systems seems in some way harder to clarify with the notions used by Pipatsrisawat and Darwiche.

	\section{The simulation order of QCDCL proof systems}
	\label{sec:sim-order}
	
	Now that we characterised the complexity of classical systems, such as \qcdcl and \qres, we want to examine the connections between the remaining QCDCL systems that we have not fully considered yet. We refer again to Figure~\ref{fig:qcdcl-sim-order} on page~\pageref{fig:qcdcl-sim-order}, depicting the resulting simulation order.
	
	First we define the formulas $\Lon_n$ that were introduced by Lonsing in \cite{LonsingDissertation}. Originally, these QCNFs were constructed to separate QBF solvers that differ in the implemented dependency schemes (we will not consider these concepts here, though). 
	
	\begin{defiC}[\cite{LonsingDissertation}]
		Let $\Lon_n$ be the QCNF
		\begin{align*}
		\exists a,b,b_1,\ldots,b_{s_n}\forall x,y\exists c,d\cdot &(a\vee x\vee c)\wedge(a\vee b\vee b_1\vee \ldots\vee b_{s_n})\wedge (b\vee y\vee d)\wedge (x\vee c)\\
		& \wedge (x\vee \bar{c})\wedge\mathtt{PHP}_n^{n+1}(b_{1},\ldots,b_{s_n})\enspace.
		\end{align*}
	\end{defiC}
	
	It was shown in \cite{LonsingDissertation} that this formula becomes easy to refute by choosing the standard dependency scheme. However, $\Lon_n$ serves as a witness for separating our systems as well.
	
	
	\begin{prop}\label{PropLonsingsFormula}
		The QCNFs $\Lon_n$ require exponential-size proofs in the proof systems $\loar$ and $\lonr$, but have constant-size proofs in  $\asroar$ and \qres.
	\end{prop}	
	\begin{proof} With the policy $\lo$ we are forced to start assigning the variables $a,b,b_1,\ldots,b_{s_n}$. As long as we do this, we can only use the clauses from $(a\vee b\vee b_1\vee \ldots\vee b_{s_n})\wedge \mathtt{PHP}_n^{n+1}(b_{1},\ldots,b_{s_n})$  as antecedent clauses. Since $\mathtt{PHP}_n^{n+1}$ contains subclauses of $b_1\vee \ldots\vee b_{s_n}$, we do not need the clause $a\vee b\vee b_1\vee \ldots\vee b_{s_n}$ either. The propositional formula $\mathtt{PHP}_n^{n+1}$ is unsatisfiable, therefore we will falsify this formula before reaching the decisions $x$ and $y$. We will always learn clauses $C$ whose \ldqres proofs consist only of axioms from $\mathtt{PHP}_n^{n+1}$. These proofs are in fact \resolution proofs because the contained clauses do not include any universal variables. At the end we obtain a \qcdcl-proof $\iota$ (in the system $\loar$ or $\lonr$) that contains a \resolution refutation of $\mathtt{PHP}_n^{n+1}$, which is exponential by  \cite{Hak85}. Hence also $|\iota|$ is exponential.

		We obtain short \qres refutations of $\Lon_n$ by $\red\left( x\vee c\resop{c} x\vee \bar{c}   \right)=(\bot)$. 
		The following trail yields a short $\asroar$ refutation of $\Lon_n$:
		\begin{align*}
		\mathcal{T}=(\mathbf{\bar{x}},c,\bot)\text{.}
		\end{align*}
		From this we can learn $(\bot)$ in the same way as in the \qres refutation.
	\end{proof}

	Next we want to compare the policies $\ar$ and $\nr$ when fixing the decision policy $\lo$. As we indicated before, these two policies seem to operate orthogonally to each other. We will prove this intuition now, again using the $\mathtt{Trapdoor}_n$ formulas.
	
	\begin{prop}\label{PropTrapdoorEasyForLONR}
		The QCNFs $\mathtt{Trapdoor}_n$ have polynomial-size $\lonr$ refutations.
	\end{prop}	
	\begin{proof}
		
		The refutation consists of the trails $\mathcal{T}_1,\mathcal{T}_2$:
		\begin{align*}
		\mathcal{T}_1:=( \mathbf{y_1};
		\mathbf{y_2};
		\ldots;
		\mathbf{y_{s_n}};
		\mathbf{\bar{w}},t,\bot)
		\end{align*}
		with $\ante_{\mathcal{T}_1}(t)=\bar{y}_1\vee w\vee {t}$ and $\ante_{\mathcal{T}_1}(\bot)=\bar{y}_1\vee w\vee \bar{t}$. We learn the clause $(\bar{y}_1)$ and backtrack to $(0,0)$.
		\begin{align*}
		\mathcal{T}_2:=( \bar{y}_1;
		\mathbf{\bar{y}_2};
		\ldots;
		\mathbf{\bar{y}_{s_n}};
		\mathbf{\bar{w}}, {t},\bot )
		\end{align*}
		with $\ante_{\mathcal{T}_2}({t})={y}_1\vee w\vee {t}$ and $\ante_{\mathcal{T}_2}(\bot)={y_1}\vee w\vee \bar{t}$. We finally learn the clause $(\bot)$.
	\end{proof}
	
	Combined with previous results, we can conclude the following:
	
	\begin{thm}\label{CorLONRandLOARincomparable}\label{TheoremLONRandLOARincomparable}
		The systems $\lonr$ and $\loar$ are incomparable.
	\end{thm}	
	\begin{proof} The QCNFs $\mathtt{QParity}_n$ are hard for $\lonr$ by \cite{BBH19} since this system is p-simulated by \qres, but easy for $\loar$ as proven in Proposition \ref{PropQParityHardForQCDCL}. The formulas $\mathtt{Trapdoor}_n$ are hard for $\loar$ (Proposition \ref{PropTrapdoorHardForQCDCL}), but easy for $\lonr$ (Proposition \ref{PropTrapdoorEasyForLONR}).
	\end{proof}
	
	In Section \ref{sec:hardness-qcdcl} we already introduced a whole class of QCNFs that require large \qcdcl (=$\loar$)  refutations. We can exponentially improve this classical \qcdcl system by exchanging the decision policy $\lo$ for a more liberal one. Although we have already shown this in Proposition~\ref{PropLonsingsFormula}, we will give another, more interesting separation by QBFs whose hardness does not rely on propositional  resolution complexity.

	\begin{prop} \label{prop:equality-easy}
		The formulas $\mathtt{Equality}_n$ have polynomial-size $\asroar$ refutations.
	\end{prop}	
	\begin{proof} First we define the clauses 
		\begin{align*}
		L_i&:=\bar{x}_i\vee \bar{u}_i\vee \bigvee_{j=i+1}^n(u_j\vee \bar{u}_j)\vee \bigvee_{k=1}^{i-1}\bar{t}_k\text{,}\\
		R_i&:=x_i\vee u_i\vee \bigvee_{j=i+1}^n(u_j\vee \bar{u}_j)\vee \bigvee_{k=1}^{i-1}\bar{t}_k
		\end{align*}
		for $i=2,\ldots,n$.
		
		We will construct $\asroar$ trails $\mathcal{T}_n,\mathcal{U}_n,\ldots,\mathcal{T}_2,\mathcal{U}_2$ from which we learn the clauses $L_n,R_n,\ldots,L_2,R_2$. We will restart after each trail.
		
		The initial trail is
		\begin{align*}
		\mathcal{T}_n=(  \mathbf{x_1};\mathbf{x_2};\ldots;\mathbf{x_n};\mathbf{u_1},t_1;\mathbf{u_2},t_2;\ldots;\mathbf{u_{n-1}},t_{n-1},\bar{t}_n,\bot )
		\end{align*}
		coupled with the antecedent clauses
		\begin{align*}
		\ante_{\mathcal{T}_n}(t_j)&=\bar{x}_j\vee \bar{u}_j\vee t_j \quad\text{ for $j=1,\ldots,n-1$,}\\
		\ante_{\mathcal{T}_n}(\bar{t}_n)&=\bar{t}_1\vee\ldots\vee\bar{t}_n\text{,}\\
		\ante_{\mathcal{T}_n}(\bot)&=\bar{x}_n\vee \bar{u}_n\vee t_n\text{.}
		\end{align*}
		After resolving over $\bar{t}_n$ we learn $L_n$.
		
		We restart and can create $\mathcal{U}_n$, symmetrically to $\mathcal{T}_n$:
		\begin{align*}
		\mathcal{U}_n=(  \mathbf{\bar{x}_1};\mathbf{\bar{x}_2};\ldots;\mathbf{\bar{x}_n};\mathbf{\bar{u}_1},t_1;\mathbf{\bar{u}_2},t_2;\ldots;\mathbf{\bar{u}_{n-1}},t_{n-1},\bar{t}_n,\bot )
		\end{align*}
		with 
		\begin{align*}
		\ante_{\mathcal{U}_n}(t_j)&={x}_j\vee {u}_j\vee t_j \quad\text{ for $j=1,\ldots,n-1$,}\\
		\ante_{\mathcal{U}_n}(\bar{t}_n)&=\bar{t}_1\vee\ldots\vee\bar{t}_n\text{,}\\
		\ante_{\mathcal{U}_n}(\bot)&={x}_n\vee {u}_n\vee t_n\text{.}
		\end{align*}
		Analogously, we learn $R_n$.
		
		Now suppose that we already learned the clauses $L_n,\ldots,L_i$ and $R_n, \ldots,R_i$ for $3\leq i\leq n$. Next, let us learn the clause $L_{i-1}$ using the trail $\mathcal{T}_{i-1}$:
		\begin{align*}
		\mathcal{T}_{i-1}=(  \mathbf{x_1};\mathbf{x_2};\ldots;\mathbf{x_{i-1}};\mathbf{u_1},t_1;\mathbf{u_2},t_2;\ldots;\mathbf{u_{i-1}},t_{i-1},\bar{x}_i,\bot )
		\end{align*}
		with
		\begin{align*}
		\ante_{\mathcal{T}_{i-1}}(t_j)&=\bar{x}_j\vee \bar{u}_j\vee t_j \quad\text{ for $j=1,\ldots,i-1$,}\\
		\ante_{\mathcal{T}_{i-1}}(\bar{x}_i)&=L_i\text{,}\\
		\ante_{\mathcal{T}_{i-1}}(\bot)&=R_i\text{.}
		\end{align*}
		We resolve over $\bar{x}_i$ and $t_{i-1}$, obtaining $(R_i\resop{x_i}L_i)\resop{t_{i-1}}(\bar{x}_{i-1}\vee \bar{u}_{i-1}\vee t_{i-1}) =L_{i-1}$.
		
		Again, in a symmetrical way we derive $R_{i-1}$:
		\begin{align*}
		\mathcal{U}_{i-1}=(  \mathbf{\bar{x}_1};\mathbf{\bar{x}_2};\ldots;\mathbf{\bar{x}_{i-1}};\mathbf{\bar{u}_1},t_1;\mathbf{\bar{u}_2},t_2;\ldots;\mathbf{\bar{u}_{i-1}},t_{i-1},\bar{x}_i,\bot )
		\end{align*}
		with
		\begin{align*}
		\ante_{\mathcal{U}_{i-1}}(t_j)&={x}_j\vee {u}_j\vee t_j \quad\text{ for $j=1,\ldots,i-1$,}\\
		\ante_{\mathcal{U}_{i-1}}(\bar{x}_i)&=L_i\text{,}\\
		\ante_{\mathcal{U}_{i-1}}(\bot)&=R_i\text{.}
		\end{align*}

		We end the proof with two trails:
		\begin{align*}
		\mathcal{T}_1=( \mathbf{x_1};\mathbf{u_1},t_1,\bar{x}_1,\bot )\text{,}
		\end{align*}
		with similar antecedent clauses as before, from which we learn the unit clause $(\bar{x}_1)$, and 
		\begin{align*}
		\mathcal{U}_1=( \bar{x}_1;\mathbf{\bar{u}_1},t_1,\bar{x}_1,\bot )\text{,}
		\end{align*}
		from which we finally derive $(\bot)$. This whole proof has size $\mathcal{O}(n^2)$.
	\end{proof}

	Using $\asro$ instead of $\lo$ allowed us to skip existential decisions. As a result we were able to restrict ourselves to the decisions $x_1,\ldots, x_{i-1}$ in the trails $\mathcal{T}_{i-1}$ and $\mathcal{U}_{i-1}$ since the other variables $x_i,\ldots,x_n$ are either resolved away or useless for the current resolution step. 
	
	This leads to the following separation:
	
	\begin{thm}\label{CorASROARdstrongerThanLOAR}\label{TheoremASROARdstrongerThanLOAR}
		$\asroar$ is exponentially stronger than $\loar$.
	\end{thm}
	
	Note that the decision policy $\asro$ in $\asroar$ guarantees the possibility to learn asserting clauses. Having shown that $\asroar$ is actually stronger than classical \qcdcl, the system $\asroar$ seems to be a promising candidate for practical implementation.

	\section{Conclusion}
	\label{sec:conclusion}
	
	In this paper we performed a formal, proof-theoretic analysis of QCDCL. In particular, we focused on the relation of QCDCL and \qres, showing both the incomparability of practically-used QCDCL to \qres as well as the equivalence of a new QCDCL version to \qres.
	
	In addition to the theoretical contributions of this paper, we believe that our findings will also be interesting for practitioners. Firstly, because we have shown the first rigorous dedicated hardness results for QCDCL, not only in terms of formula families with  at most one instance per input size (as is typical in proof complexity), but also in terms of a large family of random QBFs. 
	
	Secondly, we believe that it would be interesting to test the potential of our new QCDCL variants for practical solving. Though we have formulated these as proof systems, it should be fairly straightforward to incorporate our new policies  into actual QCDCL implementations. In particular, the insight that  decisions do not need to follow the  order of quantification in the prefix should be a welcome discovery. Of course, when just using the policy $\ao$, it is not clear that asserting clauses can always be learnt. Therefore, we suggest that for practical implementations, the most interesting new systems should be $\asonr$ and $\asroar$. Both facilitate liberal decision policies, not necessarily following the prefix order, while still allowing to learn asserting clauses. Since both systems are incomparable, it is a priori not clear which one to prefer in practice. However, we would suggest that $\asroar$ should be the more interesting system, since it uses the same unit propagation as QCDCL, but provides an exponential strengthening of QCDCL (as shown in Theorem~\ref{TheoremASROARdstrongerThanLOAR}) via the decision policy $\asro$.
	
	We close with some open questions that are triggered by the results presented here:
	\begin{itemize}
		\item Can we find an alternative formula instead of $\mathtt{Trapdoor}_n$ for the separation between \qres and \qcdcl (easy for \qres, hard for \qcdcl)? I.e., we are primarily interested in formulas whose hardness does not depend on propositional resolution.
		\item Can we find a separation between $\asroar$ and \ldqres?
		\item Can we even find a separation between $\aoar$ and \ldqres, or are the systems possibly even equivalent?
	\end{itemize}

	\bibliography{compl,QCDCLsternsternBIBTEX}
	\bibliographystyle{alphaurl}

\end{document}